\documentclass[11pt,reqno]{amsart}
\pdfoutput=1
\usepackage{amsthm,amsfonts,amssymb,amsmath,oldgerm,upgreek,xcolor}
\numberwithin{equation}{section}
\usepackage{fullpage}
\usepackage{tikz-cd}
\usepackage{amsmath}
\usepackage{amsfonts}
\usepackage{amssymb}
\usepackage{setspace}
\usepackage{stmaryrd}
\usepackage{mathrsfs}
\usepackage{fancyhdr}
\usepackage{graphicx}
\usepackage[colorlinks = true, linkcolor=blue,
filecolor=darkgreen,
urlcolor=red,
citecolor=red,pdftex]{hyperref}
\usepackage{psfrag}
\usepackage{listings} 
\usepackage{makecell}
\usepackage{bbold}
\usepackage[top=30mm,bottom=30mm,left=25mm,right=25mm,a4paper]{geometry}
\usepackage{enumitem}
\usepackage{pifont}
\usepackage{enumitem}

\def\eps{\varepsilon }

\newcommand\br{\begin{remark}}
\newcommand\er{\end{remark}}
\newcommand\bp{\begin{pmatrix}}
\newcommand\ep{\end{pmatrix}}
\newcommand{\be}{\begin{equation}}
\newcommand{\ee}{\end{equation}}
\newcommand{\ba}[1]{\begin{array}{#1}}
\newcommand{\ea}{\end{array}}

\newcommand{\beg}{\begin{example}}
\newcommand{\eeg}{\end{exaplem}}
\newcommand{\bpr}{\begin{proposition}}
\newcommand{\epr}{\end{proposition}}
\newcommand{\bt}{\begin{theorem}}
\newcommand{\et}{\end{theorem}}
\newcommand{\bc}{\begin{corollary}}
\newcommand{\ec}{\end{corollary}}
\newcommand{\bl}{\begin{lemma}}
\newcommand{\el}{\end{lemma}}
\newcommand{\bd}{\begin{definition}}
\newcommand{\ed}{\end{definition}}
\newcommand{\brs}{\begin{remarks}}
\newcommand{\ers}{\end{remarks}}

\newcommand{\zstroke}{%
  \text{\ooalign{\hidewidth -\kern-.3em-\hidewidth\cr$z$\cr}}%
}

\newtheorem{theorem}{Theorem}[section]
\newtheorem{proposition}[theorem]{Proposition}
\newtheorem{corollary}[theorem]{Corollary}
\newtheorem{lemma}[theorem]{Lemma}

\theoremstyle{remark}
\newtheorem{remark}[theorem]{Remark}
\theoremstyle{definition}
\newtheorem{definition}[theorem]{Definition}

\newtheorem{example}[theorem]{Example}

\def\eps {\varepsilon}

\newcommand\R{\mathbf R}
\newcommand\C{\mathbf C}
\newcommand{\N}{\mathbf N}

\newcommand\cA{{\mathcal A}}

\newcommand\cC{{\mathcal C}}

\newcommand\cF{{\mathcal F}}
\newcommand\cG{{\mathcal G}}
\newcommand\cH{{\mathcal H}}

\newcommand\cK{{\mathcal K}}
\newcommand\cL{{\mathcal L}}
\newcommand\cM{{\mathcal M}}
\newcommand\cN{{\mathcal N}}

\newcommand\cP{{\mathcal P}}

\newcommand\cW{{\mathcal W}}

\author{Shahnaz Farhat}
\title{
{\large Report}
}
\title{Quantum-classical motion of charged particles
interacting with scalar fields}

\begin{document}

\begin{abstract}
The goal of this article is to investigate  the dynamics of  semi-relativistic or non-relativistic charged particles in interaction with a scalar meson field. Our main contribution is the derivation of  the classical dynamics of a particle-field system as an effective equation of  the quantum microscopic Nelson model, in the classical limit where the value of the Planck constant approaches zero ($\hbar\to 0$). Thus,  we prove the validity of Bohr's correspondence principle, that is to establish the transition from  quantum to classical  dynamics.  We use a Wigner measure approach to study such transition. Then, as a consequence of  this interplay between classical and quantum dynamics, we establish the global well-posedness of  the classical particle-field  interacting system, despite the low regularity of the related vector field, which prevents the use of a fixed point argument. 
\end{abstract}

\date{\today}
\maketitle

\medskip

$ \qquad \ ${\scshape \footnotesize Keywords.} {\footnotesize Particle-field equation; Nelson model; Semi-classical analysis; Probabilistic representations;

$ \qquad \qquad \qquad \qquad \, $  Measure theoritical techniques; Wigner measures;  Density matrices;  Liouville equations.}


\tableofcontents




\section{Introduction}
 Classical and quantum mechanics may initially appear to be fundamentally different, as classical mechanics deals with the trajectories of particles while quantum mechanics focuses on wave functions evolution. Furthermore,  quantum mechanics are successful in describing microscopic objects, whereas macroscopic systems are better described by classical theories like classical mechanics and classical electrodynamics. The point at which quantum and classical physics are in accordance is known as  the correspondence limit, or the classical limit.   The  \emph{Correspondence  Principle} has been introduced to quantum theory in 1920 by Niels Bohr. Such principle emphasizes the importance of understanding the relationship between  the two theories and how they converge in  specific scaling limits. Bohr proposed that as the quantum numbers increase, the system behaves more classically and the predictions of quantum mechanics align with those  of classical mechanics.  In the 
 mathematical physics literature, the Bohr's principle is discussed in different frameworks (quantum mechanics, many-body theory, quantum field theory); and rigorously proved using mainly the Hepp's method \cite{MR332046}.  However, this method applies only to a specific selection of density matrices, namely coherent states. In this article, we explore this principle by studying the convergence from quantum to classical dynamics in a particle-field interaction model as the value of the Planck constant, denoted by $\hbar$, approaches zero (see also \cite{marco, gin}).
 \vskip 1mm
 On the other hand, the interaction between matter and fields has been a subject of great interest in recent decades.  Here, the focus is on exploring the dynamics of charged particles and a scalar meson field  interacting  according to the Yukawa theory.  Recall that the Yukawa theory models the strong nuclear force as an interaction between nucleons (non-relativistic or semi-relativistic particles) and mesons (fields).  It is known that, despite the ability of  classical mechanics  in resolving many physical problems, there are still some phenomena that can not be explained by classical laws alone. Here, the low regularity of the vector field associated to the interacting system makes it difficult to construct global solutions using standard arguments. To overcome this issue, we use the quantum-classical  transition of the Yukawa theory.   Then, by employing transition, it becomes possible to construct global solutions for the classical interacting system.   Another recent study \cite{Z}  has also explored this type of convergence for non-relativistic particles interacting with the electro-magnetic field, considering the Pauli-Fierz Hamiltonian which converges to the Newton-Maxwell equation. 
\vskip 1mm
From a classical standpoint, in our case the dynamics are governed by a particle-field equation \eqref{top4particlefieldequation}, also known as  Newton-Klein-Gordon equation, which is a nonlinear system of coupled PDE-ODEs.   Previous studies have examined this type of equation, as demonstrated in the articles \cite{effective,soliton,long}. These works focus on analyzing the long-term behavior of the  solutions to a particle-field equation. Specifically, the authors of these articles introduce a form factor within the interacting system to ensure that  the Hamiltonian remains bounded from below, and they assume that this form factor is compactly supported. In our investigation, we adopt a more general framework by imposing less restrictions on this form factor.
\vskip 1mm
From a quantum standpoint, the time  evolution is generated by the so-called Nelson Hamiltonian \eqref{top4nelsonmodelform}. The Nelson model was first introduced by Edward Nelson in \cite{edn,ed} to describe the interaction between  particles (nucleons) and meson field (strong nuclear force). The Nelson model has been widely studied by many researchers, and a selection of relevant articles includes \cite{Abdesselam_2012,alb,amm,arai,1999CMaPh.207..249B,MR3385343,MR3246036,fro,gera,G_rard_2011,hiro,hoe,teu}. 
\vskip 1mm
Our objectives are then:
\begin{itemize}
\item [$-$] Proving the validity of Bohr's correspondence principle. More precisely, we want to establish a 
relationship between quantum and classical dynamics by showing that Nelson model reduces to the classical particle-field equation in the classical limit $\hbar \rightarrow 0$;
\vskip 1mm
\item [$-$] Establishing the global well-posedness of a particle-field equation under weak assumptions on the form factor $\chi$ (see  \eqref{top4A0}) and on the potential $V$ (see \eqref{top4A1}).
\end{itemize}
The quantum dynamics have a well-defined global behavior. Our  method involves transferring certain quantum regularization effects to the classical dynamics. This leads to the  derivation of  the classical dynamics of the particle-field system as an effective equation of a quantum  microscopic dynamical system in the limit $\hbar \rightarrow 0$. 
\vskip 1mm
To achieve this scenario, we investigate the transition using Wigner measures approach  in infinite dimensional bosonic quantum field theory. In recent years, this Wigner measure method have been used in many-body theory \cite{ammari2008ahp} and in quantum field theory \cite{ammari2017sima} with an a priori knowledge of global well-posedness (GWP) for effective equations. Whereas in this work, our   strategy  furnishes global well-posedness  and convergence at the same time. Usually, this convergence  is non-trivial, and there is no prior guarantee of obtaining unique limits. However, we overcome this difficulty by relying on our assumptions. 
\vskip 1mm
The main results are the classical limit (Theorem \ref{top4theorem1}) and the global well-posedness of a particle field equation (Theorem \ref{top4theorem2}). To prove these outcomes, our strategy is summarized in the steps below:
\begin{itemize}
\item We first extract the quantum dynamical system using the family of density matrices $(\varrho_\hbar)_\hbar$   satisfying \eqref{top4S0} and \eqref{top4S1};
\vskip 1mm
\item Then, after proving the propagation (uniformly in any compact time interval) of the two uniform estimates  \eqref{top4S0} and \eqref{top4S1}, we take  the limit to obtain the classical dynamics on the inverse Fourier transform of the Wigner measure. This results in a specific  classical  equation which is equivalent to a statistical Liouville equation, thanks to the regularities associated with the Wigner measure and vector field in this context;
\vskip 1mm
\item  We employ then  measure-theoretical methods \cite{MR3721874,alc2020,Z}  which provides us with the almost sure existence of global solutions. This requires us to prove the uniqueness of the solutions to a particle-field equation by using classical tools;
\vskip 1mm
\item  Finally, we extend the existence result to all initial data. It is important to note, however, that the associated flow is Borel measurable with respect to initial data  and may not be  continuous.
\end{itemize}

\subsection{General framework}
This section provides a concrete mathematical description of the previous introduction. From a classical perspective, the dynamics are governed by a particle-field equation, as detailed in Paragraph \ref{top4parttt}. From a quantum perspective, the dynamics are governed by the Nelson Hamiltonian, which is explained in Paragraph \ref{top4NE}.

\subsubsection{The particle-field equation}\label{top4parttt} Consider $n$ fixed number of  classical particles in the configuration space $\R^d$ with $d \in \N^*$, interacting with  field. 
Let $M_j$ be the mass of the $\text{j}^{\it th}$ particle. The dynamic of the particles is characterized by their momenta $p_j\in \R^d$ and their positions $q_j\in \R^d$. Whereas, the field is described by $\alpha:\R^d\rightarrow \C$. 
Let $p=(p_1,\cdots,p_n)$, $q=(q_1,\cdots,q_n)$ and $f_j \, : \,\R^{d} \longrightarrow \R$, the Hamiltonian of the particle-field system is 
\[\begin{aligned}H(p,q,\alpha)&=\sum_{j=1}^{n} f_j({p}_j)+V(q_1,\cdots, q_n) +\int_{\R^d} \overline{\alpha(k)} \ \omega(k) \  \alpha(k)\, dk\\&
+\sum_{j=1}^{n} \int_{\R^d} \frac{\chi(k)}{\sqrt{\omega(k)}} \bigl[\alpha(k) e^{2\pi i k\cdot q_j} +\overline{\alpha(k)} e^{-2\pi i k\cdot q_j}  \big]\, dk.
 \end{aligned}\]
 We consider two cases:
\begin{itemize}
\item [$-$] Choosing $f_j({p}_j)= \sqrt{{p}_j^2+M_j^2}$ corresponds to the semi-relativistic case.
\item [$-$] Choosing $f_j({p}_j)={{p}_j^2}/{2M_j}$ corresponds to the non-relativistic case.
\end{itemize}
 The parameter $\omega$ represents the dispersion relation defined by $\omega(k)=\sqrt{k^2+{m_f}^2}\geq {m_f}>0$, where ${m_f}$ is the mass of the meson field. 
 The function $V:\R^{dn}\rightarrow \R$ represents the external potential  and $\chi:\R^d \rightarrow \R$ is the form factor. The equation of motion for the particle-field system is given by
 \begin{equation}\label{top4particlefieldequation}
  \begin{cases}
 \begin{aligned}
 &\partial_t p_j=-\frac{\partial H}{\partial q_j}=-\nabla_{q_j}V(q) -\int_{\R^d}2\pi i k \frac{\chi(k)}{\sqrt{\omega(k)}}\bigl[\alpha(k) e^{2\pi i k\cdot q_j} -\overline{\alpha(k)} e^{-2\pi i k\cdot q_j}  \big]\, dk ;\\&
 \partial_t q_j= \  \ \frac{\partial H}{\partial p_j}=\nabla f_j({p}_j);
 \\ & 
 i \partial_t \alpha = \frac{\partial H}{\partial{\overline{\alpha}}}=\omega(k) \ \alpha(k)+ \sum_{j=1}^{n}  \frac{\chi(k)}{\sqrt{\omega(k)}}e^{-2\pi i k\cdot q_j}.
 \end{aligned}
 \end{cases}
 \end{equation}

 To clarify, the interaction term between particles and the scalar field has a specific form which is: linear in the field (for both semi and non-relativistic case); and in the momentum (only in the non-relativistic case).  
 The solution $u=(p,q,\alpha)$ to \eqref{top4particlefieldequation} belongs to  the following classical space 
\[ X^\sigma:=\R^{dn} \times \R^{dn} \times \mathcal{G}^{\sigma},\]
where $ \mathcal{G}^\sigma$ with $\sigma\geq 0$ is  the weighted $L^2$ lebesgue space endowed with the following norm 
\[\begin{aligned} 
\Vert \alpha \Vert^2_{\mathcal{G}^{\sigma}}:=\langle \alpha,\omega(\cdot)^{2\sigma} \ \alpha \rangle_{L^2}= \int_{\R^d} \omega(k)^{2\sigma} \vert \alpha(k) \vert^2 \, dk=\Vert \omega^\sigma \ \alpha \Vert_{L^2}^2.
 \end{aligned}\]
 We have then for $u=(p,q,\alpha) \in X^\sigma$ the following norm 
 \[\begin{aligned}
 \Vert u \Vert^2_{X^{\sigma}}:= \sum_{j=1}^{n} (|q_j|^2+|p_j|^2)+ \Vert \alpha \Vert^2_{\mathcal{G}^{\sigma}}.
 \end{aligned} \]
 The form factor serves as a way to term  the interaction between particles and the field, by smoothing out the Hamiltonian and ensuring that it is bounded from below under certain assumptions. The magnitude of the coupling between the particles and the field is controlled by the form factor. We consider the energy space where the Hamiltonian is well-defined, namely $X^{1/2}$, but our main results are stated in the spaces $X^{\sigma}$ with $\sigma \in [\frac{1}{2},1]$.

\subsubsection{The Nelson model}\label{top4NE}
The particle-field equation can be formally quantized to obtain the Nelson model. 
The Hilbert space of the quantized particle-field system is 
\[\cH:= L^2(\R^{dn}_x,\C)\otimes \Gamma_s(L^2(\R^d_k,\C)) \, , \]
where $\Gamma_s (L^2(\R^d_k,\C))$ is the symmetric Fock space  which could be  identified with 
 \[\begin{aligned}\Gamma_s (L^2(\R^d_k,\C)):= \bigoplus_{m=0}^{+\infty} L^2(\R^d,\C)^{\bigotimes_s m}\simeq \bigoplus_{m=0}^{+\infty}L^2_s(\R^{dm},\C).\end{aligned} \] We denote by $\cF^m:=L^2_s(\R^{dm},\C)$ the symmteric $L^2$ space over $\R^{dm}$. Let 
 \be \label{top4notation}
 \begin{aligned}
 & X_n=(x_1,\cdots, x_n),\quad dX_n=dx_1\cdots dx_n, 
 \\ & K_m=(k_1,\cdots,k_m), \quad dK_m=dk_1 \cdots dk_m.
 \end{aligned}\ee
  Then, the Hilbert space $\cH$ is endowed with the following norm for all $\psi=\{\psi^m\}_{m\geq 0}$
\[\Vert\psi \Vert_\cH:= \Bigl[\sum_{m=0}^{\infty} \int_{\R^{dn}} \int_{\R^{dm}} \vert \psi^m(X_n,K_m) \vert^2 \  dX_n \  dK_m \Bigr]^{1/2} .\]
Let  $\hat p_j$ and $\hat q_j$ be the quantized  momentum  and   position operators such that for all $j\in \lbrace 1,\cdots,n \rbrace$
\[ \hat p_j=-i \hbar   \nabla_{x_j},\qquad \hat q_j=x_j.\]
The $\hbar$ scaled creation-annihilation operators for the field are defined on $\Gamma_s$ for any $f \in L^2(\R^d,\C)$ as
\[  \hat a_\hbar(f)= \int_{\R^d} \overline{f(k)} \ \hat a_\hbar(k) \, dk ,\qquad \hat a_\hbar^*(f)= \int_{\R^d} {f}(k) \ \hat a^*_\hbar(k) \, dk\, , \]
where $\hat a_\hbar(k)$ and $\hat a_\hbar^*(k)$ are the creation-annihilation operator-valued distributions  defined as follows
\[ \begin{aligned}
&[\hat a_\hbar(k)  \ \psi]^m(k_1,\cdots,k_m)=\sqrt{\hbar (m+1)} \,  \psi^{m+1}(k,k_1,\cdots,k_m)\, ; \\&
[\hat a^*_\hbar(k)  \  \psi]^m(k_1,\cdots,k_m)=\frac{\sqrt{\hbar }}{\sqrt{m}}\, \sum_{j=1}^{m}\delta(k-k_j) \ \psi^{m-1}(k_1,\cdots,\hat k_j,\cdots,k_m).
\end{aligned}\]
In our case, we will work with the  generalized $\hbar$ scaled creation-annihilation operators. 
The two operators $\hat a^\sharp_\hbar(G):\cH\rightarrow \cH$ are defined for 
\[ \begin{array}{rcl}
G \, : \,L^2(\R^{dn}_x,\C)& \longrightarrow &  L^{2}(\R^{dn}_x,\C) \otimes L^2(\R^d_{k},\C)\\
\psi & \longmapsto &G \ \psi.
\end{array}\]
  with \[(G \ \psi)(X_n,k)= \sum_{j=1}^n {\frac{\chi(k)}{\sqrt{\omega(k)}} \ e^{-2\pi i k\cdot \hat{q}_j}} \  \psi(X_n).\]
  In general, we have
\begin{align}
&[\hat a_\hbar(G) \  \psi(X_n)]^m(K_m)=\sqrt{\hbar(m+1)}  \sum_{j=1}^{n} \int_{\R^d} \frac{\chi(k)}{\sqrt{\omega(k)}} \  e^{2\pi i k\cdot \hat{q}_j}  \ \psi^{m+1} (X_n; K_m,k) \, dk \, ; \label{top4anni}
\\&
[\hat a^*_\hbar(G) \  \psi(X_n)]^m(K_m)= \frac{\sqrt{\hbar}}{\sqrt{m}} \sum_{j=1}^{m} \sum_{\ell=1}^{n}\frac{\chi(k_j)}{\sqrt{\omega(k_j)}} \  e^{-2\pi i k_j \cdot \hat{q}_\ell}  \ \psi^{m-1} (X_n;k_1,\cdots, \hat k_j,\cdots,k_m) .\label{top4creat} 
\end{align} 
Introduce  the second quantization $d\Gamma(A):\cH \rightarrow \cH$ for the self-adjoint operator $A$ with  
$d\Gamma(A) \  \psi = \lbrace [d\Gamma(A) \ \psi]^m \rbrace_{m>0}$
and where 
\[[d\Gamma(A) \  \psi]^m= \hbar \sum_{j=1}^{m} \psi \otimes \cdots \otimes\underbrace{ A\psi}_{{j}^{\it th} \, position} \otimes \cdots \otimes \psi .\]
The $\hbar$ scaled  number operator   $\hat N_\hbar=d\Gamma({\rm Id})$  and the number operator $\hat{N}$ are defined as follows  
\[ [\hat N_\hbar \  \psi]^m= \hbar \  m  \  \psi^m,\qquad [\hat N  \ \psi]^m=  m \   \psi^m.\]
 The  free field Hamiltonian $ d\Gamma(\omega):\cH\rightarrow \cH$ is  defined as follows
\[[d\Gamma(\omega)  \ \psi(X_n)]^m= \hbar \sum_{j=1}^{m} \omega(k_j) \  \psi^{m}(X_n;K_m). \]
Formally, one can express this as:
\[ d\Gamma(\omega) = \int_{\R^d} \hat a_\hbar^*(k) \ \omega(k) \ \hat a_\hbar(k) \ dk.\]
The  non-interacting Hamiltonian is defined as follows
 \[\hat{H}_0:=\hat{H}_{01}+\hat{H}_{02},\]
 where we have introduced the two terms   $\hat H_{01} $ and $ \hat{H}_{02}$ as follows
 \[\hat{H}_{01}=\sum_{j=1}^nf_j(\hat{p}_j),\qquad \hat{H}_{02}= d\Gamma(\omega).\]
 The interaction Hamiltonian $\hat H_1:\cH\rightarrow \cH$ is defined  in terms of $\hat a_\hbar,\ \hat a^*_\hbar$ as in  \eqref{top4anni}-\eqref{top4creat}  as follows
\[\hat{H}_1= \hat a_{\hbar}(G)+\hat a^*_\hbar(G).\]
The Nelson-Hamiltonian takes then the following form
\be \label{top4nelsonmodelform} \hat H_\hbar \equiv\hat H=\hat H_0+V(\hat{q})+\hat{H}_1.\ee
The inclusion of a form factor $\chi$ in the interaction term of the particle-field equation ensures the well-definedness of the corresponding quantum dynamics and leads to a self-adjoint Nelson Hamiltonian. It has been demonstrated that, under certain mild assumptions on $\chi$ and the potential $V$, the unbounded operator $\hat H_\hbar$ is indeed self-adjoint (as discussed in \cite{ammari2017sima} and references therein). In the following, we aim to identify the minimal conditions on $\chi$ and $V$ that enable further analysis.
\subsection{Assumptions and main results}
We have to impose the following assumptions on the external potential $V: \R^{dn} \longmapsto \R $ and the form factor $\chi:\R^d \longmapsto \R$ with $\sigma \geq 0$:
\begin{align}
&  V\in \cC^2_b(\R^{dn};\R), \label{top4A0}
\\ & \omega(\cdot)^{\frac{3}{2}-\sigma} \chi(\cdot)  \in L^2({\R^d;\R}). \label{top4A1}
\end{align}
Note that the following identities hold true:
\begin{itemize}
\item  If $  \omega(\cdot)^{\frac{3}{2}-\sigma} \chi(\cdot)   \in L^2({\R^d;\R})$ then  $ \chi(\cdot)  \in L^2({\R^d;\R})$;
\item If $  \chi(\cdot)   \in L^2({\R^d;\R})$ then  for any $\gamma>0$, we have   $\omega(\cdot)^{-\gamma} \chi(\cdot)  \in L^2({\R^d;\R})$.
\end{itemize}
Let $(\varrho_{\hbar})_{\hbar \in (0,1)}$ be a family of density matrices on $\cH$ of the particle-field quantum system. The main assumptions  on the family of states $(\varrho_{\hbar})_{\hbar \in (0,1)}$ are:
\begin{align}
& \exists  C_0>0, \ \forall\hbar \in (0,1), \quad    {\rm Tr}[\varrho_\hbar \ d\Gamma(\omega^{2\sigma})]\leq C_0, \label{top4S0}
\\ &  \exists  C_1>0, \ \forall\hbar \in (0,1), \quad    {\rm Tr}[\varrho_\hbar \ (\hat{q}^2+\hat{p}^2)]\leq C_1. \label{top4S1}
\end{align}
Remark that the following identities hold true:
\begin{itemize}
\item If $ {\rm Tr}[\varrho_\hbar \ d\Gamma(\omega^{2\sigma})]\leq c_0$, then ${\rm Tr}[\varrho_\hbar \ d\Gamma(\omega)]\leq c_0^\prime $ for some $c_0,c_0^\prime\in \R^*_+$;
\item  If ${\rm Tr}[\varrho_\hbar \ d\Gamma(\omega)]\leq c_1 $, then $ {\rm Tr}[\varrho_\hbar \ \hat{N}_\hbar]\leq c_1^\prime$ for some $c_1,c_1^\prime\in \R^*_+$;
\item  If $ {\rm Tr}[\varrho_\hbar \ (\hat{q}^2+\hat{p}^2)]\leq c_2$, then  ${\rm Tr}[\varrho_\hbar \ (\hat{H}_0+1)]\leq c_2^\prime$ for some $c_2,c_2^\prime\in \R^*_+$.
\end{itemize}
The first result presented in this section concerns the flow of the particle-field equation. 
\begin{theorem}[Global well-posedness of the particle-field equation]\label{top4theorem1}
Let $\sigma \in [\frac{1}{2},1]$. Assume \eqref{top4A0} and \eqref{top4A1} hold. Then for any initial condition $u_0 \in X^{\sigma}$ there exists a unique global strong solution $u(\cdot)\in \cC(\R,X^{\sigma})\cap \cC^1(\R,X^{\sigma-1}) $ of the particle-field equation \eqref{top4particlefieldequation}.  Moreover, the global flow  map $u_0\rightarrow \Phi_t(u_0)=u(t)$ associated to the particle-field equation \eqref{top4particlefieldequation} is Borel measurable.
\end{theorem}
\noindent The above global flow is not constructed from  a fixed point argument, whereas it is constructed by means of statistical arguments. More precisely, we use measure theoritical techniques to construct this flow. And thus, it is only Borel measurble and not necessarily continuous.


\noindent Denote  by $\cP(X^0)$ the set of all Borel probability measure over the space $X^0$.
\begin{definition}[Wigner measures]\label{top4definitionwigner}
A Borel probability measure $\mu \in \cP(X^0)$ is a Wigner measure of a family of density matrices $(\varrho_{\hbar})_{\hbar\in (0,1)}$ on the Hilbert space $\cH$ if and only if there exists a countable subset $\cA \subset  (0,1)$ with $0\in \overline{\cA}$ such that for any $\xi=(p_0,q_0,\alpha_0) \in X^0:$ 
\[\lim_{\hbar \rightarrow 0,\hbar \in \cA} {\rm Tr}\Big[ \cW(2 {\pi q_0},-2\pi p_0, \sqrt{2} \pi \alpha_0) \varrho_h\Big]= \int_{X^0} e^{2\pi i\Re e \langle \xi,u\rangle_{X^0} }  \ d\mu(u).\]
\end{definition}
\noindent The next result concerns the classical limit which relies on the  construction of  a Wigner measure in the context of infinite-dimensional bosonic quantum field theory. This allows us to establish convergence from the quantum to the classical dynamics. 
 Denote by \[\cM(\varrho_\hbar,\hbar\in \cA),\]
the set of all Wigner measure associated to the density matrices $(\varrho_\hbar)_{\hbar\in \cA}$.
\begin{theorem}[Validity of Bohr's correspondence principle]\label{top4theorem2}
Let $\sigma \in [\frac{1}{2},1]$ and assume \eqref{top4A0} and \eqref{top4A1} hold true. Let $(\varrho_{\hbar})_{\hbar \in (0,1)}$ be a family of density matrices on $\cH$ satisfying \eqref{top4S0} and \eqref{top4S1}. Let $(\hbar_n)_{n\in \N}\subset (0,1)$ such that $\underset{n\rightarrow \infty}{\hbar_n \longrightarrow 0 } $ and assume that $\cM(\varrho_{\hbar_n}, \, n\in \N)=\{\mu_0\}$. Then for all times $t\in \R$, there exists a subsequence $(\hbar_\ell)_{\ell\in \N}$ and a family of Borel probability measure $(\mu_t)_{t\in \R}$ such that 
\[\cM(e^{-i\frac{t}{\hbar_\ell }\hat{H}}\varrho_{\hbar_\ell} \  e^{i\frac{t}{\hbar_\ell}\hat{H}}, \, \ell \in \N)=\{\mu_t\},\]  
where $\mu_t\in \cP(X^0)$ satisfying
\begin{itemize}
\item [(i)] $ \mu_t$ is concentrated on $X^\sigma$ i.e. $\mu_t(X^\sigma)=1$;
\item [(ii)] $\mu_t=(\Phi_t)_{\sharp}\mu_0$, where $u_0\longmapsto \Phi_t(u_0)=u(t)$ is the  Borel measurable global flow of the particle-field equation \eqref{top4particlefieldequation}.
\end{itemize}
\end{theorem}

\noindent The result above indicates that when $\varrho_{\hbar}$ are density matrices  on $\mathcal{H}$ that approach the Wigner probability measure $\mu_0$ as $\hbar$ approaches zero, the evolved density matrices $\varrho_{\hbar}(t)$ will converge to $\mu_t=(\Phi_t)\,_{\sharp}\, \mu_0$ for all times $t$. Here, $\Phi_t$ is the flow that solves \eqref{top4particlefieldequation}. 

\medskip
\noindent To demonstrate the aforementioned results, we adopt the following approach: Firstly, we employ classical techniques to establish the uniqueness property of the particle-field solutions. Subsequently, we establish crucial uniform propagation estimates on the quantum dynamics. Then,  we present a probabilistic representation of measure-valued solutions for the Liouville's equation (see \cite{MR3721874,alc2020}). This representation is used to construct a generalized global flow for a particle-field equation. As a conclusion,  we establish by means of Wigner measures the global well-posedness for the particle-field equation and the Bohr's correspondence principle for the Nelson model.

\section{The classical system}
This section is dedicated to examining various classical properties of the particle-field equation. Firstly, in Subsection \ref{top4interactionrepresent}, we introduce the particle-field equation as a semi-linear partial differential equation and establish its interaction representation. In Subsection \ref{top4uniquenessproperty}, we prove the uniqueness of the particle-field equation using this representation. 
\subsection{The interaction representation}\label{top4interactionrepresent}
 The particle-field equation  \eqref{top4particlefieldequation} takes the following form 
\be \label{top4condpartfieldeqn}
\begin{cases}
\displaystyle \frac{d u(t)}{dt}={\rm w}(u(t))=\cL(u(t))+\cN(u(t)),\\
u(0)=u_0 \in X^{\sigma},
\end{cases}\tag{PFE}
\ee
where $t \rightarrow u(t)=(p(t),q(t),\alpha(t))$ is a solution, $\cL(u)=(0,0,-i\omega \alpha)$ is a linear operator such that $\cL:X^{\sigma} \longrightarrow X^{\sigma-1}$ and $\cN$ is the nonlinearity given by  
\be \label{top4defF}
\begin{aligned}
& (\cN(u))_{p_j}:= -\nabla_{q_j} V(q)-\nabla_{q_j}I_j(q,\alpha),
\\ &  (\cN(u))_{q_j}:= \nabla f_j(p_j),
\\ & (\cN(u))_{\alpha}(k):= -i\sum_{j=1}^{n} \frac{\chi(k)}{\sqrt{\omega(k)}} \ e^{-2\pi i k \cdot q_j},
\end{aligned}
\ee
where we have introduced $ I_j:\R^{dn} \times L^2(\R^d,\C)\rightarrow \R$
\be \label{top4defI} I_j(q,\alpha):= \int_{\R^d} \frac{\chi(k)}{\sqrt{\omega(k)}} \bigl[\alpha(k) e^{2\pi i k\cdot q_j} +\overline{\alpha(k)} e^{-2\pi i k\cdot q_j}  \big]\, dk
\ee
 with 
 \be \label{top4defnablaI} \nabla_{q_j} I_j(q,\alpha)= \int_{\R^d} 2 \pi i k  \ \frac{\chi(k)}{\sqrt{\omega(k)}}  \bigl[\alpha(k) e^{2\pi i k\cdot q_j} -\overline{\alpha(k)} e^{-2\pi i k\cdot q_j}  \big]\, dk.
 \ee
 
\noindent We consider now the particle-field equation as a non-autonomous initial value problem over the  Hilbert space $X^{\sigma}$ with 
\begin{equation}\label{top4ivp}
\begin{cases}
\begin{aligned}
& \frac{du(t)}{dt} =v(t,u(t)),
\\ & u(0)=u_0 \in X^{\sigma}.
\end{aligned}
\end{cases}
\tag{IVP}
\end{equation}
The non-autonomous vector field $ v$ is defined  in terms of the non-linearity $\cN: X^{\sigma} \longrightarrow X^{\sigma}$ of the particle-field equation as well as the free field  flow $\Phi_t^{\it f}:X^{\sigma} \longrightarrow X^{\sigma}$ as follows:
\be \label{top4vectorfield} v(t,u)=\Phi_{-t}^{f} \circ \cN \circ \Phi_t^{\it f} (u),\ee
where we have introduced the  free field flow $\Phi^{\it f}_t$  as follows
\be \label{top4freeflow} \Phi_{t}^{\it f}(p,q,\alpha)=(p,q,e^{-it\omega(k)}\alpha). \ee

\begin{lemma}[Explicit expression for the vector field $v$]
The vector field $v: \R \times X^\sigma \rightarrow X^\sigma$ takes the following  explicit form:
\be \label{top4vectorfieldv}
\begin{aligned}
& (v(t,u))_{p_j}=\big(\cN\circ \Phi^{\it f}_t(u)\big)_{p_j} ,
\\&(v(t,u))_{q_j}= \big(\cN\circ \Phi^{\it f}_t(u)\big)_{q_j} ,
\\&(v(t,u))_{\alpha}(k)=e^{it\omega(k)}  \ \big(\cN\circ \Phi^{\it f}_t(u)\big)_{\alpha}(k) ,
\end{aligned}
\ee
where $v(t,u)={}^t\big((v(t,u))_{p_1},\cdots,(v(t,u))_{p_n},(v(t,u))_{q_1},\cdots, (v(t,u))_{q_n},v(t,u)_\alpha\big)$.
\end{lemma}

\begin{proof}
The result follows from direct  computations of $v$ using the relation \eqref{top4vectorfield}.
\end{proof}

\begin{proposition}[Equivalence between \eqref{top4condpartfieldeqn} and \eqref{top4ivp}]\label{top4equivalence}
Assume \eqref{top4A0} and  \eqref{top4A1} are satisfied.
Let $I$ be a bounded open interval containing the origin. Then, the statments below are equivalent:
\begin{enumerate}
\item $u(\cdot) \in \cC^{1}(I,X^{\sigma})$ is a strong solution of \eqref{top4ivp};
\item $u(\cdot)\in \cC(I,X^\sigma)$ solves the following Duhamel formula
\[u(t)=u_0+\int_{0}^tv(s,u(u)) \ ds, \quad \forall t\in I.\]
\item The curve $t\longmapsto   \Phi_{t}^{\it f}(u(t))\in \cC(I,X^{\sigma}) \cap \cC^{1}(I,X^{\sigma-1})$  is a strong solution to the particle-field equation \eqref{top4particlefieldequation}.
\end{enumerate}
\end{proposition}

\begin{proof}
The first two assertions can be proved easily since $v$ is continuous vector field (by Lemma \ref{top4continuityofv}). Let us now prove the equivalence between (1) and (3).
Suppose that $u(t)=(p(t),q(t),\alpha(t))$ is a solution to \eqref{top4ivp}. Require to prove that \[\tilde{u}(t)=(\tilde p(t),\tilde q(t),\tilde \alpha(t))=\Phi_t^{\it f}(u(t))=(p(t),q(t), e^{-it\omega(k)}\alpha(t)),\] is a solution to \eqref{top4condpartfieldeqn}. 
The first term 
\[
\begin{aligned}
 \partial_t \tilde{ p}_j&= \partial_t { p}_j= (v(t,u))_{p_j}=\big(\cN\circ \Phi^{\it f}_t(u)\big)_{p_j} =({\rm w}(\tilde{u}))_{p_j}.
 \end{aligned}
 \]
 The second term 
 \[
\begin{aligned}
 \partial_t \tilde{ q}_j= \partial_t  q_j=\nabla  f_j(p_j)=({\rm w}(\tilde{u}))_{q_j}.
 \end{aligned} 
 \]
 The third term 
 \[
\begin{aligned}
 \partial_t \tilde \alpha&=-i\omega(k) \ e^{-it\omega(k)}  \  \alpha(k)+  e^{-it\omega(k)} \partial_t  \alpha
 \\ & =-i\omega(k) \ \tilde  \alpha(k)+e^{-it\omega(k)}  \ e^{it\omega(k)}  \ \big(\cN\circ \Phi^{\it f}_t(u)\big)_{\alpha}(k) 
 \\ & =-i\omega(k) \ \tilde  \alpha(k)+(\cN(\tilde{u}))_{\alpha}(k)
 \\ &=({\rm w}(\tilde u))_{\alpha}(k).
 \end{aligned} 
 \]
 We conclude that $\tilde u$ is a solution to \eqref{top4condpartfieldeqn}. Similarly, we can prove the reverse sense.
\end{proof}

\noindent Let $I$ be an open interval containing the  origin. We are interested in strong solution to the particle-field equation \eqref{top4condpartfieldeqn} such that 
\[ u(\cdot) \in \cC(I,X^{\sigma}) \cap \cC^1( I,X^{\sigma -1 }),\]
 and \eqref{top4condpartfieldeqn} is satisfied for all $t\in I$. In particular, from the second assertion of Proposition \ref{top4equivalence},  these solutions satisfy the following Duhamel formula for all $t\in I$
 \be \label{top4duhamelparti} u(t)=\Phi^{\it f}_t(u(0))+\int_{0}^{t} \Phi_{t-s}^{\it f}\circ\cN(u(s)) \ ds,\ee
 where $\Phi_t^{\it f}(\cdot)$ is the free field flow  defined above in \eqref{top4freeflow}.
 
\subsection{Properties of the particle-field equation}\label{top4uniquenessproperty}
In this section, we establish various properties related to the particle-field equation and its time interaction representation. Of most significance is the recovery of the uniqueness property of solutions to the particle-field equation \eqref{top4condpartfieldeqn} on the energy space $X^{\sigma}$. Our approach starts with deriving estimates for $\nabla_{q_j} I_j(\cdot)$.

\begin{lemma}[Estimates for $\nabla_{q_j} I_j $]\label{top4lemma2.1}
We have the following two estimates. 
\begin{itemize}
\item [(i)] Assume $\omega^{{1}/{2}} \chi \in L^2(\R^d,dk)$. Then, for all $(q,\alpha) \in \R^{dn}\times L^2(\R^d,\C)$, for all $j \in \lbrace 1,\cdots, n\rbrace$, we  have the following estimate
 \begin{align}
  \vert \nabla_{q_j} I_j(q,\alpha) \vert \leq 4 \pi  \  \Vert \omega^{1/2} \chi \Vert_{L^2}  \ \Vert \alpha \Vert_{L^2}\label{top4deriv}.
 \end{align}
 \vskip 2mm
\item [(ii)] Assume \eqref{top4A1} is satisfied. Then, for all $j \in \lbrace 1,\cdots, n\rbrace$, for all $q_1, \ q_2 \in \R^{dn}$ with $q_{1}=(q_{1j})_{j}$ and $q_{2}=(q_{2j})_{j} $, for all 
  $ \alpha_1, \ \alpha_2 \in \cG^{\sigma}$, we have  
\[\begin{aligned}
& \vert \nabla_{q_{j}} I_j(q_1,\alpha_1)-\nabla_{q_{j}} I_j(q_2,\alpha_2)  \vert \\ &\leq 4  \pi  \  \Vert \omega^{1/2} \chi \Vert_{L^2}  \ \Vert \alpha_1-\alpha_2 \Vert_{L^2}+8\sqrt{2} \pi^2 \ \Vert \omega^{\frac{3}{2}-\sigma} \chi \Vert_{L^2} \ \vert q_{1j}- q_{2j} \vert \  \Vert \alpha_2 \Vert_{\cG^{\sigma}}.
\end{aligned}\]
\end{itemize}
\end{lemma}
\begin{proof}
For (i), by Cauchy-Schwatrz inequality, we have $ \forall (q,\alpha) \in \R^{dn}\times L^2(\R^d,\C)$
 \[ 
 \begin{aligned}
   \vert \nabla_{q_j}  I_j(q,\alpha) \vert &  = \Big \vert \int_{\R^d} \frac{\chi(k)}{\sqrt{\omega(k)}} 2 \pi i k   \bigl[\alpha(k) e^{2\pi i k\cdot q_j} -\overline{\alpha(k)} e^{-2\pi i k\cdot q_j}  \big]\, dk \Big \vert 
   \\ & \leq  4 \pi   \ \int_{\R^d} \Big \vert \frac{\chi(k)}{\sqrt{\omega(k)}}   \ \omega(k) \Big \vert  \  \vert \alpha(k) \vert \, dk
   \\ & \underset{c-s}{\leq}  4 \pi   \ \Vert \omega^{1/2} \chi \Vert_{L^2}  \ \Vert \alpha \Vert_{L^2}.
 \end{aligned}
 \]
 \vskip 1mm
 \noindent  For (ii), by Cauchy-Schwartz inequality and using the estimate $\vert e^{iy}-1 \vert \leq \sqrt{2}  \ \vert y \vert $, we have $ \forall (q_1,\alpha_1), (q_2,\alpha_2) \in \R^{dn}\times \cG^{\sigma}$, $\forall j \in  \lbrace 1,\cdots, n \rbrace$ the following estimates
 \[ 
 \begin{aligned}
   &\vert \nabla_{q_{j}} I_j(q_1,\alpha_1)-\nabla_{q_{j}} I_j(q_2,\alpha_2) \vert 
   \\&  = \Big \vert \int_{\R^d} \frac{\chi(k)}{\sqrt{\omega(k)}} 2\pi i k\bigl[\alpha_1(k) e^{2\pi i k\cdot q_{1j}} -\overline{\alpha_1(k)} e^{-2\pi i k\cdot q_{1j}} -\alpha_2(k) e^{2\pi i k\cdot q_{2j}} +\overline{\alpha_2(k)} e^{-2\pi i k\cdot q_{2j}} \big]\, dk \Big \vert 
   \\ & = \Big \vert \int_{\R^d}\frac{\chi(k)}{\sqrt{\omega(k)}}2\pi i k \bigl[ (\alpha_1(k )-\alpha_2(k )) e^{2\pi i k\cdot q_{1j}} + \alpha_2(k) \big( e^{2\pi i k\cdot q_{1j}}-e^{2\pi i k\cdot q_{2j}} \big)
   \\ & \qquad \qquad \qquad \qquad + (\overline{\alpha_2(k )}-\overline{\alpha_1(k )}) e^{-2\pi i k\cdot q_{1j}} + \overline{\alpha_2(k)} \big( e^{-2\pi i k\cdot q_{2j}}-e^{-2\pi i k\cdot q_{1j}} \big)\bigr]
    \\ & \leq 4 \pi  \int_{\R^d} \Big \vert \frac{\chi(k)}{\sqrt{\omega(k)}}  \omega(k) \Big \vert \Big[ \vert {\alpha}_1(k)-{\alpha}_2(k) \vert+ \big \vert \alpha_2(k)  \  \big(e^{2\pi i k\cdot(q_{1j}- q_{2j})}  -1\big)\big \vert  \Big] 
    \\ & \leq 4 \pi  \int_{\R^d} \Big \vert \sqrt{\omega(k)} \ \chi(k)\Big \vert \Big[ \vert {\alpha}_1(k)-{\alpha}_2(k) \vert+ \big \vert \alpha_2(k)  \sqrt{2} \ 2\pi  k\cdot ( q_{1j}- q_{2j})  \big \vert  \Big] 
    \\ & \leq 4 \pi  \int_{\R^d} \Big \vert \sqrt{\omega(k)} \ \chi(k)\Big \vert  \Big[ \vert {\alpha}_1(k)-{\alpha}_2(k) \vert+2 \sqrt{2} \ \pi \big \vert \alpha_2(k)\vert    \  \omega(k)  \  \vert q_{1j}- q_{2j}\vert    \Big] 
   \\ & \underset{c-s}{\leq}  4  \pi  \  \Vert \omega^{1/2} \chi \Vert_{L^2}  \ \Vert \alpha_1-\alpha_2 \Vert_{L^2}+8\sqrt{2} \pi^2 \ \Vert \omega^{\frac{3}{2}-\sigma} \chi \Vert_{L^2} \ \vert q_{1j}- q_{2j} \vert \  \Vert \alpha_2 \Vert_{\cG^{\sigma}} .
 \end{aligned}
 \]
\end{proof}

 \noindent The vector field $\cN$, which characterizes the nonlinearity of the particle-field equation, possesses the following properties.
 
 \begin{proposition}[Continuity and boundedness of $\cN$]\label{top4contbddF} Assume \eqref{top4A0} and \eqref{top4A1}  are satisfied. Then, the nonlinearity $\cN: X^{\sigma} \rightarrow X^{\sigma}$ is a continuous, and bounded on bounded sets, vector field.
 \end{proposition}
 
 \begin{proof}
 
Let us prove first that $\cN:X^{\sigma} \rightarrow X^{\sigma}$ is bounded on  bounded sets. Let $u \in X^{\sigma}$ be a bounded such that $\Vert u \Vert_{X^{\sigma}}\leq c_0$, for some $c_0>0$. Require to prove $\Vert \cN(u) \Vert_{X^{\sigma}}^2\leq c_1$ for some $c_1>0$. We have first  with some  $c_2>0$
\[ 
\begin{aligned}
\big \vert\big(\cN(u)\big)_{p_j} \big \vert&= \Big \vert-\nabla_{q_j} V(q)-\nabla_{q_j}I_j(q,\alpha)\Big \vert
\\ &\leq \big \vert \nabla_{q_j} V(q) \big \vert +\big \vert \nabla_{q_j}I_j(q,\alpha) \big \vert
\\ & \underset{Lemma \ \ref{top4lemma2.1}-(i)}{\leq}  \big \Vert \nabla_{q_j} V \big \Vert_{L^\infty}+4 \pi  \  \Vert \omega^{1/2} \chi \Vert_{L^2}  \ \Vert \alpha \Vert_{L^2}
\\ & \leq  \sup_{j=1}^n \big \Vert \nabla_{q_j} V \big \Vert_{L^\infty}+\frac{4 \pi}{{m_f}^\sigma}  \  \Vert \omega^{1/2} \chi \Vert_{L^2}  \ \Vert \alpha \Vert_{\cG^{\sigma}}
\\ & \leq  \sup_{j=1}^n\big \Vert \nabla_{q_j} V \big \Vert_{L^\infty}+\frac{4 \pi}{{m_f}^\sigma}  \  \Vert \omega^{1/2} \chi \Vert_{L^2}  \ c_0:=c_2.
\end{aligned}
\]
We also have with some $c_3, \ c_4>0$
\[ 
\begin{aligned}
\big \vert\big(\cN(u)\big)_{q_j} \big \vert= \Big \vert \nabla  f_j(p_j)\Big \vert\leq c_3 \ \vert p_j \vert\leq c_3 \ c_0:=c_4 . 
\end{aligned}
\]
Finally, we have  with some $c_5>0$
\[ 
\begin{aligned}
\Vert \big( \cN(u)\big)_{\alpha} \Vert_{\cG^{\sigma}}^2&=\Vert \omega^{\sigma} \ \big( \cN(u)\big)_{\alpha} \Vert_{L^2}^2
\\ & =\int_{\R^d} \Big \vert -i\sum_{j=1}^{n} {\omega^\sigma(k)} \  \frac{\chi(k)}{\sqrt{\omega(k)}}  \ e^{-2\pi i k \cdot q_j} \Big \vert^2 \ dk
\\ &\lesssim n^2 \ \Vert \omega^{\sigma-\frac{1}{2}}\chi \Vert_{L^2}^2:=c_5, 
\end{aligned}\]
 where $\Vert \omega^{\sigma-\frac{1}{2}}\chi \Vert_{L^2}$ is finite  since $ \sigma- \frac{1}{2}<\frac{3}{2}-\sigma$ for $\sigma\in [\frac{1}{2},1]$. This implies that there exists $c_1>0$ such that
\[\Vert \cN(u) \Vert_{X^{\sigma}}^2\leq n \ (c_2^2+c_4^2)+ c_5:=c_1.  \]
It remains to prove the continuity of the nonlinear term  $\cN:X^{\sigma} \rightarrow X^{\sigma}$. Suppose that  
\[u_\ell=(p_\ell,q_\ell,\alpha_\ell) \underset {\ell \rightarrow +\infty}{\longrightarrow  } u=(p,q,\alpha),\quad \text{ in } X^{\sigma} \text{  i.e. } \Vert u_\ell-u \Vert_{X^{\sigma}} \underset{\ell\rightarrow +\infty}{\longrightarrow 0 }.\]
 Require to prove 
 \[ \cN(u_\ell) \underset{\ell\rightarrow +\infty}{\longrightarrow  } \cN(u) \text{ in } X^{\sigma} \text{ i.e. } \Vert \cN(u_\ell)-\cN(u) \Vert_{X^{\sigma}} \underset{\ell\rightarrow +\infty}{\longrightarrow 0 }.\]
Indeed, we have 
\[
\Vert \cN(u_\ell)-\cN(u) \Vert_{X^{\sigma}}^2= \sum_{j=1}^n \left[\Big \vert\big( \cN(u_\ell)-\cN(u) \big)_{p_j} \Big \vert^2+\Big \vert\big( \cN(u_\ell)-\cN(u) \big)_{q_j} \Big \vert^2\right] +\Vert\big( \cN(u_\ell)-\cN(u) \big)_{\alpha}\Vert_{\cG^{\sigma}}^2.\]
By Lemma \ref{top4lemma2.1}-(ii),  we can assert that 
\begin{align*}
&
\begin{aligned}
&\Big \vert\big( \cN(u_\ell)-\cN(u) \big)_{p_j} \Big \vert 
\\ & \leq  \Big \vert \nabla_{q_j}V(q_\ell)-\nabla_{q_j}V(q)\Big \vert+ \frac{4  \pi}{{m_f}^\sigma}  \  \Vert \omega^{1/2} \chi \Vert_{L^2}  \ \Vert \alpha_\ell-\alpha \Vert_{\cG^{\sigma}}+8\sqrt{2} \pi^2 \ \Vert \omega^{\frac{3}{2}-\sigma} \chi \Vert_{L^2} \ \vert q_{\ell j}- q_j \vert \  \Vert \alpha \Vert_{\cG^{\sigma}}
\\ & \leq  \sum_{j^\prime=1}^n  \Vert \nabla_{q_{j^\prime}}\nabla_{q_j}V \Vert_{L^{\infty}} \ \vert q_{ \ell j^\prime}- q_{j^\prime} \vert+ \frac{4  \pi}{{m_f}^\sigma}  \  \Vert \omega^{1/2} \chi \Vert_{L^2}  \ \Vert \alpha_\ell-\alpha \Vert_{\cG^{\sigma}}+8\sqrt{2} \pi^2 \ \Vert \omega^{\frac{3}{2}-\sigma} \chi \Vert_{L^2} \ \vert q_{ \ell j}- q_j \vert \  \Vert \alpha \Vert_{\cG^{\sigma}}
\\ & \leq  \Big[ n \sup_{j^\prime,j=1}^n \Vert \nabla_{q_{j^\prime}}\nabla_{q_j}V \Vert_{L^{\infty}} \ + \frac{4  \pi}{{m_f}^\sigma}  \  \Vert \omega^{1/2} \chi \Vert_{L^2}  +8\sqrt{2} \pi^2 \ \Vert \omega^{\frac{3}{2}-\sigma} \chi \Vert_{L^2} \  \Vert \alpha \Vert_{\cG^{\sigma}}  \Big ] \ \Vert u_\ell-u \Vert_{X^{\sigma}}\underset{ \ell \rightarrow +\infty}{\longrightarrow 0 }.
\end{aligned}
\\
\\ & \begin{aligned}
 \Big \vert\big( \cN(u_\ell)-\cN(u) \big)_{q_j} \Big \vert 
 \leq c \ \vert p_{\ell j}-{p_j}\vert \leq c \ \Vert u_\ell -u\Vert_{X^\sigma} \underset{\ell\rightarrow +\infty}{\longrightarrow 0 }  .
\end{aligned}
\end{align*}
For the last term we have 
\[
\begin{aligned}
 \Vert \big( \cN(u_\ell)-\cN(u) \big)_{\alpha} \Vert_{\cG^\sigma}^2&= \int_{\R^d} \Big \vert -i\sum_{j=1}^{n} {\omega^\sigma(k)} \  \frac{\chi(k)}{\sqrt{\omega(k)}}  \ \big[e^{-2\pi i k \cdot q_{\ell j}}-e^{-2\pi i k \cdot q_j} \big] \Big\vert^2 \ dk
\\ &=\int_{\R^d} \Big \vert -i\sum_{j=1}^{n} {\omega^{\sigma-\frac{1}{2}}(k)} \  {\chi(k)}  \ \big[e^{-2\pi i k \cdot q_{\ell j}}-e^{-2\pi i k \cdot q_j} \big] \Big\vert^2 \ dk,
\end{aligned}
\]
where we have 
\begin{itemize}
\item $\displaystyle \int_{\R^d} \Big \vert -i\sum_{j=1}^{n} {\omega^{\sigma-\frac{1}{2}}(k)} \  {\chi(k)}  \ \big[e^{-2\pi i k \cdot q_{\ell j}}-e^{-2\pi i k \cdot q_j} \big] \Big\vert^2 \ dk \leq  4 n^2  \ \Vert {\omega^{\sigma-\frac{1}{2}}} \  {\chi} \Vert_{L^2}^2<+\infty $;
\vskip 2mm
\item  $ \big \vert e^{-2\pi i k \cdot q_{\ell j}}-e^{-2\pi i k \cdot q_j} \big \vert \underset{\ell \rightarrow +\infty}{\longrightarrow 0 }$.
\end{itemize}
Hence by Lebesgue dominated convergence theorem, we get 
\[ \Vert \big( \cN(u_\ell)-\cN(u) \big)_{\alpha} \Vert_{\cG^\sigma}^2 \underset{\ell\rightarrow +\infty}{\longrightarrow 0 }.\]
And thus, 
\[
\Vert \cN(u_\ell)-\cN(u) \Vert_{X^{\sigma}}^2
 \underset{\ell\rightarrow +\infty}{\longrightarrow 0}. \]
 \end{proof}
 \noindent  The above theorem implies the following results on the vector field $v$.

\begin{lemma}[Continuity and boundedness properties of the vector field $v$] \label{top4continuityofv}Assume \eqref{top4A0} and  \eqref{top4A1} are satisfied. Then, the vector field $v:\R \times X^{\sigma} \longrightarrow X^{\sigma}$ is continuous and bounded on bounded subsets of $\R \times X^{\sigma}$. 
\end{lemma}

\begin{proof}
This is a consequence of the continuity and boundedness properties of the nonlinear term $\cN:X^{\sigma} \rightarrow X^{\sigma}$ in Proposition \ref{top4contbddF}.
\end{proof}

\noindent  As a consequence of the above properties, we have the following uniqueness property.
 
 \begin{proposition}[Uniqueness property]Assume \eqref{top4A0} and \eqref{top4A1} are satisfied. Let $I$ be  an open interval  containing the origin and let $u_1,u_2\in \cC(I,X^{\sigma})$ be two strong solutions of the particle-field equation \eqref{top4condpartfieldeqn} such that 
 $u_1(0)=u_2(0).$
 Then $u_1(t)=u_2(t)$ for all $t\in I$.
  \end{proposition}
 
 \begin{proof}
 Note first that using Duhamel formula \eqref{top4duhamelparti} as well as $u_1(0)=u_2(0)$, we have for all $t \geq 0$
 \[\Vert u_1(t)-u_2(t)\Vert_{X^0} \leq \int_{0}^{t} \Vert \cN(u_1(s))-\cN(u_2(s))\Vert_{X^0} \  ds. \]
We claim that  for all $s\in [0,t]$, there exists $C>0$ such that
\[  \Vert \cN(u_1(s))-\cN(u_2(s))\Vert_{X^0}  \leq C \  \Vert u_1(s)-u_2(s)\Vert_{X^0} .\] 
Indeed, by using the Mean Value Theorem for multivariate vector-valued function $\nabla_{q_j}V$, the first component yields to 
\[ 
\begin{aligned}
&\big \vert\big(\cN(u_1(s))-\cN(u_2(s))\big)_{p_j} \big \vert
\\ & \leq  \Big \vert\nabla_{q_j} V(q_1(s))-\nabla_{q_j} V(q_2(s)) \Big \vert +\vert \nabla_{q_{j}} I_j(q_1(s),\alpha_1(s))-\nabla_{q_{j}} I_j(q_2(s),\alpha_2(s)) \vert
\\ &\underset{Lemma  \ \ref{top4lemma2.1}-(ii)}{\lesssim}  \Vert \nabla_q \nabla_{q_j}V \Vert_{L^\infty} \  \vert q_1(s)-q_2(s) \vert + \Vert \alpha_1(s)-\alpha_2(s)\Vert_{L^2} +\Vert \alpha_2(s)\Vert_{\cG^{\sigma}}  \ \vert q_1(s)-q_2(s) \vert
\\ & \lesssim \Vert u_1(s)-u_2(s)\Vert_{X^0},
\end{aligned}
\] 
where we have used  for some bounded interval $J\subset I$
\[\Vert \alpha_2(s)\Vert_{\cG^{\sigma}} \leq \sup_{s\in J}  \ \Vert u_2(s) \Vert_{X^{\sigma}} <+\infty.\]
The second component yields to 
\[\begin{aligned}
\big \vert\big(\cN(u_1(s))-\cN(u_2(s))\big)_{q_j} \big \vert \lesssim \vert p_1(s)-p_2(s) \vert 
 \lesssim \Vert u_1(s)-u_2(s)\Vert_{X^0}.
\end{aligned}
\]
The third component yields to 
\[\begin{aligned}
\Vert \big(\cN(u_1(s))-\cN(u_2(s)) \big)_{\alpha} \Vert_{L^2} \lesssim \vert q_1(s)-q_2(s) \vert 
\lesssim \Vert u_1(s)-u_2(s)\Vert_{X^0}.
\end{aligned}  \]
Therefore, we get by combining the above three components the following estimate 
\[\begin{aligned} \Vert u_1(t)-u_2(t)\Vert_{X^0}& \leq \int_{0}^{t}\Vert \cN(u_1(s))-\cN(u_2(s))\Vert_{X^0} \ ds  
\\ &\leq C \ \int_{0}^{t}  \Vert u_1(s)-u_2(s)\Vert_{X^0}\  ds.
\end{aligned} 
\]
Then, by Gronwall's Lemma, we get $\Vert u_1(t)-u_2(t)\Vert_{X^0}=0$. Thus, we get   the desired result.
 \end{proof}

 \begin{remark}[Local well-posendness in $X^{1/2}$] Under assumptions \eqref{top4A0} and \eqref{top4A1},  for all initial data $u_0\in X^{1/2}$, one can prove the existence of  a unique local solution $u(\cdot) \in \cC([0,T];X^{1/2})$ to \eqref{top4particlefieldequation}, where $T\in \R^+_*$ is small enough. This  can be proved by means of standard fixed point argument. 
 \end{remark}

 \begin{remark}[Global well-posendness in $X^{\sigma}$] In line with the above remark, one could also prove the global well-posedness in $X^\sigma$ by using Granwall arguements  and the conservation  of Hamiltonian.
 \end{remark}

\section{The quantum system}
In Subsection \ref{top4uniformestimates}, we prove some quantum estimates, which we then apply in Subsection \ref{top4selfadjoint} to establish the self-adjointness of the Nelson Hamiltonian using the Kato-Rellich theorem. Lastly, in Subsection \ref{top4Duhamelformula}, we discuss the dynamical equation for the quantum system.

\subsection{Quantum estimates}\label{top4uniformestimates}

\noindent Our initial focus here is on providing the reader with estimates that are necessary to establish the self-adjointness of the Nelson Hamiltonian. Denote by $\cL(\cH)$ the set of all bounded operator and by  $\cL^1(\cH)$  the set of trace-class operators on $\cH$.
\begin{lemma}[Creation-Annihilation estimates]\label{top4lemma1} Let \[F \in \cL\big( L^2(\R^{dn},dX_n),L^2(\R^{dn},dX_n)\otimes L^2(\R^{d},dk)\big).\]
\begin{itemize}
\item [(i)] For every $\psi \in D(\hat{N}_\hbar^{1/2})$, we have 
\begin{align}
&\Vert \hat a_\hbar(F)\psi \Vert_\cH\leq \Vert \hat N^{1/2}_\hbar\psi \Vert_\cH \ \Vert F\Vert_{\cL(L^2,L^2\otimes L^2)}\, \label{top4numberestimate1};
\\& \Vert \hat a^*_\hbar(F)\psi \Vert_\cH\leq \Vert (\hat N_\hbar+1)^{1/2}\psi \Vert_\cH \ \Vert F\Vert_{\cL(L^2,L^2\otimes L^2)}\, \label{top4numberestimate2}.
\end{align}
\vskip 2mm
\item [(ii)] For all $\psi \in D((\hat H_{02})^{1/2})$, we have
\begin{align}
& \Vert \hat a_\hbar(F)\psi \Vert_\cH\leq \ \Big \Vert \frac{F}{\sqrt{\omega}}\Big\Vert_{\cL}\, \Vert (\hat H_{02}+1)^{1/2}\psi \Vert_\cH \label{top4fieldestimate1} ;
\\&\Vert \hat a^*_\hbar(F)\psi \Vert_\cH^2 \leq \ \Big \Vert \frac{F}{\sqrt{\omega}}\Big\Vert_{\cL}^2\, \Vert (\hat H_{02}+1)^{1/2}\psi \Vert_\cH^2  + \hbar \ \Vert F \Vert_{\cL}^2 \ \Vert \psi \Vert^2.\label{top4fieldestimate2}
\end{align}
\end{itemize}
\end{lemma}
\begin{proof}
Let $K_m$ and  $X_n$  as indicated in \eqref{top4notation}.
For (i)-\eqref{top4numberestimate1},   we have with $\Vert \cdot \Vert_\cL\equiv \Vert \cdot \Vert_{\cL(L^2,L^2\otimes L^2)}$
\[\begin{aligned} 
 \Vert \hat a_\hbar(F)\psi \Vert_\cH^2&=\sum_{m\geq 0} \int_{\R^{dn}} \int_{\R^{dm}} \Big\vert \big[ \hat a_\hbar(F)\psi \big]^m (X_n,K_m)\Big\vert^2 \  dX_n \ dK_m 
\\ & =\sum_{m\geq 0} \int_{\R^{dn}}  \int_{\R^{dm}} \Big\vert \int_{\R^d}\sqrt{\hbar(m+1)}\  \overline{F(k)} \ \psi^{m+1}(X_n,K_m,k) dk\Big\vert^2 \ dX_n \ dK_m
\\ & \underset{c-s}{\leq} \Vert F \Vert_{\cL}^2 \   \sum_{m\geq 0} \int_{\R^{dn}} \int_{\R^{dm}} \Big[ \int_{\R^d} \Big\vert \sqrt{\hbar(m+1)}\  \ \psi^{m+1}(X_n,K_m,k)\Big\vert^2 dk \Big] dX_n \ dK_m
\\& \leq\Vert F \Vert_{\cL}^2 \sum_{m\geq 0} \int_{\R^{dn}}  \int_{\R^{d(m+1)}}   {\hbar(m+1)}\  \ \Big\vert\psi^{m+1}(X_n,K_{m+1})\Big\vert^2   dX_n \ dK_{m+1}
\\& \leq \Vert F \Vert_{\cL}^2 \ \Vert \hat N^{1/2}_\hbar \psi \Vert_\cH^2.
\end{aligned}
\]
For  (i)-\eqref{top4numberestimate2}, we have 
\[\begin{aligned} 
 \Vert \hat a^*_\hbar(F)\psi \Vert_\cH^2&=\langle \psi , \ \hat{a}_\hbar(F)\hat{a}^*_\hbar(F) \psi \rangle
 \\ & = \langle \psi, \ \big[ [\hat{a}_\hbar(F),\hat{a}^*_\hbar(F)]+\hat{a}^*_\hbar(F)\hat{a}_\hbar(F)  \big] \psi\rangle
 \\ & \leq \hbar \ \Vert F \Vert_{\cL}^2 \ \Vert \psi\Vert^2+ \Vert \hat{a}_\hbar(F) \psi \Vert^2
 \\ & \lesssim \Vert F \Vert_{\cL}^2 \  \Vert (\hat{N}_\hbar+1)^{1/2} \ \psi \Vert^2.
\end{aligned}
\]
For (ii)-\eqref{top4fieldestimate1},  we have  
\[\begin{aligned}
\Vert  \hat a_\hbar(F) \ \psi  \Vert^2_\cH&= \sum_{m\geq 0} \int_{\R^{dn}} \int_{\R^{dm}} \Big\vert \big[ \hat a_\hbar(F) \ \psi \big]^m (X_n,K_m)\Big\vert^2 dK_m  \ dX_n
\\ & =\sum_{m\geq 0} \int_{\R^{dn}} \int_{\R^{dm}} \Big\vert \int_{\R^d}\sqrt{\hbar(m+1)}\  \overline{F(k)} \ \psi^{m+1}(X_n,K_m,k) dk\Big\vert^2 dK_m  \ dX_n
\\ & =\sum_{m\geq 0} \int_{\R^{dn}} \int_{\R^{dm}} \Big[ \int_{\R^d} \Big\vert \sqrt{\hbar(m+1)}\  \frac{\overline{F(k)}}{\sqrt{\omega(k)}} \ \sqrt{\omega(k)} \ \psi^{m+1}(X_n,K_m,k) \Big\vert dk \Big]^2 dK_m  \ dX_n
\\ & \underset{c-s}{\leq} \Big\Vert \frac{F}{\sqrt{\omega}} \Big \Vert_{\cL}^2 \   \sum_{m\geq 0} \int_{\R^{dm}} \Big[ \int_{\R^d} \Big\vert \sqrt{\hbar(m+1)}\  \sqrt{\omega(k)} \ \psi^{m+1}(X_n,K_m,k)\Big\vert^2 dk \Big]  dk \ dK_m \ dX_n
\\ & {\leq} \Big\Vert \frac{F}{\sqrt{\omega}} \Big \Vert_{\cL}^2 \   \sum_{m\geq 0} \int_{\R^{dm}}  \int_{\R^d}  {\hbar(m+1)}\  {\omega(k)} \ \Big\vert \psi^{m+1}(X_n,K_m,k)\Big\vert^2 dk \   dK_m \ dX_n
\\& \leq \Big\Vert \frac{F}{\sqrt{\omega}} \Big \Vert_{\cL}^2 \ \Vert (\hat H_{02})^{1/2} \  \psi \Vert_\cH^2.
\end{aligned}\]
Similar discussion as for  (i)-\eqref{top4numberestimate2} works perfectly to prove  (ii)-\eqref{top4fieldestimate2}.
\end{proof}

\begin{lemma}[Field and number estimates]\label{top4lemma2}
For all $\psi \in D(\hat H_{02})$, we have 
\begin{enumerate}
\item  $\Vert \hat H_{02}  \ \psi \Vert_\cH \geq m_f \ \Vert \hat N_\hbar  \ \psi \Vert_\cH ;$
\vskip 2mm
\item $\Vert (\hat N_\hbar+1)^{1/2}\psi \Vert_\cH \leq  \eps \  \Vert \hat N_\hbar \psi \Vert_\cH  +b (\eps) \ \Vert \psi \Vert_\cH \quad \text{for some $\eps<1$ and $b(\eps)<+\infty$.}$
\end{enumerate}
\end{lemma}

\begin{proof}
For (1), we have 
\[\begin{aligned} 
 \Vert \hat H_{02}  \  \psi \Vert_\cH^2&=\sum_{m\geq 0} \int_{\R^{dn}}\int_{\R^{dm}} \Big\vert \big[ \hat H_{02} \  \psi \big]^m (X_n,k_1,\cdots,k_m)\Big\vert^2 \ dX_n \  dk_1 \cdots dk_m
\\ & =\sum_{m\geq 0}\int_{\R^{dn}} \int_{\R^{dm}} \Big\vert  \hbar \ \sum_{l=1}^{m} \underbrace{w(k_l)}_{\geq \  m_f}\psi^{m}(X_n,k_1,\cdots,k_m) \Big\vert^2 dX_n \  dk_1 \cdots dk_m
\\ & \geq \sum_{m\geq 0} \int_{\R^{dn}}\int_{\R^{dm}} \Big\vert  \hbar \  m \ m_f \ \psi^{m}(X_n, k_1,\cdots,k_m) \Big\vert^2 \ dX_n \  dk_1 \cdots dk_m
\\& \geq {m_f}^2 \sum_{m\geq 0}\int_{\R^{dn}} \int_{\R^{dm}} \hbar^2 \  m^2 \  \Big\vert \ \psi^{m}(X_n , k_1,\cdots,k_m) \Big\vert^2   dX_n \ dk_1 \cdots dk_m
\\& \geq {m_f}^2 \ \Vert \hat N_\hbar \psi \Vert^2_\cH.
\end{aligned}
\]
For (2), we have
\[ \begin{aligned}
\ \Vert (\hat N_\hbar +1)^{1/2} \psi \Vert^2_\cH& =\langle \psi ,  (\hat N_\hbar +1)\psi\rangle = \langle \psi ,  \hat N_\hbar\psi \rangle + \Vert \psi \Vert_\cH^2 \\ & \leq \Vert \psi \Vert_\cH \ \Vert \hat N_\hbar \psi \Vert_\cH + \Vert \psi \Vert_\cH^2\\ & \leq \frac{\tilde \eps}{2}  \ \Vert \hat N_\hbar \psi \Vert_\cH^2+(1+\frac{1}{2 \tilde \eps}) \ \Vert \psi \Vert_\cH^2.
\end{aligned}\]
This implies that, by choosing appropriate $\tilde{\eps}$,  there exists $\eps<1$ and $b(\eps) <+\infty$ such that 
\[ \begin{aligned}
\ \Vert (\hat N_\hbar +1)^{1/2} \psi \Vert_\cH \leq \eps  \ \Vert \hat N_\hbar \psi \Vert_\cH+b(\eps)  \ \Vert \psi \Vert_\cH.
\end{aligned}\]
\end{proof}

Below, we give an important inequality between $\hat{H}$ and $\hat{H}_0$ which is useful for the coming discussions. Suppose $\hat{H}+a>0$ and $\hat{H}_0+b>0$ with some $a,b \in \R$.
\begin{lemma}[Equivalence between $\hat H$ and $\hat H_0$]\label{top4inequality}
Assume \eqref{top4A0} and  $\omega^{-\frac{1}{2}}\chi\in L^2(\R^d,dk)$ are satisfied. Then there exists $c,C>0$ independent of $\hbar $ such that for all $\hbar \in (0,1)$ and for all $\psi \in D(\hat{H}_0)$
\be \label{top4ineq}
c \ \langle \psi , \ (\hat{H}+a) \ \psi  \rangle \leq \langle \psi , \ (\hat{H}_0+b) \ \psi \rangle \leq C \  \langle \psi , \ (\hat{H}+a) \ \psi \rangle.
\ee
In particular, for all $\psi \in D((\hat{H}_0)^{1/2})$, we have 
\be 
c \  \Vert(\hat{H}+a)^{1/2}\psi  \Vert \leq \Vert(\hat{H}_0+b)^{1/2}\psi  \Vert \leq C \  \Vert(\hat{H}+a)^{1/2}\psi  \Vert.
\ee
\end{lemma}
\begin{proof}
Note that first we have the following estimates on $V$ below
\[- \Vert V\Vert_{L^\infty} \leq V \leq \Vert V\Vert_{L^\infty}\Rightarrow  \Vert V\Vert_{L^\infty} + V \geq 0. \] 
We have 
 \[\begin{aligned} \langle \psi ,  \ (\hat{H}+a) \ \psi  \rangle& = \langle \psi ,  \ \hat{H}_0  \ \psi  \rangle +\langle \psi ,  V \ \psi  \rangle +\langle \psi ,  \  \hat{H}_1 \ \psi  \rangle  + \langle \psi ,  a \ \psi  \rangle
 \\ & \leq \langle \psi ,  \ \hat{H}_0  \ \psi  \rangle + \Vert V \Vert_{L^\infty}\langle \psi ,   \ \psi  \rangle +\Vert  \psi\Vert   \  \Vert  \hat{H}_1 \ \psi \Vert  + \langle \psi ,  a \ \psi  \rangle
 \\ & \leq \langle \psi ,  \ \hat{H}_0  \ \psi  \rangle + \Vert V \Vert_{L^\infty}\langle \psi ,   \ \psi  \rangle + 2\Vert \omega^{-1/2} \ \chi \Vert_{L^2} \  \Vert  \psi\Vert   \  \Vert  (\hat{H}_0+1)^{1/2} \ \psi \Vert   + \langle \psi ,  a \ \psi  \rangle
  \\ & \leq \langle \psi ,  \ \hat{H}_0  \ \psi  \rangle + \Vert V \Vert_{L^\infty}\langle \psi ,   \ \psi  \rangle + \Vert \omega^{-1/2} \ \chi \Vert_{L^2}^2 \  \Vert  \psi\Vert^2   +  \Vert  (\hat{H}_0+1)^{1/2} \ \psi \Vert^2   + \langle \psi ,  a \ \psi  \rangle
  \\ & \leq \langle \psi ,  \ \Big( 2\hat{H}_0+\Vert V \Vert_{L^\infty}+ \Vert \omega^{-1/2} \ \chi \Vert_{L^2}^2+1+a \Big)  \ \psi  \rangle
 \\ &  \leq c \ \langle \psi ,  \ ( \hat{H}_0+b )  \ \psi  \rangle ,
 \end{aligned} \]
 where $c \in \R^+_*$ depends on $ \Vert V \Vert_{L^\infty}, \  \Vert \omega^{-1/2} \ \chi \Vert_{L^2}^2$ and independent on $\hbar$.
 
 \noindent On the reverse side, we have 
 \[\begin{aligned} \langle \psi ,  \ \hat{H}_0 \ \psi  \rangle& = \langle \psi ,  \ \hat H_{01} \ \psi  \rangle +\langle \psi ,  \hat H_{02} \ \psi  \rangle
 \\ & \leq \langle \psi ,  \ (\hat H_{01}+V +\Vert V\Vert_{L^\infty})  \ \psi  \rangle + \langle \psi ,  \hat H_{02} \ \psi  \rangle .
 \end{aligned} \]
Then by using Lemma \ref{top4lemma1}, we can also assert that for $\eps>0$, there exists $c_\eps>0$ (depends on the norm $\Vert \omega^{-1} \ \chi \Vert_{L^2}$) such that 
\[ \vert \langle \psi, \ \hat{H}_1 \   \psi \rangle \vert \leq   \Vert \omega^{-1} \chi \Vert_{L^2}  \ \big[ \frac{1}{\eps} \langle \psi, \psi \rangle + \eps \langle \psi,  \ (\hat H_{02}+1)  \ \psi \rangle \big] . \]
This means as quadratic form there exists a constant $c_\eps\in \R^*_+$ which depends on $\eps$ such that 
\[\begin{aligned}
 \hat H_{02}+ \hat{H}_1  &=(1-\eps) \  \hat H_{02} + \eps \hat H_{02}+\hat{H}_1
 \\ &\geq  (1-\eps) \  \hat H_{02} -c_\eps {1}
 \end{aligned}
 \]
 This implies 
 \[ \hat H_{02}\leq \frac{1}{1-\eps} \ \big[ \hat H_{02}+\hat{H}_1 +c_\eps\big] .\]
 We conclude that there exists $C>0$ such that 
 \[\begin{aligned} \langle \psi ,  \ (\hat{H}_0+b) \ \psi  \rangle& \leq \langle \psi ,  \ (\hat H_{01}+V +\Vert V\Vert_{L^\infty})  \ \psi  \rangle + \langle \psi ,  \frac{1}{1-\eps} \ \big[ \hat H_{02}+\hat{H}_1 +c_\eps\big]  \ \psi  \rangle \\ & \leq C \  \langle \psi ,  \ (\hat{H}+a) \  \psi \rangle.
 \end{aligned} \]
\end{proof}

\subsection{Self-adjointness of Nelson Hamiltonian}\label{top4selfadjoint}
We prove here the self-adjointness of the Nelson Hamiltonian using the estimates provided in the previous section.
\begin{proposition}[Self-adjointness of the Nelson Hamiltonian]
Assume \eqref{top4A0} and $\omega^{-\frac{1}{2}} \chi \in L^2(\R^d,dk)$. Then, the operator $\hat H:\cH\rightarrow \cH$ is self-adjoint operator on $D(\hat H_0)=D(\hat H)$.
\end{proposition}
\begin{proof}
We have first with some $C_{\tilde{\eps}} \in \R^*_+$
\[
\begin{aligned}
\Vert \hat H_1 \psi \Vert_\cH &=\Vert \hat a_\hbar(G)\psi+\hat a_\hbar^*(G)\psi \Vert_\cH
\\&\leq \Vert \hat a_\hbar(G)\psi\Vert_\cH+ \Vert \hat a_\hbar^*(G)\psi \Vert_\cH
\\ & 
\leq 2\Vert G \Vert  \ \Vert (\hat N_\hbar+1)^{1/2} \psi \Vert_\cH 
\\ & 
\leq 2 n \Big \Vert \frac{\chi}{\sqrt{\omega}} \Big \Vert_{L^2}  \ \big[\tilde  \eps \Vert \hat N_\hbar \psi \Vert_\cH  +b(\tilde \eps) \Vert \psi \Vert_\cH  \big]
\\ & \leq\frac{2 n }{m_f} \Big \Vert \frac{\chi}{\sqrt{\omega}} \Big \Vert_{L^2}  \tilde \eps  \ \Vert \hat H_{02} \  \psi \Vert_\cH+C_{\tilde \eps} \ \Vert \psi \Vert_\cH,
 \end{aligned}
\]
where we have used the estimates in Lemmas \ref{top4lemma1} and \ref{top4lemma2}. 
Choose $\tilde \eps$ small enough such that $2 n \Big \Vert {\chi}/{\sqrt{\omega}} \Big \Vert_{L^2} \  \tilde \eps<m_f $. We conclude that there exists $\eps<1$ and $C_\eps\in \R^*_+$ such that 
\[
\begin{aligned}
\Vert \hat H_1 \psi \Vert_\cH \leq \eps \ \Vert \hat H_{02} \psi \Vert_\cH+C_\eps \ \Vert \psi \Vert_\cH. 
\end{aligned}\]
The operator $\hat H_{02}$ is self-adjoint operator and $\hat H_1$ is symmteric operator. Thus by Kato-Rellich theorem, $\hat H_{02}+ \hat H_1$ is self-adjoint on $D(\hat H_{02}).$
Remark also that by \eqref{top4A0}, we have 
\[ \Vert V\psi \Vert_{L^2}\leq \eps \prime \Vert \hat H_{01} \psi \Vert_{L^2} +b \ \Vert  \psi \Vert_{L^2}.\]
Then, again by Kato-Rellich theorem, $\hat H_{01}+V$ is self-adjoint on $D(\hat H_{01} ) \subseteq D(V).$ 
We also have as a consequence of Kato-Rellich theorem that $\hat H_{01}+V\geq -c {1}$. This means $\hat H_{01}+V +c{1} \geq 0 $.
This gives
\[ \begin{aligned} 
\Vert \hat H_{02} \  \psi \Vert_\cH^2&=\langle \psi, \hat H_{02}^2 \  \psi \rangle
\\& \leq \langle \psi, \big( \hat H_{02}+\hat H_{01}+V +c 1 \big)^2 \psi \rangle\\ & \leq \Vert  \big( \hat H_{02}+\hat H_{01}+V +c {1}\big) \psi \Vert_\cH^2.
\end{aligned}\]
We conclude 
\[ \Vert \hat H_1 \psi \Vert_\cH  \leq \eps \ \Vert  \big(\hat H_{02}+\hat H_{01}+V +c {1}\big) \psi \Vert_\cH+C( \eps) \ \Vert \psi \Vert_\cH . \] 
Remark that $\hat H_{01}+V+c {1} $ commutes with $ \hat H_{02}$. This means $\hat H_{01}+V+c{1}+\hat H_{02}$ is self adjoint in $D(\hat H_{01}+V+c {1}+\hat H_{02})=D(\hat H_{01}+\hat H_{02}).$
By Kato-Rellich theorem, $\hat H $ is self adjoint on $D(\hat H_0)=D(\hat H_{01}+\hat H_{02})=D(\hat H_{01})\cap D(\hat H_{02})$.
\end{proof}

\subsection{The dynamical equation}\label{top4Duhamelformula}
The primary objective of this section is to determine the dynamical equation of the quantum system. This equation should converge, as $\hbar$ approaches zero, to a classical dynamical equation that involves the inverse Fourier transform of a specific Wigner measure. To achieve this,  in Paragraph \ref{top4duh}, we derive the Duhamel formula for the quantum system. Then, in Paragraph \ref{top4comm}, we expand the commutator within this Duhamel formula.

\subsubsection{Duhamel formula}\label{top4duh}
We begin  by introducing the Weyl Heisenberg operator, which acts on the entire interacting Hilbert space $\cH=L^2(\R^{dn},\C) \otimes \Gamma_s(\cG^0)$,  as the following map
\begin{equation}\label{chap4.eq.wcal}
\xi=(z,\alpha) \in X^0\equiv \C^{dn} \oplus \cG^0 \longmapsto \cW(\xi)\equiv \cW(z,\alpha):= W_1(z) \otimes W_2(\alpha)
\end{equation}
where we have introduced with $\Im m  \langle z,z^\prime \rangle =q\cdot p^\prime - p \cdot q^\prime, \ \forall (p,q) , \ (p^\prime,q^\prime) \in \R^{dn} \times \R^{dn}$:
\begin{itemize}
\item [-] the Weyl operator  on the particle variable which is defined,  for all $(p,q)  \in \R^{dn} \times \R^{dn}$ and for $z=q+ip \in \C^{dn}$, as follows:
\be \label{top4weylheseinbergtranslation}
 W_1(z)=e^{i \ \Im m \langle \hat q+i \hat p,z \rangle}= e^{i( p\cdot \hat q-q\cdot \hat p)};
\ee
\item [-] the Weyl operator   on the Fock space $\Gamma_s(L^2(\R^d,\C))$ which is defined for any $\alpha \in L^2(\R^d,\C)$ as follows:
\begin{equation}\label{chap4.eq.weyl.fock}W_2(\alpha)=e^{\frac{i}{\sqrt{2}}(\hat a_\hbar(\alpha)+\hat{a}^*_\hbar(\alpha))}.
\end{equation}
\end{itemize}
The above operators satisfy the following commutation relations
\begin{align}
& W_1(z) W_1(z^\prime) = e^{-i \frac{\hbar}{2} \Im  m \langle z,z^\prime \rangle } \ W_1(z+z^\prime), \quad \forall z,z^\prime \in \C^{dn},  \label{top4weylheseinbergtranslation}
\\& W_2(\alpha) W_2(\beta)= e^{-i \frac{\hbar}{2} \Im m  \langle \alpha,\beta \rangle_{L^2} } \ W_2(\alpha+\beta),\quad \forall \alpha,\beta \in \cG^0.
\end{align}

\noindent Below, we mention several crucial estimates that are necessary to establish a Duhamel formula for the evolved states of a quantum system. The prove of the following identities  requires the estimates derived in   Lemma \ref{top4lemma1}, we refer the reader to \cite{Z} for more details on the proof.
\begin{lemma}[Weyl Heisenberg estimates]\label{top4weylheisenberg} There exists a constant $C>0$ such that for any $\hbar\in(0,1)$ 
\begin{itemize}
\item [(i)] for any $\alpha\in L^2(\R^d,\C)$ and  any $\psi \in D(\hat{N}_\hbar)$
\[ \Vert(\hat N_\hbar)^{1/2} \ W_2(\alpha) \ \psi \Vert_{\Gamma_s}\leq C \ \Vert (\hat{N}_\hbar+1)^{1/2}  \ \psi \Vert_{\Gamma_{s}}; \]
\item [(ii)] for any $\alpha\in \cG^{1/2}$ and  any $\psi \in D(\hat{H}_0)$
\[ \Vert(\hat H_0)^{1/2} \ W_2(\alpha) \ \psi \Vert_{\Gamma_s}\leq C \ \Vert (\hat H_0+1)^{1/2}  \ \psi \Vert_{\Gamma_{s}}; \]
\item  [(iii)] for any $z\in \C^{dn}$  and any $\psi \in D((\hat{p}^2+\hat{q}^2)^{1/2})$
\[ \Vert(\hat{p}^2+\hat{q}^2)^{1/2} \ W_1(z) \ \psi \Vert_{L^2(\R^{dn})}\leq C \ \Vert (\hat{p}^2+\hat{q}^2+1)^{1/2}  \ \psi \Vert_{L^2(\R^{dn})}. \]
\end{itemize}
\end{lemma}
\noindent 
The matter here is to understand the propagation of the density matrices $\varrho_\hbar$ on the Hilbert space $\cH$. To this end, we define 
\be \label{top4denistymetrices}
\varrho_\hbar(t)=e^{-i\frac{t}{\hbar} \hat H}  \varrho_\hbar \  e^{i\frac{t}{\hbar} \hat H} \quad \text{and } \quad \tilde \varrho_\hbar(t)=e^{i\frac{t}{\hbar} \hat H_{02}}   \varrho_\hbar(t)  \ e^{-i\frac{t}{\hbar} \hat H_{02}}.
\ee
In order to prove the main results Theorems \ref{top4theorem1} and \ref{top4theorem2}, it is necessary to identify the Wigner measures of the evolved state $\varrho_\hbar(t)$. However, the complexity inherited from the interaction between particles and field makes direct identification unfeasible. Instead, we use the interaction representation $\tilde{\varrho}_\hbar(t)$, which helps us overcome several nonlinearities that could lead to imprecise formulas. Furthermore, recovering the Wigner measures of $\varrho_\hbar(t)$ from those of $\tilde{\varrho}_\hbar(t)$ is not difficult. To this end, we start below derivation of the quantum  dynamical system.
\begin{proposition}Assume that \eqref{top4A0} and $\omega^{{1}/{2}} \chi \in L^2(\R^d,dk)$. Let $(\varrho_\hbar)_{\hbar \in (0,1)}$ be a family of density matrices satisfying \eqref{top4S0} and \eqref{top4S1}. Then for all $\xi \in X^{1/2}$,  for all $\hbar \in (0,1)$  and for all $t,t_0\in \R,$ 
we have 
\be \label{top4duhamelformula}
{\rm Tr}\Big[ \cW(\xi) \tilde \varrho_\hbar(t) \Big]={\rm Tr}\Big[ \cW(\xi) \tilde \varrho_\hbar(t_0) \Big] -\frac{i}{\hbar} \int_{t_0}^{t} {\rm Tr}\Big( \big[ \cW(\xi),\hat H_I(s) \big] \, \tilde \varrho_\hbar(s) \Big) \, ds,
\ee
where 
\be \label{top4H-k}
\hat H_I(s):=e^{i\frac{s}{\hbar} \hat H_{02}} (\hat H-\hat H_{02}) \ e^{-i\frac{s}{\hbar} \hat H_{02}}.
\ee
\end{proposition}
\begin{proof}
By Duhamel's formula, we have 
\[{\rm Tr}\Big[ \cW(\xi) \tilde \varrho_\hbar(t) \Big]={\rm Tr}\Big[ \cW(\xi) \tilde \varrho_\hbar(t_0) \Big] + \int_{t_0}^{t} \frac{d}{ds} {\rm Tr}\Big[ \cW(\xi) \  \tilde \varrho_\hbar(s) \Big] \, ds. \]
We have also 
\[\frac{d}{dt} {\rm Tr}\Big[  \cW(\xi) \  \tilde \varrho_\hbar(t) \Big]=\lim_{s \rightarrow t} \frac{ {\rm Tr}\Big[  \cW(\xi)  \ \big(   \tilde \varrho_\hbar(t)-\tilde \varrho_\hbar(s)  \big)\Big] }{t-s}.\]
Let  $S= (\hat{H}_0+1)^{1/2} $. We start by 
\[\begin{aligned}
&  {\rm Tr}\Big[  \cW(\xi)  \ \big(   \tilde \varrho_\hbar(t)-\tilde \varrho_\hbar(s)  \big)\Big] 
\\&= {\rm Tr}\Big[  \cW(\xi) \Big(e^{i\frac{t}{\hbar} \hat H_{02}}  \ e^{-i\frac{t}{\hbar} \hat H}  \varrho_\hbar \ e^{i\frac{t}{\hbar} \hat H} \ e^{-i\frac{t}{\hbar} \hat H_{02}} -e^{i\frac{s}{\hbar} \hat H_{02}}  \ e^{-i\frac{s}{\hbar} \hat H}  \varrho_\hbar \ e^{i\frac{s}{\hbar} \hat H} \ e^{-i\frac{s}{\hbar} \hat H_{02}}\Big)  \Big]
\\ & = {\rm Tr}\Big[  \cW(\xi) \Big(e^{i\frac{t}{\hbar} \hat H_{02}}  \ e^{-i\frac{t}{\hbar} \hat H} - e^{i\frac{s}{\hbar} \hat H_{02}}  \ e^{-i\frac{s}{\hbar} \hat H}\Big) \varrho_\hbar \ e^{i\frac{t}{\hbar} \hat H} \ e^{-i\frac{t}{\hbar} \hat H_{02}}\Big]
\\& \  +{\rm Tr}\Big[  \cW(\xi) \  e^{i\frac{s}{\hbar} \hat H_{02}}  \ e^{-i\frac{s}{\hbar} \hat H}  \  \tilde \varrho_\hbar \  \Big( e^{i\frac{t}{\hbar} \hat H} \ e^{-i\frac{t}{\hbar} \hat H_{02}}-e^{i\frac{s}{\hbar} \hat H} \ e^{-i\frac{s}{\hbar} \hat H_{02}}  \Big) \Big]
\\ & = {\rm Tr}\Big[ S^{-1} \ \cW(\xi) \  S \ S^{-1} \Big(e^{i\frac{t}{\hbar} \hat H_{02}}  \ e^{-i\frac{t}{\hbar} \hat H} - e^{i\frac{s}{\hbar} \hat H_{02}}  \ e^{-i\frac{s}{\hbar} \hat H}\Big) \varrho_\hbar \ S \ S^{-1} \ e^{i\frac{t}{\hbar} \hat H} \ S \ S^{-1} \ e^{-i\frac{t}{\hbar} \hat H_{02}} \ S \Big]
\\& \  +{\rm Tr}\Big[  \cW(\xi) \  e^{i\frac{s}{\hbar} \hat H_{02}}  \ e^{-i\frac{s}{\hbar} \hat H}  \  \tilde \varrho_\hbar  \  S \ S^{-1}  \Big( e^{i\frac{t}{\hbar} \hat H} \ e^{-i\frac{t}{\hbar} \hat H_{02}}-e^{i\frac{s}{\hbar} \hat H} \ e^{-i\frac{s}{\hbar} \hat H_{02}}  \Big) \Big]
\end{aligned}\]
Remark that each step makes sense. Indeed, we have that 
\[ \cW(\xi), \ e^{\frac{it}{\hbar}\hat{H}_{02}}, \ e^{\frac{it}{\hbar}\hat{H}}, \  S^{-1} \ e^{i\frac{t}{\hbar} \hat H} \ S, \  \ S^{-1} \ e^{-i\frac{t}{\hbar} \hat H_{02}} \ S  \in \cL(\cH), \quad \varrho_\hbar,\  \varrho_\hbar  (\hat{H}_0+1) \in \cL^1(\cH).  \]
We have also 
\[ \begin{aligned}
& \lim_{ s \rightarrow t } S^{-1} \frac{\Big( e^{i\frac{t}{\hbar} \hat H} \ e^{-i\frac{t}{\hbar} \hat H_{02}}-e^{i\frac{s}{\hbar} \hat H} \ e^{-i\frac{s}{\hbar} \hat H_{02}}  \Big)  }{t-s}=\frac{i}{\hbar}  \ S^{-1} \ e^{i\frac{t}{\hbar} \hat H}  \ (\hat{H}-\hat{H}_{02}) \  e^{-i\frac{t}{\hbar} \hat H_{02}}; 
\\ &\lim_{ s \rightarrow t } S^{-1} \frac{ \Big(e^{i\frac{t}{\hbar} \hat H_{02}}  \ e^{-i\frac{t}{\hbar} \hat H} - e^{i\frac{s}{\hbar} \hat H_{02}}  \ e^{-i\frac{s}{\hbar} \hat H}\Big)}{t-s}=-\frac{i}{\hbar}  \ S^{-1} \ e^{i\frac{t}{\hbar} \hat H_{02}}  \ (\hat{H}-\hat{H}_{02}) \  e^{-i\frac{t}{\hbar} \hat H}.
\end{aligned}\]
Plugging these limits in the Duhamel's formula, we get the desired result.
\end{proof}
\subsubsection{The commutator expansion}\label{top4comm}
\noindent The aim of this subsection is to  expand the commutator $\big[ \cW(\xi),\hat H_I(s) \big] $  in the above Duhamel formula \eqref{top4duhamelformula} in terms of the parameter $\hbar \in (0,1)$. 
\begin{lemma}[Time evolved equation of $\hat H_I(s)$]\label{top4theinteractionterm}
For any $s \in \R$, the time evolved interaction term $\hat H_I(s)$ takes the following form 
\be \label{top4timeevolvedH-k}
\hat H_I(s)= \sum_{j=1}^n f_j(\hat p_j)+ V(\hat q)+\sum_{j=1}^{n} \hat a_\hbar(g_j(s))+ \hat a^*_\hbar(g_j(s)),
\ee
where we have introduced 
\be \label{top4gjs}
g_j(s)\equiv g_j(s)(\hat{q}):=\frac{\chi(k)}{\sqrt{\omega(k)}}  \ e^{-2\pi ik\cdot \hat q_j +is\omega(k) }.
\ee
\end{lemma}

\begin{proof} We have 
 \[ \hat H-\hat H_{02}=\sum_{j=1}^n f_j(\hat p_j)+ V(\hat q)+ \sum_{j=1}^{n} \hat a_\hbar(g_j)+ \hat a^*_\hbar(g_j),\]
 where the function $g_j$ is given by
 \be \label{top4g} g_j\equiv g_j(\hat{q}):=\frac{\chi(k)}{\sqrt{\omega(k)}}  \ e^{-2\pi ik\cdot \hat q_j }.
 \ee
Then, we have with $\hat q=(\hat q_1,\cdots,\hat q_n) $ 
\[\hat H_I(s)=e^{i\frac{s}{\hbar} \hat H_{02}} \Big(\sum_{j=1}^n f_j(\hat p_j)+V(\hat q)+ \sum_{j=1}^{n} \hat a_\hbar(g_j)+ \hat a^*_\hbar(g_j)\Big)e^{-i\frac{s}{\hbar} \hat H_{02}}. \]
It is sufficient then to look at the following identity
\[ e^{i\frac{s}{\hbar} \hat{H}_{02}} \ \hat a^{\sharp}(g_j) \ e^{-i\frac{s}{\hbar} \hat{H}_{02}}=\hat a^{\sharp} ( g_j(s)).
\]
\end{proof}
\noindent Now, since the Weyl operator $\cW(\xi)$ is a unitary operator, we have 
\be \label{top4expcom} \frac{1}{\hbar}\big[ \cW(\xi),\hat H_I(s) \big]= \frac{1}{\hbar}\Big( \cW(\xi)\hat H_I(s)\cW(\xi)^* -\hat H_I(s)  \Big) \cW(\xi). \ee

\begin{lemma}[Expression for the commutators] \label{top4lemma3.3} For any $s\in \R$ and $\xi=( p_0,q_0,\alpha_0) \in X^{1/2}$, the following holds true with $q_0=(q_{01},\cdots,q_{0n})$ and $p_0=(p_{01},\cdots,p_{0n})$
\[
\begin{aligned}
\cW(\xi)\hat H_I(s)\cW(\xi)^*&=\sum_{j=1}^n f_j(\hat p_j-\hbar p_{0j})+V\big(\hat q-\hbar  q_{0} \big)
\\ &
+\sum_{j=1}^{n} \hat a_\hbar(\tilde g_j(s))+ \hat a^*_\hbar(\tilde g_j(s))+\frac{i \hbar }{\sqrt{2}} \Big(\langle \alpha_0,\tilde g_j(s) \rangle_{L^2(\R^d,\C)} -\langle\tilde g_j(s),\alpha_0 \rangle_{L^2(\R^d,\C)}\Big),
\end{aligned}
\]
where we have introduced 
\be \label{top4g(s)} \tilde g_j(s): =\frac{\chi(k)}{\sqrt{\omega(k)}}  \ e^{-2\pi ik\cdot (\hat q_j-\hbar q_{0j})+is\omega(k) }=e^{2\pi ik\cdot  q_{0j}\hbar } \ g_j(s) .\ee
\end{lemma}

\begin{proof}
Let $\hat q_j=(\hat q_j^\nu)_{\nu=1,\cdots d}$ and $\hat p_j=(\hat p_j^\nu)_{\nu=1,\cdots d}$. The results follow from the following identities 
\begin{align}
&W_1(z_0)  \ \hat{q}_j^\nu  \ W_1(z_0)^*= \hat{q}_j^\nu -\hbar q_{0j}^\nu,\label{top4comm1}
\\&
 W_1(z_0)  \ \hat{p}_j^\nu  \ W_1(z_0)^*= \hat{p}^\nu_j -\hbar p_{0j}^\nu,\label{top4comm2}
\\& W_2(\alpha_0) \  \hat a_\hbar^*(f) \  W_2(\alpha)^*= \hat a_\hbar^*(f) +\frac{i \hbar}{\sqrt{2}} \langle \alpha_0, f \rangle_{L^2},\label{top4comm3}
\\&W_2(\alpha_0) \  \hat a_\hbar(f) \  W_2(\alpha)^*= \hat a_\hbar(f) -\frac{i \hbar}{\sqrt{2}} \langle  f,\alpha_0 \rangle_{L^2}.\label{top4comm4}
\end{align}
We start proving the first identity \eqref{top4comm1}. Recall  from \eqref{top4weylheseinbergtranslation} that we have  
\[
 W_1(z_0)=e^{i( p_0\cdot \hat q-q_0\cdot \hat p)},\quad W_1(z_0)^*=e^{-i( p_0\cdot \hat q-q_0\cdot \hat p)}  \]
Define \[K(t):=  e^{i t ( p_0\cdot \hat q-q_0\cdot \hat p)} \ \hat{q}_j^\nu  \  e^{-i t ( p_0\cdot \hat q-q_0\cdot \hat p)}\]
Since $\hat q$ and $\hat p$ are self adjoint operators, we claim using Taylor expansions that 
\be \label{top4taylorcomm}K(t) =K(0)+t K^\prime(0).\ee
Indeed, we have,  using the commutation relation $[\hat{q}_j^\nu,\hat p_j^\nu]=i\hbar$, that  
\[ K^\prime(0)= \frac{d}{dt} K(t)|_{t=0}= e^{i t ( p_0\cdot \hat q-q_0\cdot \hat p)} \ i[ (p_0\cdot \hat q-q_0\cdot \hat p), \hat{q}_j^\nu ] \  e^{-i t ( p_0\cdot \hat q-q_0\cdot \hat p)}|_{t=0}=-\hbar \ q_{0j}^\nu.\]
This implies that $K^r(0)=0$, for all $r\geq 2$. Take $t=1 $ in \eqref{top4taylorcomm} and since $K(0)=\hat q_j^\nu$, we  get  \eqref{top4comm1}. Similarly, we can prove the  identity \eqref{top4comm2}. Also the two identities \eqref{top4comm3} and \eqref{top4comm4} can be proved by simiar way using the commutation relations on the Fock space.
\vskip 1mm
\noindent In particular, the identity \eqref{top4comm1} gives 
\[W_1(z_0)  \ g_j(s)  \ W_1(z_0)^*=\tilde{g}_j(s).\]
\end{proof}

\begin{lemma}[The expansion of the commutator]\label{top4expcommutator}
For any $s\in \R$ and $\xi=( p_0,q_0,\alpha_0) \in X^{1/2}$,  we have the following expansion of the commutator in terms of the  semiclassical parameter $\hbar \in (0,1)$
\be\label{top4comexp}
\begin{aligned}
\frac{1}{\hbar}\big[ \cW(\xi),\hat H_I(s) \big]= \Big( {\rm B}_0(s,\hbar,\xi)+\hbar \, {\rm B}_1(s,\hbar,\xi)\Big) \cW(\xi).
\end{aligned}
\ee
The two terms $ {\rm B}_0$ and $ {\rm B}_1$ are identified as follows
\begin{align}
&\begin{aligned}
 {\rm B}_0(s,\hbar,\xi):=&- \sum_{j=1}^n \nabla f_j(\hat p_j)\cdot p_{0j}-\nabla V(\hat q)\cdot q_{0}
\\ &
+\sum_{j=1}^{n} \hat a_\hbar\Big(\frac{\tilde{g}_j(s)-g_j(s)}{\hbar}\Big)+\hat  a^*_\hbar\Big(\frac{\tilde{g}_j(s)-g_j(s)}{\hbar}\Big)
\\ & +\sum_{j=1}^{n}\frac{i  }{\sqrt{2}} \Big(\langle \alpha_0,\tilde g_j(s)\rangle_{L^2(\R^d,\C)} -\langle  \tilde g_j(s),\alpha_0 \rangle_{L^2(\R^d,\C)}\Big),
\end{aligned}\label{top4mainterm}
\\&
\begin{aligned}{\rm B}_1(s,\hbar,\xi):= \Theta_1(\hbar,\xi)+\Theta_2(\hbar,\xi)
,\label{top4remainderterm}
\end{aligned}
\end{align}
where $\Theta_1$ and $\Theta_2$ are identified below in the proof.  Moreover, we have also the following estimates
\begin{align}
& \Vert  (\hat{H}_0+1)^{-1/2} {\rm B}_0(s,\hbar,\xi) (\hat{H}_0+1)^{-1/2} \Vert_{\cL(\cH)}\lesssim \ \Big(\|\chi\|_{L^2}+\Big\Vert\sqrt{\omega} \ {\chi}{}\Big\Vert_{L^2} \Big) \  \|\xi\|_{X^0},\label{top4estimate1}
\\& \Vert (\hat{H}_0+1)^{-1/2}{\rm B}_1(s,\hbar,\xi) (\hat{H}_0+1)^{-1/2}\Vert_{\cL(\cH)}\lesssim  \|\xi\|_{X^0}^2.\label{top4estimate2}
\end{align}

\end{lemma}

\begin{proof}
\noindent Exploiting Lemma \ref{top4theinteractionterm} and Lemma \ref{top4lemma3.3} inside \eqref{top4expcom}, the commutator expansion becomes
\be \label{top4exp1com}
 \begin{aligned}
\frac{1}{\hbar} \big[ \cW(\xi),\hat H_I(s) \big] 
  &=\frac{1}{\hbar} \Big[\sum_{j=1}^n \big(f_j(\hat{p}_j-\hbar p_{0j})-f_j(\hat{p}_j) \big)+ V\big(\hat q-\hbar q_{0}\big)
  -V\big(\hat q\big)
\\& 
+\sum_{j=1}^{n} \Big( \hat a_\hbar(\tilde{g}_j(s)-g_j(s))+ \hat a^*_\hbar(\tilde{g}_j(s)-g_j(s)) \Big)
\\&+\frac{i \hbar }{\sqrt{2}} \sum_{j=1}^{n}\Big(\langle \alpha_0,\tilde g_j(s)\rangle_{L^2(\R^d,\C)} -\langle  \tilde g_j(s),\alpha_0 \rangle_{L^2(\R^d,\C)}\Big)\Big] \cW(\xi).
 \end{aligned}
 \ee
We start first by expanding the first line and then proving some estimates for the remaining terms. Let $X\in \R^d$ and $Y=(Y_1,\cdots,Y_n)\in \R^{dn}$. We apply  Taylor series to the two functions
 \begin{align*}
 &t\longrightarrow A(t):=f_j(X-t \ \hbar \ p_{0j}),
 \\ & t\longrightarrow B(t):=V(Y-t \ \hbar \ q_{0}).
 \end{align*}
 We get 
 \[ A(t)=A(0)+tA^\prime(0)+\int_{0}^{t} A^{\prime \prime}(s)(t-s) \ ds, \]
 and 
 \[ B(t)=B(0)+tB^\prime(0)+\int_{0}^{t} B^{\prime \prime}(s)(t-s) \ ds.\]
 Let $t=1$ in the above formulas and since $\hat{p}$ and $\hat{q}$
 are self adjoint operators, we get 
 \begin{align}
&f_j(\hat{p}_j-\hbar p_{0j})=f_j(\hat p_j)-\hbar \nabla f_j(\hat{p}_j)\cdot p_{0j}+\hbar^2 \underbrace{{ \int_0^1 p_{0j}^T \  H_{f_j}(\hat{p}_j-\hbar p_{0j} s) \  p_{0j} \ (1-s) \ ds}}_{:=\Theta_1(\hbar,\xi)};
\\& V(\hat q-\hbar q_0)=V(\hat{q})-\hbar \nabla V(\hat{q})  \cdot q_0+\hbar^2\underbrace{\int_0^1 q_{0}^T \  H_{V}(\hat{q}-\hbar q_{0} s) \  q_{0} \ (1-s) \ ds}_{:=\Theta_2(\hbar,\xi)},
 \end{align}
 where the notation $\cdot^T$ represents the transpose. Moreover, the two terms $H_{f_j}$ and $H_V$ are respectively the Hessian matrices related to $f_j$ and $V$.
 This implies 
 \begin{align}
 &f_j(\hat{p}_j-\hbar p_{0j})-f_j(\hat p_j)=-\hbar \nabla f_j(\hat{p}_j)\cdot p_{0j}+\hbar^2 \  \Theta_1(\hbar,\xi),\label{top4expf}
 \\ &  V(\hat q-\hbar q_0)-V(\hat{q})=-\hbar \nabla V(\hat{q})  \cdot q_0+\hbar^2 \ \Theta_2(\hbar,\xi).\label{top4expv}
 \end{align}
 And thus, using \eqref{top4expf}-\eqref{top4expv}, the commutator is expanded as indicated in \eqref{top4comexp}. 
 Now, to obtain the two estimates \eqref{top4estimate1} and \eqref{top4estimate2}, we need first to prove that  the function \[F_j(\hbar,s):=(\tilde{g}_j(s)-g_j(s))/\hbar:L^2(\R^d,dx_j) \longrightarrow L^2(\R^d,dx_j) \otimes L^2(\R^d,dk)\]  is bounded uniformly in $\hbar \in (0,1)$. Indeed, we have for all  $\psi \in L^2(\R^d,dx_j)$
\[\begin{aligned}
\Vert F_j(\hbar,s) \ \psi \Vert_{L^2_{x_j} \otimes L^2_k}^2
&=\int_{\R^d} \int_{\R^d}\Big \vert(F_j(\hbar,s) \ \psi)(x_j,k) \Big \vert^2  dx_j \ dk 
\\ &=\int_{\R^d} \int_{\R^d}\Big \vert\Big[g_j(s)\big(\frac{e^{2\pi i k \cdot q_{0j} \hbar}-1}{\hbar}\big)\psi \Big](x_j,k) \Big \vert^2  dx_j \ dk  
\\&=\int_{\R^d} \int_{\R^d}\Big \vert \frac{\chi(k)}{\sqrt{\omega(k)}} \  e^{-2\pi ik \cdot \hat{q}_j+is\omega(k)}   \ \big(\frac{e^{2\pi i k \cdot q_{0j}\hbar}-1}{\hbar}\big) \    \psi(x_j) \Big \vert^2  dx_j \ dk . 
\end{aligned}\] 
Now, with the aid of Fubini and  the estimate $\vert e^{iy}-1 \vert\leq \sqrt{2} \vert y \vert$, we find that 
\[\begin{aligned}
\Vert F_j(\hbar,s) \ \psi \Vert_{L^2_{x_j} \otimes L^2_k}^2
&\leq 8 \pi^2 \ \Vert \xi \Vert_{X^0}^2 \ \int_{\R^d} \int_{\R^d}\Big \vert {\sqrt{\omega(k)} {\chi(k)}} \Big \vert^2 \    \Big \vert  \psi(x_j)\Big \vert^2 \ dx_j \ dk
\\ & = 8 \pi^2 \ \Vert \chi \Vert_{\cG^{1/2}}^2   \  \Vert \xi \Vert_{X^0}^2 \  \Vert \psi \Vert_{L^2_{x_j}}^2.
\end{aligned}\] 
 We get finally, with some $C>0$,  that 
\be \label{top4bound1} 
\Vert F_j(\hbar,s) \Vert_{\cL(L^2_{x_j},L^2_{x_j} \otimes L^2_k)}\leq C \  \Vert \chi \Vert_{\cG^{1/2}}    \ \Vert \xi \Vert_{X^0}.
\ee 
 Now, using the estimates in Lemma \ref{top4lemma1} on the creation-annihilation operators together with the above estimate for $F_j(\hbar,s)$, we can easily prove  \eqref{top4estimate1}. It is also  not hard to see that  \eqref{top4estimate2} hold true as a consequence of the fact that the Hessian matrices of $f_j$ and $V$ are bounded.
\end{proof}
\noindent Our focus is on taking the classical limit $\hbar \rightarrow 0$. To accomplish this, it is crucial to establish a uniform bound on the expansion derived in Lemma \ref{top4expcommutator}, particularly for the remainder term. Let $S=(\hat{H}_0+1)^{1/2}$, we have 
 \be \label{top4boundforconv}\begin{aligned} {\rm Tr}\Big(\frac{1}{\hbar} \big[ \cW(\xi),\hat H_I(s) \big] \, \tilde \varrho_\hbar(s) \Big) &={\rm Tr}\Big[\underbrace{S^{-1} \  {\rm B}_0(s,\hbar,\xi) \ S^{-1}}_{\in \cL(\cH)} \underbrace{S  \  \cW(\xi) \ S^{-1}}_{\in \cL(\cH)} \underbrace{S \  \tilde{\varrho}_\hbar(s) \ S}_{\in \cL^1(\cH)} \Big]
\\& +\hbar \ {\rm Tr}\Big[\underbrace{ S^{-1} {\rm B}_1(\hbar,s,\xi) \ S^{-1}}_{\in \cL(\cH)} \underbrace{S  \  \cW(\xi) \ S^{-1}}_{\in \cL(\cH)} \underbrace{S \  \tilde{\varrho}_\hbar(s) \ S }_{\in \cL^1(\cH)} \Big]
\end{aligned}\ee

 \begin{itemize}
\item [$\triangleleft$] Lemma \ref{top4expcommutator} assures that the first term in each of the above two lines in \eqref{top4boundforconv} is bounded.
\item [$\triangleleft$] The Weyl-Heisenberg operator estimates presented in Lemma \ref{top4weylheisenberg} guarantee that the bound of the second term in the above two lines in \eqref{top4boundforconv}  holds.
\item [$\triangleleft$] The bound of the  last term in each of the above two lines in \eqref{top4boundforconv} follows from Assumption \eqref{top4S0} and \eqref{top4S1} in conjunction with the equivalent relation between $\hat{H}$ and $\hat{H}_0$ outlined in Lemma \ref{top4inequality}.
\end{itemize}
Our next step is to take the limit in the Duhamel formula \eqref{top4duhamelformula} as $\hbar$ approaches zero. Using the above arguments, we can disregard the remainder term when passing to the limit $\hbar \rightarrow 0$ in the Duhamel formula \eqref{top4duhamelformula}. 
 We achieve this in the next section by extracting a subsequence.

\section{Existence of Wigner measure}
According to Definition \ref{top4definitionwigner}, the   Wigner measures of $\tilde{\varrho}_\hbar(t)$ is  obtained by taking limits of the following map:
\begin{equation}\label{top4dynm}
 \xi \rightarrow {\rm Tr} \big[ \cW(\xi) \ \tilde{\varrho}_\hbar(t) \big].
\end{equation}
Thus, the first task is to verify that the Wigner measure associated to the above map is unique for all times.
It is worth noting that, given our assumptions on the initial  states $(\varrho_\hbar)_{\hbar\in (0,1)}$, the  associated set of Wigner measures
\[\cM(\varrho_\hbar,\  \hbar \in(0,1) ) \]
is non-empty. To ensure that the sets of Wigner measures
\[\cM(\varrho_\hbar(t),\  \hbar \in(0,1) ) \quad \text{and} \quad \cM(\tilde \varrho_\hbar(t),\  \hbar \in(0,1) ) \]
 are also non-empty, it is crucial to demonstrate that assumptions \eqref{top4S0} and \eqref{top4S1} can be uniformly propagated in time by both families  of states  $(\varrho_\hbar(t))_{\hbar \in (0,1)} $ and $(\tilde \varrho_\hbar(t))_{\hbar\in (0,1)} $. This is established in Subsection \ref{top4propagation}. Subsequently, in Subsection \ref{top4extraction}, we prove that the map \eqref{top4dynm} has a unique limit that holds for all times in compact interval.

\subsection{Propagation of assumptions}\label{top4propagation}
\noindent In order to establish the existence of a unique Wigner measure that holds for all times, we demonstrate that if an initial state $\varrho_\hbar$ is localized uniformly in $\hbar$, then it will remain localized uniformly with respect to the semiclassical parameter $\hbar \in (0,1)$ for all times in compact interval. We prove this result separately for particle operators in Paragraph \ref{top4particleoperators} and for field operators in Paragraph \ref{top4momentestimate}. Finally, in Paragraph \ref{top4propaaa}, we establish that both families of states $(\varrho_\hbar(t))_{\hbar \in (0,1)}$ and $(\tilde \varrho_\hbar(t))_{\hbar \in (0,1)}$ uniformly satisfy \eqref{top4S0} and \eqref{top4S1} for all times.
\subsubsection{Position and Momentum operator estimates}\label{top4particleoperators}
In this part, we prove some uniform estimates (in $\hbar$)  related  to the two operators $\hat{p}^2$ and $\hat{q}^2$.
\begin{lemma}[Position operator's estimate]\label{top4local}Assume that \eqref{top4A0} and   $\omega^{-{1}/{2}}\chi\in L^2(\R^d,dk)$. Then, there exists constants $C_1,C_2>0$ such that for all $\psi \in D(\hat{H}_0^{{1}/{2}})\cap D(\hat{q})$, all $t\in \R$ and all $\hbar\in (0,1)$:
\begin{equation}\label{top4second}
 \langle e^{-i \frac{t}{\hbar}  \hat{H}}\psi, \ \hat{q}^2  \ e^{-i \frac{t}{\hbar}  \hat{H}}\psi\rangle \leq C_1\langle \psi,  \ (\hat{H}_0+\hat{q}^2+1)  \ \psi \rangle \  e^{C_2|t|}.
\end{equation}
\end{lemma}
\begin{proof}
Let $\Theta_1(t):= \langle  e^{-i\frac{t}{\hbar} \hat{H}}\psi,  \ \hat{q}^2  e^{-i\frac{t}{\hbar} \hat{H}}\psi \rangle$. We have 
\[\Theta_1(t) =\Theta_1(0)+ \int_{0}^{t} \dot \Theta_1(s) \ ds. \]
Then Stone's Theorem implies that
\[ \dot \Theta_1(t)= \frac{1}{\hbar} \big\langle  e^{-i\frac{t}{\hbar} \hat{H}}\psi,  \ i [\hat{H},\hat{q}^2  ] \ e^{-i\frac{t}{\hbar} \hat{H}}\psi \big \rangle. \]
 Now, using some commutation relations, we get 
\[i [\hat{H},\hat{q}^2  ] = i\sum_{j=1}^n [f_j(\hat{p}_j),\hat{q}_j^2]=\hbar \sum_{j=1}^n \big[  \nabla f_j (\hat{p}_j)\cdot \hat q_j+\hat q_j\cdot \nabla f_j (\hat{p}_j) \big] . \]
Define $\psi(t):= e^{-i\frac{t}{\hbar} \hat{H}} \psi$. Since $\hat{q}_j$ and  $\nabla f_j(\hat{p}_j) $ are self-adjoint operators, we have the following estimates
\begin{align}
& \begin{aligned}\langle \psi(t), \ \nabla f_j (\hat{p}_j)\cdot \hat q_j  \ \psi(t)\rangle& \leq  \Vert \nabla f_j (\hat{p}_j) \  \psi(t) \Vert \ \Vert  \hat{q}_j  \ \psi(t) \Vert  
\\ & \leq \frac{1}{2} \big [ \ \Vert \nabla f_j(\hat{p}_j) \  \psi(t) \Vert^2 + \ \Vert  \hat{q}_j  \ \psi(t) \Vert^2 \big]
\\ &=\frac{1}{2} \big [ \ \Vert \nabla f_j (\hat{p}_j) \  \psi(t) \Vert^2+ \ \Theta_1(t) \big] ,
 \end{aligned}  \label{top4est1}
 \\ &  \begin{aligned}\langle \psi(t), \ \hat{q}_j \cdot  \nabla f_j (\hat{p}_j)  \ \psi(t)\rangle& \leq  \Vert \hat{q}_j \  \psi(t) \Vert \ \Vert \nabla f_j (\hat{p}_j)   \ \psi(t) \Vert 
\\ & \leq \frac{1}{2} \big [ \ \Vert \nabla f_j (\hat{p}_j) \  \psi(t) \Vert^2 + \ \Vert  \hat{q}_j  \ \psi(t) \Vert^2 \big]\\ &=\frac{1}{2} \big [ \ \Vert \nabla f_j(\hat{p}_j) \  \psi(t) \Vert^2+ \ \Theta_1(t)\big] ,
 \end{aligned} \label{top4est2} 
 \end{align}
where we have used the identity $2a\cdot b \leq a^2+b^2$.
 At this stage, we have to consider separately the two cases: the semi-relativistic and the non-relativistic case since  the function $\nabla f_j$ is bounded in the first case and not in the second one.
\vskip 1mm
\noindent 
{\it For semi-relativistic case:}
\vskip 2mm
\noindent 
Note that   $\nabla f_j(\hat{p}_j)$ is a bounded operator.
This implies that for some $c_1>0$, we have 
\[\Vert \nabla f_j (\hat{p}_j) \  \psi(t) \Vert^2 \leq c_1 \Vert   \psi \Vert^2. \] 
This gives 
\[ \Theta_1(t)\leq \Theta_1(0)+ {c_1} \langle \psi,\psi \rangle  \ t + \int_{0}^t \Theta_1(s) \  ds.  \]
Now using Gronwall's lemma and the estimate $te^{t}\leq e^{c t }$ for some $c>0$, we find with some $C_1, \ C_2>0$
\[\begin{aligned}
 \Theta_1(t)\leq \Big[  \Theta_1(0)+ {c_1}  \langle \psi,\psi \rangle \  t\Big]  \ e^{t}  \leq 	C_1 \langle \psi, \ (\hat{q}^2+1)  \ \psi  \rangle \ e^{C_2 |t|}
  \leq C_1 \langle \psi, \ (\hat{H}_0+\hat{q}^2+1) \  \psi  \rangle \ e^{C_2 |t|}.
 \end{aligned}\]
\vskip 2mm
 \noindent 
 {\it For non-relativistic case:}
\vskip 2mm 
\noindent We have $\nabla f_j(\hat{p}_j)={\hat{p}_j}/{M_j}$. This implies \eqref{top4est1} and \eqref{top4est2} become 
\begin{align*}
& \begin{aligned}\langle \psi(t), \ \nabla f_j (\hat{p}_j)\cdot \hat q_j  \ \psi(t)\rangle & \leq \frac{1}{2}  \big [ \langle \psi(t), \ \frac{\hat{p}_j^2}{M_j^2} \  \psi(t)\rangle + \ \Theta_1(t)\big]
\\ & \leq \frac{1}{M_j} \ \langle \psi(t), \ \hat{H}_0 \  \psi(t)\rangle + \ \frac{1}{2}\Theta_1(t)
\end{aligned} 
\\ & \begin{aligned}\langle \psi(t), \ \hat{q}_j \cdot  \nabla f_j (\hat{p}_j)  \ \psi(t)\rangle& \leq  \frac{1}{2} \big [ \langle \psi(t), \ \frac{\hat{p}_j^2}{M_j^2} \  \psi(t)\rangle + \ \Theta_1(t)\big]
\\ & \leq \frac{1}{M_j} \ \langle \psi(t), \ \hat{H}_0 \  \psi(t)\rangle + \ \frac{1}{2}\Theta_1(t).
\end{aligned} 
\end{align*}
By Lemma \ref{top4inequality}, we have
 \[\langle  e^{-i \frac{t}{\hbar}\hat{H}}\psi,  \hat{H}_0 \ e^{-i \frac{t}{\hbar}\hat{H}}\psi \rangle \lesssim \langle  e^{-i \frac{t}{\hbar}\hat{H}}\psi,  (\hat{H}+1) \ e^{-i \frac{t}{\hbar}\hat{H}}\psi \rangle =\langle  \psi,  (\hat{H}+1) \ \psi \rangle \lesssim \langle  \psi,  (\hat{H}_0+1) \ \psi \rangle.\]
 This leads with    some $c_2 \in \R^*_+$ to the following inequality
\[ \Theta_1(t)\leq \Theta_1(0)+  c_2 \langle \psi,  (\hat{H}_0+1) \ \psi \rangle \  t + \int_{0}^t \Theta_1(s) \  ds.  \]
Now using Gronwall's Lemma  and the estimates $te^{ct}\leq e^{c^\prime t }$, we find with some $C_1, \ C_2>0$
\[\begin{aligned}
 \Theta_1(t)\leq \Big[  \Theta_1(0)+  c_2 \langle \psi, ( \hat{H}_0+1) \ \psi \rangle \  t\Big]  \ e^{\int_{0}^t 1 \ ds}&  \leq C_1 \langle \psi, \ (\hat{H}_0+\hat{q}^2+1) \psi  \rangle \ e^{C_2 |t|}
 \\ & \leq C_1 \langle \psi, \ (\hat{H}_0+\hat{q}^2+1) \psi  \rangle \ e^{C_2 |t|}.
 \end{aligned}\]
 \end{proof}
 Now, we give some uniform estimates for the momentum operator just in the semi-relativistic case: $f_j(\hat{p}_j)=\sqrt{\hat p_j^2+M_j^2}$.
 \begin{lemma}[Momentum operator's estimate]\label{top4local1}Assume that \eqref{top4A0} and  $\omega^{{1}/{2}}\chi\in L^2(\R^d,dk)$. Then, there exists constants $C_1,C_2>0$ such that for all $\psi \in D(\hat{H}_0^{1/2})\cap D(\hat{p})$, all $t\in \R$ and all $\hbar\in (0,1)$:
\begin{align}
  \langle e^{-i \frac{t}{\hbar}  \hat{H}}\psi, \ \hat{p}^2  \ e^{-i \frac{t}{\hbar}  \hat{H}}\psi\rangle \leq C_1\langle \psi,  \ (\hat{H}_0+\hat{p}^2+1)  \ \psi \rangle \  e^{C_2|t|}.\label{top4first1}
\end{align}
\end{lemma}

\begin{proof}
Define 
\[\Theta_2(t):= \langle \  \psi(t), \  \hat{p}^2 \  \psi(t) \  \rangle,\qquad \psi(t):=e^{-i\frac{t}{\hbar} \hat{H}} \ \psi . \]
We have that the map $t\rightarrow \Theta_2(t)$ is differentiable  with 
\[ \dot{\Theta}_2(t)= \frac{i}{\hbar}  \ \langle   \psi(t), \ [\hat{H}, \  \hat{p}^2] \  \psi(t)  \rangle.\]
Then, Duhamel formula implies that 
\[ {\Theta_2(t)}=\Theta_2(0)+\int_0^{t} \dot{\Theta}_2(s) \ ds.\]
Let us compute first the explicit expression for the function $\dot{\Theta}_2(t) $. To do that, we need first to deal with  the commutator $ [\hat{H}, \  \hat{p}^2]$. Indeed, we have 
\[
\begin{aligned}
 \big[\hat{H},   \hat{p}^2 \big] &= [d\Gamma(\omega)+ \sum_{j=1}^n \sqrt{\hat{p}_j^2+M_j^2}+V(\hat{q})+ \hat{H}_1,\hat{p}^2]
\\ & =[V(\hat{q}),\hat{p}] \ \hat{p}+\hat{p} \ [V(\hat{q}),\hat{p}]+ [\hat{H}_1,\hat{p}] \ \hat{p}+ \hat{p} \ [\hat{H}_1,\hat{p}].
\end{aligned}\]
Recall that 
\[  \hat a^\sharp_\hbar(G)= \sum_{j=1}^n \hat{a}^\sharp_\hbar(g_j(\hat{q})),\qquad g_j(\hat{q})(k ):= \frac{\chi(k)}{\sqrt{\omega(k)}} \ e^{-2\pi i k\cdot \hat{q}_j}.\]
We can then assert that 
\begin{align*}
&  \hat a_\hbar(G)= \sum_{j=1}^n \int_{\R^d} \frac{\chi(k)}{\sqrt{\omega(k)}} \ e^{2\pi i k\cdot \hat{q}_j} \ \hat{a}_\hbar(k) \ dk =: \sum_{j=1}^n B_j(\hat{q}_j),
\\& \hat{a}^*_\hbar (G)=  \sum_{j=1}^n \int_{\R^d} \frac{\chi(k)}{\sqrt{\omega(k)}} \ e^{-2\pi i k\cdot \hat{q}_j} \ \hat{a}_\hbar^*(k) \ dk =: \sum_{j=1}^n B^*_j(\hat{q}_j),
\end{align*}
where ${q}_j \rightarrow B^\sharp_j({q}_j)$ is analytic function.
We know that for any analytic function $F$
\[ [F(\hat{q}_j), \hat{p}_j]= i\hbar \ \frac{\partial F(\hat{q}_j)}{\partial {q}_j}, \]
This gives 
\begin{align*}
&[V(\hat{q}),\hat{p}] \ \hat{p} =\sum_{j=1}^n i \hbar \  \nabla_{{q}_j} V(\hat{q})\cdot \hat{p}_j,
\\ & \hat{p} \ [V(\hat{q}),\hat{p}] =\sum_{j=1}^n i \hbar \  \hat{p}_j \cdot  \nabla_{{q}_j} V(\hat{q}),
\\ & [\hat{a}_\hbar(G),\hat{p}] \ \hat{p} = \big[ \sum_{j=1}^n B_j(\hat{q}_j), \hat{p}\big] \ \hat{p}=\sum_{j=1}^n i \hbar \ \frac{\partial B_j(\hat{q}_j)}{\partial{q}_j} \cdot \hat{p}_j= \sum_{j=1}^n i \hbar \ \hat{a}_\hbar(\tilde{g}_j) \cdot \hat{p}_j,
\\&   \hat{p} \ [\hat{a}_\hbar(G),\hat{p}] = \sum_{j=1}^n i \hbar \ \hat{p}_j \cdot  \hat{a}_\hbar(\tilde{g}_j),
\\ & [\hat{a}^*_\hbar(G),\hat{p}] \ \hat{p} =\big[ \sum_{j=1}^n B^*_j(\hat{q}_j), \hat{p}\big] \ \hat{p}=\sum_{j=1}^n i \hbar \ \frac{\partial B^*_j(\hat{q}_j)}{\partial{q}_j} \cdot \hat{p}_j= \sum_{j=1}^n i \hbar \ \hat{a}^*_\hbar(\tilde{g}_j) \cdot \hat{p}_j, 
\\ &   \hat{p} \ [\hat{a}^*_\hbar(G),\hat{p}] = \sum_{j=1}^n i \hbar \ \hat{p}_j \cdot  \hat{a}^*_\hbar(\tilde{g}_j),
\end{align*}
where we have introduced the term $\tilde{g}_j$ as follows
\[ \tilde g_j:= -2\pi i k \frac{\chi(k)}{\sqrt{\omega(k)}} \ e^{-2\pi i k \cdot \hat{q}_j}.\]
This implies that 
\[ \begin{aligned}
   \big[\hat{H}, \ \hat{p}^2\big]= i\hbar \ \sum_{j=1}^n \Big[ 
 \hat{p}_j \cdot \big( \nabla_{{q}_j} V(\hat{q})+\hat{a}_\hbar(\tilde{g}_j) +\hat{a}^*_\hbar(\tilde{g}_j)  \big) 
  +\big( \nabla_{{q}_j} V(\hat{q})+\hat{a}_\hbar(\tilde{g}_j) +\hat{a}^*_\hbar(\tilde{g}_j)  \big) \cdot \hat{p}_j    \Big].
\end{aligned}\]
We conclude that 
\[ \begin{aligned}
\dot{\Theta}_2(t)&= -\sum_{j=1}^n\Big[ 
\langle \psi(t), \  \hat{p}_j \cdot \big( \nabla_{{q}_j} V(\hat{q})+\hat{a}_\hbar(\tilde{g}_j) +\hat{a}^*_\hbar(\tilde{g}_j)  \big)  \   \psi(t) \rangle 
\\ &\qquad \quad  + \langle \psi(t), \   \big( \nabla_{{q}_j} V(\hat{q})+\hat{a}_\hbar(\tilde{g}_j) +\hat{a}^*_\hbar(\tilde{g}_j)  \big) \cdot \hat{p}_j \  \psi(t) \rangle \Big].
\end{aligned}
\]
We estimate now each part. Indeed, we have
\[ \begin{aligned}
\langle \psi(t), \  \hat{p}_j \cdot  \nabla_{{q}_j} V(\hat{q})  \   \psi(t) \rangle &= \langle \hat{p}_j  \  \psi(t), \  \nabla_{{q}_j} V(\hat{q})   \   \psi(t) \rangle  
\\ & \leq  \Vert \hat{p}_j  \ \psi(t) \Vert  \ \Vert \nabla_{{q}_j} V(\hat{q})  \   \psi(t) \Vert
\\ & \leq  \Vert \hat{p}_j  \ \psi(t) \Vert \ \Vert \nabla_{{q}_j} V\Vert_{L^\infty}   \   \Vert \psi \Vert
\\ & \leq  \frac{1}{2}  \ \Vert \nabla_{{q}_j} V\Vert_{L^\infty} \Big [   \Vert \hat{p}_j  \ \psi(t) \Vert^2+   \Vert \psi \Vert^2\Big].
\end{aligned}\]
Similarly, we have 
 \[ \begin{aligned}
\langle \psi(t), \    \nabla_{{q}_j} V(\hat{q})\cdot \hat{p}_j \   \psi(t) \rangle \leq  \frac{1}{2}  \ \Vert \nabla_{{q}_j} V\Vert_{L^\infty} \Big [ \Vert \hat{p}_j  \ \psi(t) \Vert^2+   \Vert \psi \Vert^2\Big].
\end{aligned}\]
We have also using Lemmas \ref{top4lemma1} and \ref{top4inequality}
\[ \begin{aligned}
\langle \psi(t), \ \hat{a}_\hbar(\tilde{g}_j) \cdot  \hat{p}_j    \   \psi(t) \rangle &= \langle \hat{a}^*_\hbar(\tilde{g}_j)  \  \psi(t), \  \hat{p}_j   \   \psi(t) \rangle  
\\ & \leq  \Vert \hat{p}_j  \ \psi(t) \Vert  \ \Vert\hat{a}^*_\hbar(\tilde{g}_j)  \   \psi(t) \Vert
\\ & \lesssim  \big[  \Vert \chi \Vert_{L^2}+ \Vert \sqrt{\omega} \ \chi \Vert_{L^2} \big] \ \Vert \hat{p}_j  \ \psi(t) \Vert \ \Vert (\hat{H}_{02}+1)^{1/2} \    \psi(t) \Vert
\\ & \lesssim   \big[  \Vert \chi \Vert_{L^2}+ \Vert \sqrt{\omega} \ \chi \Vert_{L^2} \big] \ \Vert \hat{p}_j  \ \psi(t) \Vert \ \Vert (\hat{H}_0+1)^{1/2} \    \psi(t) \Vert
\\ & \lesssim    \big[  \Vert \chi \Vert_{L^2}+ \Vert \sqrt{\omega} \ \chi \Vert_{L^2} \big] \ \Vert \hat{p}_j  \ \psi(t) \Vert \ \Vert (\hat{H}+a)^{1/2} \    \psi(t) \Vert
\\ & \lesssim   \big[  \Vert \chi \Vert_{L^2}+ \Vert \sqrt{\omega} \ \chi \Vert_{L^2} \big] \ \Vert \hat{p}_j  \ \psi(t) \Vert \ \Vert (\hat{H}+a)^{1/2} \    \psi \Vert
\\ & \lesssim   \big[  \Vert \chi \Vert_{L^2}+ \Vert \sqrt{\omega} \ \chi \Vert_{L^2} \big] \ \Vert \hat{p}_j  \ \psi(t) \Vert \ \Vert (\hat{H}_0+1)^{1/2} \    \psi \Vert
\\ & \lesssim       \big[  \Vert \chi \Vert_{L^2}+ \Vert \sqrt{\omega} \ \chi \Vert_{L^2} \big] \ \Big [   \Vert \hat{p}_j  \ \psi(t) \Vert^2+   \Vert (\hat{H}_0+1)^{1/2} \ \psi \Vert^2\Big].
\end{aligned}\]
Similarly, we have 
\begin{align*}
&  \langle \psi(t), \ \hat{a}^*_\hbar(\tilde{g}_j) \cdot  \hat{p}_j    \   \psi(t) \rangle \lesssim  \ \Vert \chi \Vert_{L^2} \ \Big [   \Vert \hat{p}_j  \ \psi(t) \Vert^2+   \Vert (\hat{H}_0+1)^{1/2} \ \psi \Vert^2\Big],
\\ &\langle \psi(t), \   \hat{p}_j\cdot  \hat{a}_\hbar(\tilde{g}_j)     \   \psi(t) \rangle \lesssim  \ \Vert \chi \Vert_{L^2} \ \Big [   \Vert \hat{p}_j  \ \psi(t) \Vert^2+   \Vert (\hat{H}_0+1)^{1/2} \ \psi \Vert^2\Big],
\\ &\langle \psi(t), \  \hat{p}_j\cdot \hat{a}^*_\hbar(\tilde{g}_j)     \   \psi(t) \rangle \lesssim \big[  \Vert \chi \Vert_{L^2}+ \Vert \sqrt{\omega} \ \chi \Vert_{L^2} \big] \ \Big [   \Vert \hat{p}_j  \ \psi(t) \Vert^2+   \Vert (\hat{H}_0+1)^{1/2} \ \psi \Vert^2\Big].
\end{align*}
We get at the end that there exists some  $C\in \R^*_+$ depending  on the quantities $ \Vert \chi \Vert_{L^2}  $,  $ \Vert \sqrt{\omega} \  \chi \Vert_{L^2}  $ and $  \Vert \nabla_{{q}_j} V\Vert_{L^\infty} $ such that 
\[ \dot{\Theta}_2(t) \leq  C \Big[ \langle \psi, \ (\hat{H}_0+1) \  \psi \rangle +\Theta_2(t) \Big].\]
We find then 
\[ \Theta_2(t)\leq \Theta_2(0)+ C \ \langle \psi, \ (\hat{H}_0+1) \  \psi \rangle  \ t +\int_0^t \Theta_2(s) 
 ds.\]
 This implies using Gronwall's Lemma that there exists $ C_1, C_2>0$ depend on   the quantities $ \Vert \chi \Vert_{L^2}  $, $ \Vert \sqrt{\omega} \  \chi \Vert_{L^2}  $ and $  \Vert \nabla_{{q}_j} V\Vert_{L^\infty} $ such that 
 \[\Theta_2(t)\leq C_1 \  \langle  \psi , \ (\hat{H}_0+\hat{p}^2+1) \ \psi \rangle \ e^{C_2 |t|} .\]
 And thus the result follows.
\end{proof}
\subsubsection{Field operator's  estimates}\label{top4momentestimate}
Below, we give some estimates for the field operator $d\Gamma(\omega^{2\sigma})$.

\noindent Let $\Gamma_{f{in}}$ be a dense subspace in the Fock space. Let $ \psi \in D(\hat{H}_0)$ and  define 
\[\Theta_3(t):= \langle \psi(t), \ d\Gamma (\omega^{2\sigma}) \ e^{-\delta d\Gamma(\omega^{2\sigma})} \ \psi(t)  \rangle, \qquad \psi(t):=e^{-i\frac{t}{\hbar} \hat{H}} \ \psi .\] 
Note that   $d\Gamma (\omega^{2\sigma}) \ e^{-\delta d\Gamma(\omega^{2\sigma})}  $ is a bounded and positive approximation of $d\Gamma (\omega^{2\sigma})  $  that strongly converges monotonically to it. The quantity $\Theta_3(t)$ is well-defined for each $t \in \R$ and $\delta>0$. In addition, the map $ t \rightarrow \Theta_3(t) $ is differentiable with 
\[ \dot\Theta_3(t)= \frac{i}{\hbar } \ \langle \psi(t), \  [\hat{H}, \ d\Gamma(\omega^{2\sigma})\ e^{-\delta d\Gamma(\omega^{2\sigma})}  ]  \ \psi(t)\rangle .\]  
\begin{lemma}\label{top4tocon}
Assume \eqref{top4A0} and  $\omega^{\sigma-\frac{1}{2}}\chi\in L^2(\R^d,dk)$. For $\sigma \in[1/2,1] $, there exists $C>0$ such that for all $\delta>0$, for all $\hbar \in (0,1)$ and for all $\phi, \ \psi \in \cC_0^\infty(\R^{dn}) \otimes \Gamma_{fin}$:
\[ \Big\vert \frac{i}{\hbar} \langle \phi, \ [\hat{H},  d\Gamma(\omega^{2\sigma})  \ e^{-\delta d\Gamma(\omega^{2\sigma})}] \ \psi \rangle \Big \vert \leq C \ \Vert \omega^{\sigma-\frac{1}{2}} \chi \Vert_{L^{2}} \Big[ \Vert \phi \Vert \ \Vert  d\Gamma(\omega^{2\sigma})^{\frac{1}{2}}\psi\Vert+  \Vert \psi \Vert \ \Vert  d\Gamma(\omega^{2\sigma})^{\frac{1}{2}}\phi\Vert\Big] . \]
\end{lemma}
\begin{proof}
Let us deal first with the term 
\[
\begin{aligned}
  \big[ d\Gamma(\omega^{2\sigma})  \ e^{-\delta d\Gamma(\omega^{2\sigma})},\hat{H}\big]&= [  d\Gamma(\omega^{2\sigma})  \ e^{-\delta d\Gamma(\omega^{2\sigma})},\hat{H}_1]
  \\ &=[  d\Gamma(\omega^{2\sigma})  \ e^{-\delta d\Gamma(\omega^{2\sigma})},\hat{a}_\hbar(G)+\hat{a}^*_\hbar(G)]
\\ & =  d\Gamma(\omega^{2\sigma})  \ [  e^{-\delta d\Gamma(\omega^{2\sigma})},\hat{a}_\hbar(G)] + [d\Gamma(\omega^{2\sigma})  ,\hat{a}_\hbar(G)] \   e^{-\delta d\Gamma(\omega^{2\sigma})}
\\ &  +d\Gamma(\omega^{2\sigma})  \ [  e^{-\delta d\Gamma(\omega^{2\sigma})},\hat{a}^*_\hbar(G)]+ [d\Gamma(\omega^{2\sigma})  ,\hat{a}^*_\hbar(G)] \   e^{-\delta d\Gamma(\omega^{2\sigma})}.
\end{aligned} 
\]
We also have 
\begin{itemize}
\item [(i)]  $ [d\Gamma(\omega^{2\sigma})  ,\hat{a}_\hbar(G)]=-\hbar \ \hat{a}_\hbar(\omega^{2\sigma} G)$,
\vskip 2mm
\item [(ii)] $[d\Gamma(\omega^{2\sigma})  ,\hat{a}^*_\hbar(G)]=\hbar \ \hat{a}^*_\hbar(\omega^{2\sigma} G) $,
\vskip 2mm
\item [(iii)] $[  e^{-\delta d\Gamma(\omega^{2\sigma})},\hat{a}_\hbar(G)]= e^{-\delta d\Gamma(\omega^{2\sigma})} \ \hat{a}_\hbar(\beta  \ G), \quad \beta= 1-e^{-\delta \ \hbar  \ \omega^{2\sigma}},$
\vskip 2mm
\item [(iv)] $[  e^{-\delta d\Gamma(\omega^{2\sigma})},\hat{a}^*_\hbar(G)]=\hat{a}^*_\hbar(-\beta  \ G) \  e^{-\delta d\Gamma(\omega^{2\sigma})} , \quad \beta= 1-e^{-\delta \ \hbar  \ \omega^{2\sigma}}.$
\end{itemize}
Using (i)-(ii)-(iii) and (iv), we get 
\[\big[ d\Gamma(\omega^{2\sigma})  \ e^{-\delta d\Gamma(\omega^{2\sigma})},\hat{H}\big]=\hbar \Big( B_1+B_2+B_3 \Big), \]
where we have introduced the three terms $B_1$, $B_2$ and $ B_3$ as follows
\begin{align*}
& B_1:=   \big[ \hat{a}^*_\hbar(\omega^{2\sigma} G)-\hat{a}_\hbar(\omega^{2\sigma} G) \big] \  e^{-\delta d\Gamma(\omega^{2\sigma})},
 \\ & B_2:=  d\Gamma(\omega^{2\sigma}) \  e^{-\delta d\Gamma(\omega^{2\sigma})} \  \hat{a}_\hbar\big(\frac{\beta \ G}{\hbar}\big),
 \\ & B_3:=  d\Gamma(\omega^{2\sigma}) \ \hat{a}^*_\hbar\big(\frac{- \beta \ G}{\hbar}\big) \   e^{-\delta d\Gamma(\omega^{2\sigma})} .
\end{align*}
We get then 
\[\begin{aligned}
\frac{i}{\hbar} \langle \phi, \ [\hat{H},  d\Gamma(\omega^{2\sigma})  \ e^{-\delta d\Gamma(\omega^{2\sigma})}] \ \psi \rangle &=-i \  \langle \phi,  \ (B_1+B_2+B_3) \  \psi \rangle
\\ & =\underbrace{-i \langle \phi,  \ B_1  \ \psi \rangle}_{(a)}+ \underbrace{-i \langle \phi,  \ B_2 \  \psi \rangle}_{(b)}+ \underbrace{-i \langle \phi,  \ B_3 \  \psi \rangle}_{(c)}.
\end{aligned}\]
For (b), we have 
\[ \begin{aligned}
\vert \langle \phi,  \ B_2 \  \psi \rangle \vert&= \big \vert \langle \phi,  \ d\Gamma(\omega^{2\sigma}) \  e^{-\delta d\Gamma(\omega^{2\sigma})} \  \hat{a}_\hbar\big(\frac{\beta \ G}{\hbar}\big) \ \psi  \rangle \big \vert
\\ & =\big \vert \langle \phi,  \delta \ d\Gamma(\omega^{2\sigma}) \  e^{-\delta d\Gamma(\omega^{2\sigma})} \  \hat{a}_\hbar\big(\frac{\beta \ G}{\delta \ \hbar}\big) \ \psi  \rangle \big \vert
\\ & \leq \Vert \phi \Vert  \ \Vert \delta  d\Gamma(\omega^{2\sigma}) \  e^{-\delta d\Gamma(\omega^{2\sigma})}   \ \hat{a}_\hbar\big(\frac{\beta \ G}{\delta \ \hbar}\big) \  \psi \Vert  
\\ & \leq \frac{1}{e} \  \Vert \phi \Vert  \ \Vert  \hat{a}_\hbar\big(\frac{\beta \ G}{\delta \ \hbar}\big) \  \psi \Vert,
\end{aligned}
\]
where in the last line we have used the fact that $ \sup_{\delta>0}  \Vert \delta  d\Gamma(\omega^{2\sigma}) \  e^{-\delta d\Gamma(\omega^{2\sigma})}\Vert \leq 1/e.$
Remark also that we have with $K_m$ and $X_n$ as in \eqref{top4notation} 
\[ \begin{aligned}
&\Vert \hat{a}_\hbar\big(\frac{\beta \ G}{\delta \ \hbar}\big) \  \psi  \Vert_{\cH}^2 =  \sum_{m\geq 0} \int_{\R^{dn}} \int_{\R^{dm}} \Big\vert \big[ \hat{a}_\hbar\big(\frac{\beta \ G}{\delta \ \hbar}\big) \ \psi \big]^m (X_n,K_m)\Big\vert^2  dK_m \ dX_n
\\ & =\sum_{m\geq 0}\int_{\R^{dn}} \int_{\R^{dm}} \Big\vert \int_{\R^d} \sum_{j=1}^{n}\sqrt{\hbar(m+1)}\ \frac{1-e^{-\delta \ \hbar \ \omega^{2\sigma}}}{\hbar \ \delta}  \ \frac{\chi(k)}{\sqrt{\omega(k)}} \ e^{2\pi  i k\cdot \hat{q}_j} \ \psi^{m+1}(X_n,K_m,k) dk\Big\vert^2 dK_m  \ dX_n
\\ & \leq \sum_{m\geq 0}\int_{\R^{dn}} \int_{\R^{dm}} \Big\vert \int_{\R^d} \sum_{j=1}^{n} \sqrt{\hbar(m+1)}\ \omega^{2\sigma}(k)  \ \frac{\chi(k)}{\sqrt{\omega(k)}} \ e^{2\pi  i k\cdot \hat{q}_j} \ \psi^{m+1}(X_n,K_m,k) dk\Big\vert^2 dK_m  \ dX_n
\end{aligned}
\]

\[\begin{aligned}
 & = \sum_{m\geq 0}\int_{\R^{dn}} \int_{\R^{dm}} \Big\vert \int_{\R^d} \sum_{j=1}^{n} \omega^{\sigma-\frac{1}{2}}(k)  \ {\chi(k)} \ \ \sqrt{\hbar(m+1)}\ \omega^{\sigma}(k) \ e^{2\pi  i k\cdot \hat{q}_j} \ \psi^{m+1}(X_n,K_m,k) dk\Big\vert^2 dK_m  \ dX_n
\\ & \underset{c-s}{\leq} n^2  \ \Vert \omega^{\sigma-\frac{1}{2}} \ \chi \Vert_2^2 \   \sum_{m\geq 0} \int_{\R^{dn}}\int_{\R^{dm}} \Big[ \int_{\R^d}  {\hbar(m+1)}\ \omega^{2 \sigma}(k) \ \Big\vert \psi^{m+1}(X_n,K_m,k)\Big\vert^2 dk \Big]   dK_m \ dX_n 
\\& \leq n^2   \  \Vert \omega^{\sigma-\frac{1}{2}}\ \chi \Vert_2^2 \ \Vert  d\Gamma(\omega^{2\sigma})^{1/2} \psi \Vert_\cH^2.
\end{aligned}
\]
Using the above estimates, we find that 
\[ \begin{aligned}
\vert \langle \phi,  \ B_2 \psi \rangle \vert \leq \frac{n}{e} \ \Vert \omega^{\sigma-\frac{1}{2}}\ \chi \Vert_2 \ \Vert \phi \Vert \ \Vert  d\Gamma(\omega^{2\sigma})^{1/2} \psi \Vert_\cH. 
\end{aligned}
\]
For (c), remark first that  we have 
\[ [ d\Gamma(\omega^{2\sigma}), \ \hat a^*_\hbar\big(\frac{-\beta \ G}{ \hbar}\big) ] =\hbar \ \hat{a}^*_\hbar\big(\frac{- \omega^{2\sigma} \ \beta \ G}{ \hbar}\big),\]
and    
\[[ d\Gamma(\omega^{2\sigma}), \ \hat a^*_\hbar\big(\frac{-\beta \ G}{ \hbar}\big) ] =d\Gamma(\omega^{2\sigma}) \ \hat a^*_\hbar\big(\frac{-\beta \ G}{ \hbar}\big)- \hat a^*_\hbar\big(\frac{-\beta \ G}{ \hbar}\big)  \ d\Gamma(\omega^{2\sigma}) \]
This implies that 
\[d\Gamma(\omega^{2\sigma}) \ \hat a^*_\hbar\big(\frac{-\beta \ G}{ \hbar}\big)= \hat{a}^*_\hbar\big(- \omega^{2\sigma} \ \beta \ G\big)+ \hat a^*_\hbar\big(\frac{-\beta \ G}{ \hbar}\big)  \ d\Gamma(\omega^{2\sigma}). \]
Then, we have 
\[\begin{aligned}
\vert \langle \phi, B_3 \ \psi \rangle \vert &  \leq \vert \langle\phi,  \hat{a}^*_\hbar\big(- \omega^{2\sigma} \ \beta \ G\big) \ e^{-\delta d\Gamma(\omega^{2\sigma})} \ \psi\rangle \vert +\vert \langle\phi, \ \hat a^*_\hbar\big(\frac{-\beta \ G}{ \hbar}\big)  \ d\Gamma(\omega^{2\sigma}) \ e^{-\delta d\Gamma(\omega^{2\sigma})} \ \psi\rangle \vert
\\ & \lesssim   \ \Vert \omega^{\sigma-\frac{1}{2}}\ \chi \Vert_2 \ \Vert \psi \Vert \ \Vert  d\Gamma(\omega^{2\sigma})^{1/2} \phi \Vert_\cH ,
\end{aligned} \]
where in the last line, we have used the same tricks as before as well as the fact that $\vert \beta \vert \leq 2$ and  $ e^{-\delta \  d\Gamma(\omega^{2\sigma })} \in \cL(\cH)$ with $\Vert e^{-\delta   \ d\Gamma(\omega^{2\sigma })}   \Vert\leq 1$.
Similarly for (a), we can have by same techniques that 
\[\begin{aligned}
\vert \langle \phi,  \hat{a}_\hbar^*(\omega^{2\sigma }   G) \ e^{- \delta d\Gamma(\omega^{2\sigma})} \ \psi \rangle \vert & =\vert \langle\hat{a}_\hbar(\omega^{2\sigma }   G)  \ \phi,  \ e^{- \delta d\Gamma(\omega^{2\sigma})} \ \psi \rangle \vert 
\\ & \lesssim   \ \Vert \omega^{\sigma-\frac{1}{2}}\ \chi \Vert_2 \ \Vert \psi \Vert \ \Vert  d\Gamma(\omega^{2\sigma})^{1/2} \phi \Vert_\cH .
\end{aligned} \]
For the other term in (a), note that 
 \[[ e^{- \delta d\Gamma(\omega^{2\sigma})},  \  \hat{a}_\hbar(\omega^{2\sigma }   G)]=e^{- \delta d\Gamma(\omega^{2\sigma})} \  \hat{a}_\hbar(\beta  \ \omega^{2\sigma }   G) .\]
This implies that 
\[\  \hat{a}_\hbar( \omega^{2\sigma }   G) \ e^{- \delta d\Gamma(\omega^{2\sigma})} =e^{- \delta d\Gamma(\omega^{2\sigma})} \  \hat{a}_\hbar( (1-\beta) \ \omega^{2\sigma }   G) .\]
Then, using the above equality, we get 
\[\begin{aligned}
\vert \langle \phi,  \hat{a}_\hbar(\omega^{2\sigma }   G) \ e^{- \delta d\Gamma(\omega^{2\sigma})} \ \psi \rangle \vert
 \lesssim   \ \Vert \omega^{\sigma-\frac{1}{2}}\ \chi \Vert_2 \ \Vert \phi \Vert \ \Vert  d\Gamma(\omega^{2\sigma})^{1/2} \psi \Vert_\cH .
\end{aligned} \]
We conclude that 
\[ \vert \langle \phi, B_1 \ \psi \rangle \vert \lesssim  \ \Vert \omega^{\sigma-\frac{1}{2}}\ \chi \Vert_2 \ \Big[  \Vert \psi \Vert \ \Vert  d\Gamma(\omega^{2\sigma})^{1/2} \phi \Vert_\cH+ \Vert \phi \Vert \ \Vert  d\Gamma(\omega^{2\sigma})^{1/2} \psi \Vert_\cH  \Big]. \] 
And thus, the final result follows.
\end{proof}
\begin{lemma}\label{top4gron}
 There exists $C_1, \ C_2 \in \R^*_+$ such that  
 \[\Theta_3(t)\leq C_1  \  \langle \psi , \ (d\Gamma(\omega^{2\sigma}) \ e^{-\delta d\Gamma(\omega^{2\sigma})}+1) \  \psi  \rangle \ e^{C_2|t|} . \]
\end{lemma}
\begin{proof}
Use the previous Lemma \ref{top4tocon} with $\phi=\psi=\psi(t)$, we get that 
\[
\begin{aligned}
\dot \Theta_3(t)&=\frac{i}{\hbar} \ \langle \psi(t), [ \hat{H}, \ d\Gamma(\omega^{2\sigma}) \ e^{-\delta d\Gamma(\omega^{2\sigma})}] \ \psi(t) \rangle
\\ & \lesssim  \Vert \omega^{\sigma-\frac{1}{2}}\ \chi \Vert_2 \ \Big[  \Vert \psi \Vert \ \Vert  d\Gamma(\omega^{2\sigma})^{1/2} \psi(t) \Vert_\cH  \Big]
\\ & \lesssim  \Vert \omega^{\sigma-\frac{1}{2}}\ \chi \Vert_2 \ \Big[  \Vert \psi \Vert \ \Vert  \big[ d\Gamma(\omega^{2\sigma})e^{-\delta d\Gamma(\omega^{2\sigma})}\big]^{1/2} \psi(t) \Vert_\cH  \Big]
\\ & \lesssim  \Vert \omega^{\sigma-\frac{1}{2}}\ \chi \Vert_2 \ \Big[  \Vert \psi \Vert^2+ \Vert  \big[ d\Gamma(\omega^{2\sigma})e^{-\delta d\Gamma(\omega^{2\sigma})}\big]^{1/2} \psi(t) \Vert_\cH^2  \Big]
\\ & \lesssim  c\langle \psi, \ \psi\rangle + c \ \Theta_3(t).
\end{aligned}\]
And thus the result follows by applying the Gronwall's Lemma.
\end{proof}
\begin{lemma}[Field estimate]\label{top4local2}Assume that \eqref{top4A0} and  $\omega^{\sigma-\frac{1}{2}}\chi\in L^2(\R^d,dk)$. Then, there exists constants $C_1,C_2>0$ such that for all $\psi \in D(\hat{H}_0^{1/2})\cap D(d\Gamma(\omega^{2\sigma})^{1/2})$, all $t\in \R$ and all $\hbar\in (0,1)$:
\begin{align}
  \langle e^{-i \frac{t}{\hbar}  \hat{H}}\psi, \ d\Gamma(\omega^{2\sigma})  \ e^{-i \frac{t}{\hbar}  \hat{H}}\psi\rangle \leq C_1\langle \psi,  \ (d\Gamma(\omega^{2\sigma})+1)  \ \psi \rangle \  e^{C_2|t|}.\label{top4first1bar}
\end{align}
\end{lemma}
\begin{proof}
It is a consequence of the previous Lemma \ref{top4gron}. Indeed, the approximation map \\  $e^{-\delta d\Gamma(\omega^{2\sigma})}  \ d\Gamma(\omega^{2\sigma})$ converges strongly to $ d\Gamma(\omega^{2\sigma})$. This leads to 
\[ \begin{aligned}
  \Vert  \big( d\Gamma(\omega^{2\sigma})\big)^{1/2} \  \psi(t) \Vert^2 & = \lim_{\delta\rightarrow 0} \Vert \big( e^{-\delta d\Gamma(\omega^{2\sigma})}  \ d\Gamma(\omega^{2\sigma})\big)^{1/2} \  \psi(t) \Vert^2
\\ &  \lesssim  \lim_{\delta \rightarrow 0} C_1 \ \Vert (e^{-\delta d\Gamma(\omega^{2\sigma})}  \ d\Gamma(\omega^{2\sigma})+1)^{1/2} \  \psi \Vert^2 \ e^{C_2|t|}\\ &= C_1 \ \Vert ( d\Gamma(\omega^{2\sigma})+1)^{1/2} \  \psi \Vert^2 \ e^{C_2|t|} .
\end{aligned} \]
And thus, we achieve the desired result.
\end{proof}
\subsubsection{Propagation of estimates uniformly for all times}\label{top4propaaa}
As a consequence of the previous estimates, the uniform bound on the initial states $(\varrho_\hbar)_{\hbar\in(0,1)}$  propagates in time.
\begin{lemma}[Propagation of  the assumptions \eqref{top4S0} and \eqref{top4S1} in time] \label{top4uniformestimate} Assume  \eqref{top4A0} and  $\omega^{{1}/{2}}\chi\in L^2(\R^d,dk)$. Let $(\varrho_\hbar)_{\hbar\in(0,1)}$  be a family of density matrices satisfying \eqref{top4S0} and \eqref{top4S1}. Then, the family of states $(\varrho_h(t))_{\hbar \in (0,1)}$ and $(\tilde{\varrho}_h(t))_{\hbar \in (0,1)}$ satisfy the same assumptions
 \eqref{top4S0} and \eqref{top4S1} uniformly for any $t\in \R$ in arbitrary compact interval.  
\end{lemma}
\begin{proof}
Before we begin the proof, remark that by spectral decomposition, we have 
\[\varrho_\hbar=\sum_{m\in \N}  \lambda_{\hbar}(m) \  \vert e^{-i \frac{t}{\hbar} \hat H} \psi_\hbar(m)\rangle \langle  e^{-i \frac{t}{\hbar} \hat H} \psi_\hbar(m) \vert ,\]
where $\lambda_\hbar (m)$ are the eigenvalues and $ \psi_{\hbar}(m)$ are their related eigenfunctions.
 Let $J$ be a compact interval. Then for all $t\in J$:
\begin{itemize}
\item We have with some $c \in \R^*_+$ the following uniform estimate\[ \begin{aligned}
{\rm Tr}\big[ \tilde \varrho_\hbar(t) \ \hat{p}^2\big]&= {\rm Tr}\Big[ e^{i \frac{t}{\hbar} \hat H_{02}} \  \varrho_\hbar(t)  \ e^{-i \frac{t}{\hbar} \hat H_{02}}\ \hat{p}^2\Big]
\\ & ={\rm Tr}\big[ \varrho_\hbar(t) \ \hat{p}^2\big]
\\ & = \sum_{m\in \N}  \lambda_{\hbar}(m) \  \Vert \hat{p} \ e^{-i \frac{t}{\hbar} \hat H} \psi_\hbar(m) \Vert^2
\\ &=  \sum_{m\in \N}  \lambda_{\hbar}(m) \ \langle e^{-i \frac{t}{\hbar} \hat H} \psi_\hbar(m), \ \hat{p}^2 \ e^{-i \frac{t}{\hbar} \hat H} \psi_\hbar(m)  \rangle
\\  & \leq C_1 \  \sum_{m\in \N}  \lambda_{\hbar}(m) \ \langle  \psi_\hbar(m), \ (\hat{H}_0+\hat{p}^2+1) \  \psi_\hbar(m)  \rangle \ e^{C_2 |t|}
\leq c,
\end{aligned}\]
 where we have used Lemma \ref{top4local1} as well as assumptions \eqref{top4S0} and \eqref{top4S1}. 
\item We have with some $c^\prime\in \R^*_+$ the following estimate
\[ \begin{aligned}
{\rm Tr}\big[ \tilde \varrho_\hbar(t) \ \hat{q}^2\big]&= {\rm Tr}\Big[ e^{i \frac{t}{\hbar} \hat H_{02}} \  \varrho_\hbar(t)  \ e^{-i \frac{t}{\hbar} \hat H_{02}}\ \hat{q}^2\Big]
\\ &= {\rm Tr}\big[ \varrho_\hbar(t) \ \hat{q}^2\big]
\\ & = \sum_{m\in \N}  \lambda_{\hbar}(m) \  \Vert \hat{q} \ e^{-i \frac{t}{\hbar} \hat H} \psi_\hbar(m) \Vert^2
\\ &=  \sum_{m\in \N}  \lambda_{\hbar}(m) \ \langle e^{-i \frac{t}{\hbar} \hat H} \psi_\hbar(m), \ \hat{q}^2 \ e^{-i \frac{t}{\hbar} \hat H} \psi_\hbar(m)  \rangle
\\  & \leq C_1 \  \sum_{m\in \N}  \lambda_{\hbar}(m) \ \langle  \psi_\hbar(m), \ (\hat{H}_0+\hat{q}^2+1) \  \psi_\hbar(m)  \rangle \ e^{C_2 |t|}
 \leq c^\prime,
\end{aligned}\]
where we have used Lemma \ref{top4local} as well as assumptions \eqref{top4S0} and \eqref{top4S1}. 
\item We have for some $c^{\prime\prime}$ the following uniform estimate
\[ \begin{aligned}
{\rm Tr}\big[\tilde \varrho_\hbar(t) \ d\Gamma(\omega^{2\sigma})\big]&= {\rm Tr}\Big[e^{i \frac{t}{\hbar} \hat H_{02}} \  \varrho_\hbar(t)  \ e^{-i \frac{t}{\hbar} \hat H_{02}}\ d\Gamma(\omega^{2\sigma})\Big]
\\ & ={\rm Tr}\big[ \varrho_\hbar(t) \ d\Gamma(\omega^{2\sigma})\big]
\\ & = \sum_{m\in \N}  \lambda_{\hbar}(m) \  \Vert d\Gamma(\omega^{2\sigma})^{1/2} \ e^{-i \frac{t}{\hbar} \hat H} \psi_\hbar(m) \Vert^2
\\ &=  \sum_{m\in \N}  \lambda_{\hbar}(m) \ \langle e^{-i \frac{t}{\hbar} \hat H} \psi_\hbar(m), \ d\Gamma(\omega^{2\sigma}) \ e^{-i \frac{t}{\hbar} H} \psi_\hbar(m)  \rangle
\\  & \leq C_1 \  \sum_{m\in \N}  \lambda_{\hbar}(m) \ \langle  \psi_\hbar(m), \ (d\Gamma(\omega^{2\sigma})+1) \  \psi_\hbar(m)  \rangle \ e^{C_2 |t|}
\leq c^{\prime \prime}.
\end{aligned}\]
where we have used Lemma \ref{top4local2} as well as assumptions \eqref{top4S0}. 
\end{itemize}
\end{proof}

\subsection{Existence of unique Wigner measure}\label{top4extraction}

In this section, we prove that for any family of states $(\varrho_h)_{\hbar \in (0,1)}$ which satisfies \eqref{top4S0} and \eqref{top4S1} and for any sequence $\hbar_n\rightarrow 0$, we can extract a subsequence $\hbar_{n_\ell}\rightarrow 0$ such that the set of Wigner measure 
\[ \cM(\tilde \varrho_{\hbar_{n_\ell}}(t), \ell \in \N)\]
is singleton. The main results are stated below: 

\begin{proposition}[Existence of unique  Wigner measure $\tilde{\mu}_t$ for all times] \label{top4uniquewignermeasure} Assume that \eqref{top4A0} and  \eqref{top4A1} hold true. Let $(\varrho_\hbar)_{\hbar\in (0,1)}$  be a family of density matrices satisfying \eqref{top4S0} and \eqref{top4S1}. For any sequence $(\hbar_n)_{n\in \N} $ in $(0,1) $ such that $\hbar_n \rightarrow 0 $, there exists a subsequence $(\hbar_{n_\ell})_{\ell \in \N}$  and a family of probability measures $(\tilde{\mu}_t)_{t\in \R}$ such that for all $t\in \R$,
\[\cM(\tilde{\varrho}_{\hbar_{n_\ell}}(t), \ell \in \N)=\lbrace \tilde\mu_t \rbrace. \]
Moreover, for every compact time interval $J$ there exists a constant $C>0$ such that for all times $t\in J$,
\be \label{top4boundd} \int_{X^0} \Vert u\Vert_{X^{\sigma}}^2 \ d\tilde{\mu}_t(u)<C.\ee
\end{proposition}
\begin{proof}
We prove the above proposition in two steps. Step 1 is dedicated to the extraction of a unique Wigner measure at fixed times. Step 2 generalizes for all times. 
To establish Step 1, it is necessary to recall the following result  from  \cite[Theorem 6.2]{ammari2008ahp}.
\begin{proposition}[The set of Wigner measure is not empty] \label{top4existwigner}
Let $({\varrho}_h)_{\hbar\in (0,1)}$ be a family of density matrices satisfying \eqref{top4S0} and \eqref{top4S1}. Then for all sequences $(\hbar_n)_{n\in \N}$ with $\displaystyle \lim_{n\rightarrow \infty} \hbar_n=0$, there exists a subsequence $(\hbar_{n_\ell})_{\ell \in \N} $ with $\displaystyle \lim_{\ell \rightarrow \infty} \hbar_{n_\ell}=0$ such that 
\[ \cM(\varrho_{\hbar_{n_\ell}}, \ell \in \N)=\lbrace \mu \rbrace.\]
Moreover, we have 
\be \label{top4integrabilityformulasfixedtimes} \int_{X^0} \Vert u \Vert_{X^0}^2 \ d\mu(u)<+\infty, \qquad \int_{X^0} \Vert u \Vert_{X^\sigma}^2 \ d\mu(u)<+\infty.\ee
\end{proposition} 
\noindent {\bf Step 1}: {\it Extraction of a unique Wigner measure at fixed times.}
\vskip 1mm
\noindent Let $s \in \R$ be a fixed time. Let $\underset{n\rightarrow \infty}{\hbar_n\rightarrow 0}$. Then, by Proposition \ref{top4existwigner}, there exists a subsequence $(\hbar_\ell)_{\ell \in \N}\equiv(\hbar_{n_\ell})_{\ell \in \N}  $ such that  $\underset{\ell \rightarrow \infty}{\hbar_{\ell}\rightarrow 0}$ and a probability measure $\tilde \mu_s\in \cP(X^0)$ such that 
\[\cM(\tilde \varrho_{\hbar_{n_\ell}}(s), \ell \in \N)=\lbrace \tilde \mu_s \rbrace. \]
Moreover,  we have the following integrability formula  
\be \label{est2} \int_{X^0} \Vert u \Vert_{X^\sigma}^2 \ d\tilde \mu_s(u)= \int_{X^0} (p^2+q^2+\Vert \alpha \Vert_{\cG^\sigma}^2) \ d\tilde{\mu}_s(p,q,\alpha)<+\infty.\ee
The above integrability formula is a consequence of the following implications proved in  \cite[Lemma 3.12]{ammari:hal-00644656} for some $C>0$:
\begin{itemize}
\item [(i)]  If ${\rm Tr} [\varrho_{\hbar} \  \hat{N}_\hbar]\leq C \Longrightarrow \forall \mu \in \cM(\varrho_\hbar;\hbar\in(0,1)), \ \int_{X^0} \Vert \alpha\Vert_{\cG^0}^2 \ d\mu \leq C;$
\item [(ii)]  If ${\rm Tr} [\varrho_{\hbar} \  d\Gamma(\omega^{2\sigma})]\leq C \Longrightarrow \forall \mu \in \cM(\varrho_\hbar;\hbar\in(0,1)), \ \int_{X^0} \Vert \alpha\Vert_{\cG^\sigma}^2 \ d\mu \leq C;$
\item [(iii)]  If ${\rm Tr} [\varrho_{\hbar} \ ( \hat{q}^2+\hat{p}^2)]\leq C \Longrightarrow \forall \mu \in \cM(\varrho_\hbar;\hbar\in(0,1)), \ \int_{X^0} (q^2 +p^2)\ d\mu \leq C.$
\end{itemize}
Now,  by the help of uniform estimate in Lemma \ref{top4uniformestimate}, we have the two  family of states $ ({\varrho}_h(t))_{\hbar \in (0,1)}$ and $ (\tilde{\varrho}_h(t))_{\hbar \in (0,1)}$ satisfy uniformly the  bounds of (i)-(ii)-(iii), one obtains then  that \eqref{est2} holds true  as a consequence of  \eqref{top4boundd} in  Proposition \ref{top4existwigner}.
\vskip 2mm

\noindent {\bf Step 2:} {\it Generalization for all times.}
\vskip 1mm
\noindent Claim first that we have for all times $t\in \R$
\be \label{top4uniquesetwignermeasure} \cM(\tilde \varrho_{\hbar_{\ell}}(t), \ell \in \N)=\lbrace \tilde \mu_t \rbrace. \ee
Let us now prove the integrability formula in Proposition \ref{top4uniquewignermeasure}. Recall that our density metrices $(\varrho_\hbar)_{\hbar \in (0,1)}$ satisfies the assumptions \eqref{top4S0} and \eqref{top4S1}. And, by Lemma \ref{top4uniformestimate}, the  family of states $(\tilde \varrho_\hbar(t))_{\hbar \in (0,1)}$ satisfies the same assumptions uniformly in any compact time interval $J$. Then, using \eqref{top4integrabilityformulasfixedtimes},  for all $t\in J$, we have \eqref{top4boundd}. we come back now to prove the claim \eqref{top4uniquesetwignermeasure}. Let $(t_j)_{j\in \N}$ be a countable dense set in $\R$.  We have by Step 1 that 
\begin{itemize}
\item for $t_1$, for $ \underset{n\rightarrow \infty}{\hbar_n\longrightarrow 0} $, there exists a subsequence $ \underset{\ell\rightarrow \infty}{\hbar_\ell \longrightarrow 0}  $ such that 
\[\cM(\tilde \varrho_{\hbar_{\ell}}(t_1), \ell \in \N)=\lbrace \tilde \mu_{t_1} \rbrace. \]
\item for $t_2$, for $(\hbar_\ell)_{\ell \in \N}$, there exists a subsequence $ (\hbar_{\phi_2(\ell)})_{{\ell \in \N}}\subset (\hbar_\ell)_{\ell \in \N}$ such that 
\[\cM(\tilde \varrho_{\hbar_{\phi_2(\ell)}}(t_2), \ell \in \N)=\lbrace \tilde \mu_{t_2} \rbrace. \]
\item for $t_3$, for $ (\hbar_{\phi_2(\ell)})_{\ell \in \N}  $, there exists a subsequence $ (\hbar_{\phi_3(\ell)})_{\ell \in \N} \subset (\hbar_{\phi_2(\ell)})_{\ell \in \N}$ such that 
\[\cM(\tilde \varrho_{\hbar_{\phi_3(\ell)}}(t_3), \ell \in \N)=\lbrace \tilde \mu_{t_3} \rbrace. \]
\item And so on,    for $t_j$, for $ (\hbar_{\phi_{j-1}(\ell)})_{\ell \in \N}  $, there exists a subsequence $ (\hbar_{\phi_j(\ell)})_{\ell \in \N} \subset (\hbar_{\phi_{j-1}(\ell)})_{\ell \in \N}$ such that 
\[\cM(\tilde \varrho_{\hbar_{\phi_j(\ell)}}(t_j), \ell \in \N)=\lbrace \tilde \mu_{t_j} \rbrace. \]
\end{itemize}
By diagonal arguments, we extract the subsequence  $ (\hbar_{\phi_\ell(\ell)})_\ell$ denoted by $(\hbar_\ell)_\ell$  for simplicity  such that for all $j\in \N$, we have 
\[\cM(\tilde \varrho_{\hbar_{\ell}}(t_j), \ell \in \N)=\lbrace \tilde \mu_{t_j} \rbrace. \]
The above formula implies that for all $\xi=(p_0,q_0,\alpha_0) \in X^0$ and $\tilde \xi=(-2\pi q_0,2\pi p_0, \sqrt{2} \pi \alpha_0)$
\be \label{top4fixedtimelimit} \lim_{\ell \rightarrow \infty} {\rm Tr} \big[\cW(\tilde{\xi})   \ \tilde \varrho_{\hbar_{\ell}}(t_j) \big]=\int_{X^0} e^{2\pi i \Re e \langle  \xi, u\rangle_{X^0}  } \ d\tilde\mu_{t_j}(u)\ee
We have 
\[ \int_{X^0} \Vert u \Vert_{X^\sigma}^2 \ d\tilde \mu_{t_j}(u) <+\infty.\]
The above formula implies that the set of Wigner measure $\{ \tilde{\mu}_{t_j}\}_{j\in \N}$ is  tight in $\cP(X^0)$. This implies that according to   Prokhorov's theorem in Lemma \ref{prokhorov}, for all $\xi \in X^0$, there exists a subsequence still denoted by $t_j$ and a  probability measure $\tilde \mu_t\in \cP(X^0)$ such that $\tilde \mu_{t_j}$ converges weakly narrowly to $\tilde \mu_t$. This gives since the function $ e^{2\pi i \Re e \langle  \xi, u\rangle_{X^0} }$ is bounded that  
\[\int_{X^0} e^{2\pi i \Re e \langle  \xi, u\rangle_{X^0}  } \ d\tilde\mu_{t_j}(u) \underset{t_j \rightarrow t}{\longrightarrow} \int_{X^0} e^{2\pi i \Re e \langle  \xi, u\rangle_{X^0}  } \ d\tilde\mu_{t}(u) \]
Now, we need to prove that 
\be \label{top4needtoprove} 
 \int_{X^0} e^{2\pi i \Re e \langle  \xi, u\rangle_{X^0}  } \ d\tilde\mu_{t_j} (u) \underset{t_j \rightarrow t}{\longrightarrow} \lim_{\ell \rightarrow \infty} {\rm Tr} \big[  \cW(\tilde{\xi}) \ \tilde \varrho_{\hbar_{\ell}}(t)  \big]   
\ee
We start by 
\be \label{firstdecom}
\begin{aligned}
 &\Big \vert \int_{X^0} e^{2\pi i \Re e \langle  \xi, u\rangle_{X^0}  } \ d\tilde\mu_{t_j}(u)- \lim_{\ell \rightarrow \infty} {\rm Tr} \big[ \cW(\tilde{\xi})\ \tilde \varrho_{\hbar_{\ell}}(t) \big]  \Big \vert 
\\ & \leq \Big \vert \int_{X^0} e^{2\pi i \Re e \langle  \xi, u\rangle_{X^0}  } \ d\tilde\mu_{t_j}(u)- \lim_{\ell \rightarrow \infty} {\rm Tr} \big[  \cW(\tilde{\xi})\ \tilde \varrho_{\hbar_{\ell}}(t_j) \big]  \Big \vert \qquad \cdots  \quad (1) 
\\ & \qquad +\Big \vert \lim_{\ell \rightarrow \infty} {\rm Tr} \big[ \cW(\tilde{\xi}) \  \tilde \varrho_{\hbar_{\ell}}(t_j) \big]- \lim_{\ell \rightarrow \infty} {\rm Tr} \big[\cW(\tilde{\xi}) \ \tilde \varrho_{\hbar_{\ell}}(t) \big]  \Big \vert \qquad \cdots  \quad (2) 
\end{aligned}
\ee
The quantity (1)  is zero by \eqref{top4fixedtimelimit}. The quantity (2) is zero by using the following estimates:
\begin{itemize}
\item [(i)] For $\xi \in X^0$, for all $t,t_0\in J$ where $J$ is compact interval, we have 
\[ \Big \vert {\rm Tr}\Big[ \cW(\xi) \ \big(\tilde \varrho_\hbar(t)-\tilde \varrho_\hbar(t_0) \big) \Big]  \Big \vert \lesssim \vert t-t_0 \vert \ \Vert \xi \Vert_{X^0} \Big[ \Vert \xi \Vert_{X^0}+\Vert \chi \Vert_{L^2}+\Vert \chi \Vert_{\cG^0} \Big]; \]
\vskip 2mm
\item [(ii)] For all $\xi_1,\xi_2\in X^0$, for all $t\in J$ where $J$ is compact interval, we have 
\[\Big \vert {\rm Tr}\Big[ \big(\cW(\xi_1)-\cW(\xi_2)\big) \ \tilde \varrho_\hbar(t) \Big]  \Big \vert \lesssim  \ \Vert \xi_1-\xi_2 \Vert_{X^0} \Big[ \Vert \xi_1 \Vert_{X^0}+\Vert \xi_2 \Vert_{X^0} +1\Big].\]
\end{itemize}
For (i), we exploit \eqref{top4duhamelformula}, we have with $S=(\hat{H}_{0}+1)^{1/2}$ 
\[\Big \vert {\rm Tr}\Big[ \cW(\xi) \ \big(\tilde \varrho_\hbar(t)-\tilde \varrho_\hbar(t_0) \big) \Big]  \Big \vert \leq |t-t_0| \  \Big \Vert ({\rm B}_0+\hbar{\rm B}_1) \ S^{-1} \Big \Vert_{\cL(\cH)} \Big \Vert S \cW(\xi) S^{-1} \Big \Vert_{\cL(\cH)}   \Big \Vert S \tilde \varrho_\hbar (t)\Big \Vert_{\cL^1(\cH)}. \]
Now using Lemma \ref{top4weylheisenberg}, the two estimates \eqref{top4estimate1} and \eqref{top4estimate2} and the two assumptions \eqref{top4S0} and \eqref{top4S1}, we get the desired result.
\vskip 1mm
\noindent For (ii), we have 
\[\Big \vert {\rm Tr}\Big[ \big(\cW(\xi_1)-\cW(\xi_2)\big) \ \tilde \varrho_\hbar(t) \Big]  \Big \vert \leq \Big \Vert ( \cW(\xi_1)-\cW(\xi_2)) (\hat N_\hbar+1)^{-1/2} \Big \Vert_{\cL(\cH)}  \underbrace{ \Big \Vert (\hat N_\hbar+1)^{1/2} \tilde \varrho_\hbar (t)\Big \Vert_{\cL^1(\cH)}}_{<\infty \ \text{by } \eqref{top4S0}}.\]
And thus, using following the same computations as in \cite[Lemma 3.1]{ammari2008ahp}, the result follows.
\end{proof}

\section{Derivation of the characteristic equations}
 Subsection \ref{top4convv} focuses on investigating the convergence of the quantum dynamics towards the evolution of the particle-field equation. In Subsection \ref{top4wignerchara}, we derive the characteristic equation that the Wigner measure satisfies. Finally, in Subsection \ref{top4liouville}, we demonstrate that this characteristic equation is equivalent to a Liouville equation.
\subsection{Convergence}\label{top4convv}
In this section, we take the classical limit $\hbar_{n_\ell} \rightarrow 0$ as $ \ell \rightarrow \infty$ in the Duhamel formula \eqref{top4duhamelformula} to derive the characteristics equation satisfied by the Wigner measure $(\tilde{\mu}_t)_{t\in \R}$.
\begin{lemma}[Convergence]\label{top4conver}
Assume \eqref{top4A0} and $ \omega^{{1}/{2}} \chi \in L^2(\R^d,dk)$. Let $(\varrho_\hbar)_{\hbar \in (0,1)}$ be a family of density matrices satisfying \eqref{top4S0} and  \eqref{top4S1}. Then for all $\xi=(z_0,\alpha_0)  \in X^{0}$ and all $t,t_0\in \R,$, the Duhamel formula \eqref{top4duhamelformula} converges to the following characteristics equation 
\be \label{characfirstform}\int_{X^0}e^{Q(\xi,u)}  \, d\tilde{\mu}_t(u)=\int_{X^0}e^{Q(\xi,u)}  \, d\tilde{\mu}_{t_0}(u)-i\int_{t_0}^{t}  \int_{X^0}b(s,\xi) \, e^{Q(\xi,u)}   \, d\tilde{\mu}_s(u)\, ds \ee
with 
\begin{equation}\label{chap4.eq.Q.phase}
Q(\xi,u)=i \Im m \langle z,z_0 \rangle+ \sqrt{2} i \Re e\langle \alpha_0,\alpha \rangle,\qquad \xi=(z_0,\alpha_0),\quad u=(z,\alpha),
\end{equation}
and where we have introduced 
\[ 
\begin{aligned}
b(s,\xi)&:=-\sum_{j=1}^n \nabla f_j({p}_j)\cdot p_{0j}
 -\sum_{j=1}^n \nabla_{q_j} V( q)\cdot q_{0j}
\\ & \quad +\sum_{j=1}^{n} \big( \langle \alpha , b_j^0(s)\rangle_{L^2} + \langle  b_j^0(s),\alpha \rangle_{L^2} \big) 
+\frac{i  }{\sqrt{2}} \big(\langle \alpha_0,g_j(s)\rangle_{L^2} -\langle  g_j(s),\alpha_0 \rangle_{L^2}\big).
\end{aligned}
\]
The function $b_j^0(s)$ is such that $b_j^0(s)\equiv b_j^0(s)({p}_j,{q}_j)$ is defined as follows
\be \label{top4bj0} b_j^0(s):=2\pi i k \cdot q_{0j} \ \frac{\chi(k)}{\sqrt{\omega(k)}} \ e^{-2\pi i k \cdot q_j+is\omega(k)} \ee
\end{lemma}
\begin{proof}
From the Definition \ref{top4definitionwigner} of Wigner measure, we have 
\[\lim_{\ell \rightarrow+\infty} {\rm Tr}[\cW(\xi) \ \tilde{\varrho}_{\hbar_{n_{\ell}}}(s)]=\int_{X^0}e^{Q(\xi,u)}  \, d\tilde{\mu}_s(u). \]
We plug \eqref{top4comexp} in the Duhamel's formula \eqref{top4duhamelformula}, we get
\[\int_{X^0}e^{Q(\xi,u)}  \, d\tilde{\mu}_t(u)=\int_{X^0}e^{Q(\xi,u)}  \, d\tilde{\mu}_{t_0}(u)-i\int_{t_0}^{t} \lim_{\ell \rightarrow\infty} {\rm Tr}[{\rm B}_0(s,\hbar_{n_\ell},\xi) \ \cW(\xi) \ \tilde{\varrho}_{\hbar_{n_{\ell}}}(s)] \, ds. \]
We have to prove then 
\[\lim_{\ell \rightarrow+\infty} {\rm Tr}[{\rm B}_0(s,\hbar_{n_\ell},\xi) \ \cW(\xi) \ \tilde{\varrho}_{\hbar_{n_{\ell }}}(s)]= \int_{X^0}b(s,\xi) \, e^{Q(\xi,u)}   \, d\tilde{\mu}_s(u).\]
We start with 
\[
\begin{aligned}
&{\rm Tr}[{\rm B}_0(s,\hbar_{n_\ell},\xi) \ \cW(\xi) \ \tilde{\varrho}_{\hbar_{n_{\ell }}}(s)]
\\&=-\sum_{j=1}^n {\rm Tr}[\nabla f_j(\hat{p}_j)\cdot p_{0j} \ \cW(\xi) \ \tilde{\varrho}_{\hbar_{n_{\ell}}}(s)]
  -\sum_{j=1}^{n}{\rm Tr}[\nabla_{q_j} V({\hat q}) \cdot q_{0j}\ \cW(\xi) \ \tilde{\varrho}_{\hbar_{n_{\ell}}}(s)]
\\ & \quad  +\sum_{j=1}^{n} {\rm Tr}[\hat a_{\hbar_{n_\ell}}\big(\frac{\tilde{g}_j(s)-g_j(s)}{\hbar_{n_\ell}}\big)\ \cW(\xi) \ \tilde{\varrho}_{\hbar_{n_\ell}}(s)]+{\rm Tr}[ \hat a^*_{\hbar_{n_\ell}}\big(\frac{\tilde{g}_j(s)-g_j(s)}{{\hbar_{n_\ell}}}\big)\ \cW(\xi) \ \tilde{\varrho}_{\hbar_{n_\ell}}(s)] 
\\ & \quad +\sum_{j=1}^{n}\frac{i  }{\sqrt{2}} \Big({\rm Tr}[\langle \alpha_0,\tilde g_j(s)\rangle_{L^2(\R^d,\C)}\ \cW(\xi) \ \tilde{\varrho}_{\hbar_{n_\ell}}(s)] -{\rm Tr}[\langle  \tilde g_j(s),\alpha_0 \rangle_{L^2(\R^d,\C)}\ \cW(\xi) \ \tilde{\varrho}_{\hbar_{n_\ell}}(s)]\Big).
\end{aligned}
\]
\noindent Let us start with the first two   terms. We have 
\begin{align}
& \lim_{\ell  \rightarrow \infty}{\rm Tr}[\nabla f_j(\hat p_j )\cdot  p_{0j} \ \cW(\xi) \ \tilde{\varrho}_{\hbar_{n_{\ell }}}(s)]=
\int_{X^0} e^{Q(\xi,u)} \  \nabla f_j( p_j )\cdot  p_{0j} \, d\tilde{\mu}_s(u),
\\&\lim_{\ell  \rightarrow \infty}{\rm Tr}[\nabla_{q_j} V(\hat q )\cdot  q_{0j} \ \cW(\xi) \ \tilde{\varrho}_{\hbar_{n_{\ell }}}(s)]=
\int_{X^0} e^{Q(\xi,u)} \  \nabla_{q_j} V( q )\cdot  q_{0j}   \, d\tilde{\mu}_s(u),
\end{align}
where we have used in the above  two lines the convergent results in \cite[Lemma B.1]{Z} since $\langle p_j \rangle ^{-1}\nabla f_j( p_j )\cdot  p_{0j}\in L^\infty  $ and $\langle q \rangle^{-1}  \nabla_{q_j} V( q )\cdot  q_{0j} \in L^\infty$.
Let us deal now with the second line.
 \noindent The goal is  to prove the following limit:
\[\lim_{\ell }{\rm Tr}\Big[ \hat a^*_{\hbar_{n_\ell}}\Big( \frac{\tilde{g}_j(s)-g_j(s)}{\hbar_{n_\ell}}  \Big)\ \cW(\xi) \ \tilde{\varrho}_{\hbar_{n_{\ell}}}(s)\Big]=\int_{X^0}   e^{Q(\xi,u)}  \ \langle\alpha, b_j^0(s) \rangle_{L^2} \, d\tilde{\mu}_s(u).\]
We start  then with 
\[\begin{aligned}
&\Big \vert {\rm Tr}\Big[ \hat a^*_{\hbar_{n_\ell}}\Big( \frac{\tilde{g}_j(s)-g_j(s)}{\hbar_{n_\ell}}  \Big)\ \cW(\xi) \ \tilde{\varrho}_{\hbar_{n_{\ell}}}(s)\Big]- \int_{X^0}   e^{Q(\xi,u)}  \ \langle\alpha, b_j^0(s) \rangle_{L^2} \, d\tilde{\mu}_s(u)\Big \vert
\\&  \leq \underbrace{ \Big \vert {\rm Tr}\Big[ \hat a^*_{\hbar_{n_\ell}}\Big( \frac{\tilde{g}_j(s)-g_j(s)}{\hbar_{n_\ell}}-\tilde b_j^0(s)  \Big)\ \cW(\xi) \ \tilde{\varrho}_{\hbar_{n_{\ell}}}(s)\Big]\Big \vert }_{(1)}
\\ & +\underbrace{\Big \vert {\rm Tr}\Big[ \hat a^*_{\hbar_{n_\ell}}\Big( \tilde b_j^0(s) \Big)\ \cW(\xi) \ \tilde{\varrho}_{\hbar_{n_{\ell}}}(s)\Big]- \int_{X^0}   e^{Q(\xi,u)}  \ \langle\alpha, b_j^0(s) \rangle_{L^2} \, d\tilde{\mu}_s(u)\Big \vert}_{(2)} ,
\end{aligned}\]
where
\begin{equation}\label{chap4.eq.b0js}
\tilde b_j^0(s):=2\pi i  k \cdot q_{0j} \ \frac{\chi(k)}{\sqrt{\omega(k)}} \ e^{-2\pi i k \cdot \hat q_j+is\omega(k)} \,.
\end{equation}
For (1), let $S=(\hat{H}_0+1)^{1/2}$,  we have
\[\begin{aligned}
& \Big \vert {\rm Tr}\Big[ \hat a^*_{\hbar_{n_\ell}}\Big( \frac{\tilde{g}_j(s)-g_j(s)}{\hbar_{n_\ell}}-\tilde b_j^0(s)  \Big)\ \cW(\xi) \ \tilde{\varrho}_{\hbar_{n_{\ell}}}(s)\Big]\Big \vert
\\ & \leq \underbrace{\Big \Vert \frac{1}{\sqrt{\omega}}\Big [ \frac{\tilde{g}_j(s)-g_j(s)}{\hbar_{n_\ell}}-\tilde b_j^0(s)   \Big  ] \Big \Vert_{\cL(L^2,L^2\otimes L^2)}}_{\underset{\ell \rightarrow \infty}{\longrightarrow 0}}  \ \Vert S  \  \cW(\xi) \ S^{-1} \Vert_{\cL(\cH)} \  \Big \Vert S \  \tilde{\varrho}_\hbar(s)\Big \Vert_{\cL^1(\cH) }.
\end{aligned}\]
The above convergence follows from dominated convergence theorem and our assumptions. 
\vskip 1mm
For (2), according to the expression \eqref{chap4.eq.b0js},  we have
$$
\tilde b^0_j(s)= e^{-2\pi i k \cdot \hat q_j} \varphi_j(k),
$$
for some $\varphi_j\in L^2(\R^d)$.  Hence, applying Lemma  \ref{chap4.corr.convergence.eixk.ca} in the appendix, we conclude that $(2)$ converges to zero as $\ell\to \infty$.
\noindent Similar discussions lead  to 
\[ \lim_{\ell\rightarrow \infty}{\rm Tr}[ \hat a_{\hbar_{n_\ell}}\big( \frac{\tilde{g}_j(s)-g_j(s)}{\hbar_{n_\ell}}  \big)\ \cW(\xi) \ \tilde{\varrho}_{\hbar_{n_{\ell}}}(s)]=\int_{X^0}   e^{Q(\xi,u)}  \ \langle b_j^0(s), \alpha \rangle_{L^2} \, d\tilde{\mu}_s(u) \]
We deal now with the last line 
\[
\begin{aligned}
&{\rm Tr}[ \langle \alpha_0,\tilde g_j(s)\rangle_{L^2}\ \cW(\xi) \ \tilde{\varrho}_{\hbar_{n_{k}}}(s)]
\\& = \int_{\R^d} \overline{\alpha_0(k)}  \frac{\chi(k)}{\sqrt{\omega(k)}}  \ e^{is\omega(k) } \ e^{2\pi i k \cdot q_{0j} \hbar_{n_\ell}} \ {\rm Tr}[e^{-2\pi ik\cdot \hat q_j } \ \cW(\xi) \ \tilde{\varrho}_{\hbar_{n_{\ell}}}(s)] \  dk
\\&=\int_{\R^d} \overline{\alpha_0(k)}  \frac{\chi(k)}{\sqrt{\omega(k)}}  \ e^{is\omega(k) } \ e^{2\pi i k \cdot q_{0j} \hbar_{n_\ell}} \ {\rm Tr}[W_1\big(-2\pi k,0\big) \ \cW(\xi) \ \tilde{\varrho}_{\hbar_{n_{\ell}}}(s)] \  dk
\\&=\int_{\R^d} \overline{\alpha_0(k)}  \frac{\chi(k)}{\sqrt{\omega(k)}}  \ e^{is\omega(k) } \ e^{2\pi i k \cdot q_{0j} \hbar_{n_\ell}} \  {\rm Tr}[\cW\big(-2\pi k,0,0 \big) \ \cW(\xi) \ \tilde{\varrho}_{\hbar_{n_{\ell}}}(s)] \  dk
\\&=\int_{\R^d} \overline{\alpha_0(k)}  \frac{\chi(k)}{\sqrt{\omega(k)}}  \ e^{is\omega(k) }  \ e^{2\pi i k \cdot q_{0j} \hbar_{n_\ell}} \  {\rm Tr}[e^{\frac{-i\hbar_{n_\ell}}{2} \Im m \langle i2\pi k,q_0+ip_0 \rangle} \ \cW(p_0-2\pi k e_j,q_0,\alpha_0) \ \tilde{\varrho}_{\hbar_{n_{\ell}}}(s)] \  dk.
\end{aligned}
\]
We conclude 
\[ \begin{aligned}
&\lim_{\ell \rightarrow+\infty}{\rm Tr}[ \langle \alpha_0,\tilde g_j(s)\rangle_{L^2}\ \cW(\xi) \ \tilde{\varrho}_{\hbar_{n_{\ell}}}(s)]
\\&=\int_{\R^d} \overline{\alpha_0(k)}  \frac{\chi(k)}{\sqrt{\omega(k)}}  \ e^{is\omega(k) }\lim_{\ell \rightarrow+\infty}{\rm Tr}[ \cW(p_0-2\pi k,q_0,\alpha_0) \ \tilde{\varrho}_{\hbar_{n_{\ell}}}(s)] \  dk
\\&=\int_{\R^d} \overline{\alpha_0(k)}  \frac{\chi(k)}{\sqrt{\omega(k)}}  \ e^{is\omega(k) } \ \int_{X^0} e^{-2\pi i k\cdot q_j} e^{Q(\xi,u)} \    \, d\tilde{\mu}_s(u) dk
\end{aligned}
\]
By Fubini, we get 
\[ \begin{aligned}
&\lim_{\ell  \rightarrow+\infty}{\rm Tr}[ \langle \alpha_0,\tilde g_j(s)\rangle_{L^2}\ \cW(\xi) \ \tilde{\varrho}_{\hbar_{n_{\ell}}}(s)]
\\&=\int_{X^0}  e^{Q(\xi,u)} \int_{\R^d} \overline{\alpha_0(k)} \frac{\chi(k)}{\sqrt{\omega(k)}}   \ e^{-2\pi i k \cdot q_j+is\omega(k) } \ dk    \, d\tilde{\mu}_s(u) 
\\ &= \int_{X^0}  e^{Q(\xi,u)}  \  \langle \alpha_0,  g_j(s) \rangle_{L^2}  \, d\tilde{\mu}_s(u) .
\end{aligned}
\]
Similar discussions also work to prove 
\[\lim_{\ell  \rightarrow+\infty}{\rm Tr}[ \langle \tilde g_j(s), \alpha_0 \rangle_{L^2}\ \cW(\xi) \ \tilde{\varrho}_{\hbar_{n_{\ell}}}(s)]=\int_{X^0}   \langle g_j(s), \alpha_0 \rangle_{L^2}  \  e^{Q(\xi,u)}   \, d\tilde{\mu}_s(u) . \]

\end{proof}

\subsection{The characteristic equation}\label{top4wignerchara}
Below, we derive the final form of the  time-evolution equation satisfied by the Wigner measure $\tilde{\mu}_t$.
\begin{corollary}[Characteristic  equation] \label{top4characcc}
Assume \eqref{top4A0} and  $ \omega^{{1}/{2}} \chi \in L^2(\R^d,dk)$.
 Then, the charactristic equation \eqref{characfirstform}  can be further reduced to the following form
\be \label{top4chara} 
\begin{aligned}
\int_{X^0} e^{2i \pi \Re e\langle y,u \rangle_{X^{\sigma}}}  \ d\tilde{\mu}_t(u)&=\int_{X^0} e^{2i \pi \Re e\langle y,u \rangle_{X^{\sigma}}}  \ d\tilde{\mu}_{t_0}(u)
\\ &\quad +2\pi i \int_{t_0}^{t}\int_{X^0}  e^{2i \pi \Re e\langle y,u \rangle_{X^{\sigma}}} \Re e \langle v(s,u), y \rangle_{X^{\sigma}}  \ d\tilde{\mu}_s(u) \ ds,
\end{aligned}
\ee
for all $t, t_0 \in \R$ and $y \in X^{\sigma}$.
\end{corollary}
\begin{proof}
Define
\[ \tilde{\xi}:=(\frac{z_0}{2i\pi},\frac{\alpha_0}{\sqrt{2}\pi})\in X^0, \quad with \quad {\xi}=({z_0},{\alpha_0}) \in X^0.\]
We claim that    
\[b(s,\xi)=-2\pi \Re e\langle v(s,u),\tilde{\xi} \rangle_{X^0}. 
\]
Indeed, we first remark that 
\[
\begin{aligned}
-2\pi \Re e\langle v(s,u),\tilde{\xi} \rangle_{X^0}
 & =\underbrace{ -2\pi \Re e\langle (v(s,u))_z,\frac{z_0}{2i\pi} \rangle}_{(1)} \underbrace{-2\pi \Re e\langle (v(s,u))_\alpha,\frac{\alpha_0}{\sqrt{2}\pi} \rangle_{L^2}}_{(2)}.
\end{aligned}
\]
For (1), we have 
\[\begin{aligned}
&  -2\pi \Re e\langle (v(s,u))_z,\frac{z_0}{2i\pi} \rangle
 =- \Im m \langle (v(s,u))_z,{z_0} \rangle
\\ & =-\sum_{j=1}^{n}\Big( (v(s,u))_{q_j}\cdot p_{0j} - (v(s,u))_{p_j}\cdot q_{0j}  \Big)
\\ & =-\sum_{j=1}^n \nabla f_j({p}_j)\cdot p_{0j}-\sum_{j=1}^{n} \nabla_{q_j}V(q) \cdot q_{0j}
 +  \sum_{j=1}^{n} \big( \langle \alpha , b_j^0(s)\rangle_{L^2} + \langle b_j^0(s),\alpha \rangle_{L^2} \big) ,
\end{aligned}
\]
where recall that $v(s,u)$ is as in \eqref{top4vectorfieldv} and $b_j^0(s)$ is as in \eqref{top4bj0}.
\vskip 1mm
\noindent For (2), we have 
\[\begin{aligned}
&
-2\pi \Re e\langle (v(s,u))_\alpha,\frac{\alpha_0}{\sqrt{2}\pi} \rangle_{L^2}
=-\sqrt{2} \Re e\langle (v(s,u))_\alpha,{\alpha_0} \rangle_{L^2}
\\ &=-\sqrt{2} \Re e\langle -i\sum_{j=1}^{n}  \frac{\chi(k)}{\sqrt{\omega(k)}}e^{-2\pi i k\cdot q_j+is\omega(k)},{\alpha_0} \rangle_{L^2}
\\ & =\sqrt{2}  \sum_{j=1}^{n} \Im m \langle g_j(s),\alpha_0 \rangle_{L^2}.
\end{aligned}
\]
where $g_j(s)$ is as in \eqref{top4gjs}. On the other hand, we have 
\[\begin{aligned}
&\frac{i  }{\sqrt{2}} \Big(\langle \alpha_0,g_j(s)\rangle_{L^2(\R^d,\C)} -\langle  g_j(s),\alpha_0 \rangle_{L^2(\R^d,\C)}\Big) 
\\&=-\sqrt{2} \Im m \langle \alpha_0,g_j(s)\rangle_{L^2(\R^d,\C)} 
 =\sqrt{2} \Im m \langle g_j(s),\alpha_0\rangle_{L^2(\R^d,\C)}.
\end{aligned}
\]
And, thus combining the above arguments, we prove the claimed results.
The Characteristic equation \eqref{characfirstform} becomes then 
\be \label{intemediateform}\int_{X^0}e^{Q(\xi,u)}  \, d\tilde{\mu}_t(u)=\int_{X^0}e^{Q(\xi,u)}  \, d\tilde{\mu}_{t_0}(u)+2\pi i\int_{t_0}^{t}  \int_{X^0} \, e^{Q(\xi,u)} \ \Re e\langle v(s,u),\tilde{\xi} \rangle_{X^0}   \, d\tilde{\mu}_s(u)\, ds. \ee
We have, with $Q(\xi,u)$ as in \eqref{chap4.eq.Q.phase},  that 
\be \label{QandR}Q(\xi,u)= 2\pi i \Re e \langle \tilde{ \xi}, u \rangle_{X^0}. \ee
We have also  for all $y=(p,q,\alpha)\in X^{2\sigma}$ and all $ \tilde{\xi}=(p,q,\omega^{2\sigma} \alpha)\in X^0$ that 
\be \label{helpful2}\begin{aligned}
&  \Re e \langle y,u  \rangle_{X^{\sigma}}=\Re e \langle \tilde{\xi}, u \rangle_{X^0},
\\&\Re e \langle v(s,u), y \rangle_{X^{\sigma}}=\Re e \langle v(s,u),\tilde{\xi} \rangle_{X^0}
\end{aligned}\ee
By this way, plugging \eqref{QandR}-\eqref{helpful2} in \eqref{intemediateform}  gives that \eqref{top4chara} is valid for all $y \in X^{2\sigma}$. The latter could be extended to all $y\in X^{\sigma}$ by dominated convergence theorem and the bound \eqref{top4integrability}.
\end{proof}

\subsection{The Liouville equation}\label{top4liouville}
In this part, we relate the characteristic equation \eqref{top4chara} satisfied by the set of Wigner measures $(\tilde{\mu}_t)_{t\in \R}$ to a special Liouville equation. To do that, we need to have some integrability condition of the vector field $v$ of \eqref{top4ivp} with respect to this Wigner measure and some regularities of the latter measure.

\begin{lemma}[Integrability  of the vector field $v$]\label{top4integr} Assume \eqref{top4A0} and  \eqref{top4A1} hold true.
 Then, there exists a constant $C>0$ such that for all $u= (p,q,\alpha) \in X^{\sigma}$,
\be \label{top4boundv}
\begin{aligned}
\Vert v(t,u) \Vert_{X^{\sigma}} \leq C \big( \Vert u \Vert_{X^0}^2 +1 \big).
\end{aligned}
 \ee
Moreover, for any bounded open interval $I$,
\be\label{top4integrability}
\int_I \int_{X^{\sigma}}\Vert v(t,u) \Vert_{X^{\sigma}}  \ d\tilde{\mu}_t (u) \ dt<+\infty.  
\ee
\end{lemma}
\begin{proof}
The non-autonomous vector field $v$ is defined in terms of the nonlinearity $\cN$ as indicated in \eqref{top4vectorfieldv}.  Then it is not hard to  see by looking at the proof of Proposition \ref{top4contbddF}  that 
\begin{itemize}
\item in the semi-relativistic case, since the function $\nabla f_j(\hat p _j)$ is bounded, we get 
\be \label{top4boundv1}
\begin{aligned}
\Vert v(t,u) \Vert_{X^{\sigma}}\leq C \big( \Vert \alpha \Vert_{L^2}^2 +1 \big).
\end{aligned}
 \ee
\item in the non-relativistic case, we get 
\be \label{top4boundv2}
\begin{aligned}
\Vert v(t,u) \Vert_{X^{\sigma}}\leq C \big( \Vert \alpha \Vert_{L^2}^2+|p|^2 +1 \big).
\end{aligned}
 \ee
\end{itemize}
Thus, both inequalities \eqref{top4boundv1} and \eqref{top4boundv2} lead to \eqref{top4boundv}.
Now,  the integrability condition \eqref{top4integrability} is a consequence of \eqref{top4boundd} in Proposition \ref{top4uniquewignermeasure}.
\end{proof}
\noindent We establish now some regularity of the Wigner measures ${(\tilde{\mu}_t)}_{t\in \R}$ with respect to time.
\begin{lemma}[Regular properties of the Wigner Measure $\tilde{\mu}_t$]\label{top4regularprp}
The Wigner measures ${(\tilde{\mu}_t)}_{t\in \R}$ extracted in Proposition \ref{top4uniquewignermeasure}  satisfy
\begin{itemize}
\item [(i)]  ${\tilde{\mu}_t}$ concentrates on $X^\sigma$ i.e. $ {\tilde{\mu}_t}(X^\sigma)=1$;
\item [(ii)] $\R \ni t \longmapsto \tilde{\mu}_t \in \cP(X^\sigma)$ is weakly narrowly continuous.
\end{itemize}
\end{lemma}
\begin{proof}
For the first assertion (i), we have from Proposition \ref{top4uniquewignermeasure}  that 
 \[\int_{X^0} \Vert u\Vert_{X^{\sigma}}^2 \ d\tilde{\mu}_t(u)<C.\]
 And, from the Markov's inequality, we have
 \[\tilde{\mu}_t(\{  u\in X^0:  \  \Vert u \Vert_{X^\sigma} \geq \eps  \}) \leq \frac{1}{\eps^2} \ \tilde{\mu}_t(  \Vert u \Vert_{X^\sigma}).\] 
 Let $\eps \rightarrow \infty$, we get 
 \[ \tilde{\mu}_t(\{  u\in X^0;  \  u \notin X^\sigma  \}) =0.\]
 Hence, we get that the measure $\tilde\mu_t$ is concentrated in $X^\sigma$.
 The second assertion (ii) is proved in a  similar fashion as in \cite[Lemma 5.5]{Z} using Prokhorov's Theorem. 
\end{proof}

\noindent In the coming discussions, for more details, we refer the reader to Appendix A in \cite{Z}. 
Let I be an open bounded interval. Define the space of smooth cylindrical functions on $I \times X^\sigma$, denoted by $\cC_{0,{\it cyl}}^\infty(I\times X^\sigma)$, as follows 
\[ \begin{aligned}
\cC_{0,{\it cyl}}^{\infty}(I\times X^\sigma)&:=\Big \{ \phi:I\times X^{\sigma} \rightarrow \R; \ \phi(t,u)=\psi(t, \pi(u)), \ \psi \in \cC_0^\infty(I \times \R^{d^{\prime}}),
\\ & \qquad    \ \pi:X^{\sigma} \rightarrow \R^{d^{\prime}} , \  {d^{\prime}}\in \N   \Big\},
\end{aligned}\]
where $\pi:X^\sigma \rightarrow \R^{d^\prime}$ is a projection of the form $ \pi:u\rightarrow \pi(u)=(\Re e\langle u,e_1 \rangle_{X^\sigma},\cdots, \Re e\langle u,e_{d^\prime}\rangle_{X^\sigma}),$ with $(e_1,\cdots,e_{d^\prime})$ is an arbitrary orthonormal family of $X^\sigma$.

\begin{proposition}\label{top4satsifyliou}
The family of Wigner measures $(\tilde{\mu}_t)_{t\in \R}$ defined in Proposition \ref{top4uniquewignermeasure} is a weakly narrowly continuous solution to the following Liouville equation
\begin{equation} \label{top4le}\int_{I}\int_{X^\sigma} \{  \partial_t \phi (t,u) +\Re e\langle v(t,u),\nabla \phi(t,u) \rangle_{X^\sigma}\}d \tilde{\mu}_t(u) \ dt=0,  \tag{LE}
\end{equation} 
for any bounded open interval $I$ containing the origin with $\phi\in \cC_{0,cyl}^\infty(I\times X^\sigma).$
\end{proposition}
\begin{proof}
It is a direct consequence of Lemma \ref{top4relation} by selecting ${\rm H}\equiv X^\sigma$ which is a Hilbert space. 
More precisely, all the prerequists of Lemma \ref{top4relation} are satisfied. Indeed, we have 
\begin{itemize}
\item from Corollary \ref{top4characcc} that the set of Wigner measures $\{\tilde{\mu}_t\}_{t\in I}$ solves the characteristic equation \eqref{top4chara};
\vskip 2mm
\item from Lemma \ref{top4regularprp}, we have checked that $\tilde \mu_t \in \cP(X^\sigma)$ is a weakly narrowly continuous;
\vskip 2mm
\item from Lemma \ref{top4integr}, we have checked the integrability condition of $v$ with respect to $\tilde{\mu}_t$.
\end{itemize}
And thus the result follows.
\end{proof}

\section{Proof of the main result}
 In order to prove the main Theorems \ref{top4theorem1} and \ref{top4theorem2}, we must establish some identities. It is important to note that the statement of Theorem \ref{top4theorem1} is not related directly to  the quantum dynamics and does not require any restrictions on it. Therefore, our plan is to ensure that the assumptions   \eqref{top4S0} and \eqref{top4S1} are applied to a specific class of density matrices, namely the coherent states.
To achieve this, we must first define the coherent states for the particle and field components separately, and then generalize to the entire interacting space since we are dealing with an interaction between particles and field.
Let $u_0=(z_0,\alpha_0)\in X^0$ and consider the  family of  coherent states
\[\cC_\hbar(u_0)=\Big\vert W_1(\frac{\sqrt{2}}{i\hbar}z_0)\psi\otimes W_2(\frac{\sqrt{2}}{i\hbar}\alpha_0) \Omega \Big\rangle \Big\langle W_1(\frac{\sqrt{2}}{i\hbar}z_0)\psi\otimes W_2(\frac{\sqrt{2}}{i\hbar}\alpha_0) \Omega\Big \vert\]
where we have introduced
\begin{itemize}
\item [$\rightarrow$] the coherent vector: $W_1(\frac{\sqrt{2}}{i\hbar}z_0)\psi, $ centered on $z_0\in \C^{dn}$
where $\psi(x)=(\pi \hbar)^{-dn/4} \ e^{-x^2/2\hbar} \in L^2(\R^{dn},dx)$ is the normalized gaussian function on the particles related to the particle space $L^2(\R^{dn},\C)$. 
\item [$\rightarrow$] the coherent vector: $W_2(\frac{\sqrt{2}}{i\hbar}\alpha_0) \Omega$  in the Fock space, for $\alpha \in \cG^0$ and $\Omega$ is the vacuum vector on the fock space.  
\end{itemize}
It bears noting that these family of coherent states gives rise to a family of density matrices satisfying the assumptions   \eqref{top4S0} and \eqref{top4S1}.
\begin{lemma} [The family of coherent states]\label{top4coherentstates}
The family of coherent states  $(\cC_\hbar(u_0))_{\hbar\in(0,1)}.$ satisfies
\[\cM(\cC_\hbar(u_0),\hbar\in(0,1))=\{\delta_{u_0}\},\] where $\delta_{u_0}$ is the Dirac measure centered on $u_0$.
Moreover, if $u_0=(z_0,\alpha_0) \in X^\sigma$, then $(\cC_\hbar(u_0))_{\hbar\in(0,1)}$ satisfies \eqref{top4S0}and \eqref{top4S1}.
\end{lemma}
\begin{proof}
 We have
\begin{align*}
& {\rm Tr}(\cC_\hbar(u_0) \  d\Gamma(\omega^{2\sigma}))= \Vert \alpha_0 \Vert_{\cG^\sigma}^2
\\ & {\rm Tr}(\cC_\hbar(u_0) \  \hat{p}^2)= \langle \psi, \hat{p}^2 \psi \rangle-2p_0^2
\\ &{\rm Tr}(\cC_\hbar(u_0) \  \hat{q}^2)=  \langle \psi, \hat{q}^2 \psi \rangle-2q_0^2.
\end{align*}
\end{proof}

\noindent Below, we give useful lemma which relates the Wigner measure $\tilde{\mu}_t$ to $\mu_t$ in terms of the free field flow $\Phi^{\it f}_t$.
\begin{lemma}[Relations between the sets of Wigner measure]\label{top4reee} Let $(\varrho_\hbar )_{h\in(0,1)}$ be a family of density matrices satisfying \eqref{top4S0} and \eqref{top4S1}.  Define 
\[\tilde \varrho_\hbar(t):= e^{i \frac{t}{\hbar} d\Gamma(\omega)} \ \varrho_\hbar \  e^{-i \frac{t}{\hbar} d\Gamma(\omega)}. \]
Then, we can assert that 
\begin{enumerate}
\item the family of states $ (\tilde \varrho_\hbar(t))_{\hbar\in(0,1)}$ satisfies \eqref{top4S0} and \eqref{top4S1};
\item for all sequences $(\hbar_{n})_{n \in \N}$ with $\hbar_{n} \rightarrow 0$, there exists a subsequence $\hbar_{n_\ell}$ with $\hbar_{n_\ell} \rightarrow 0$ such that 
\[\cM(\tilde \varrho_{\hbar_{n_\ell}}(t) , \ell \in \N)= \{ (\Phi^{\it f}_{-t})_\sharp \mu ; \ \mu \in \cM(  \varrho_{\hbar_{n_\ell}} , \ell \in \N)\}, \]
where $\Phi^{\it f}_t$ is the free field flow as in \eqref{top4freeflow} .
\end{enumerate}
\end{lemma}
\begin{proof}
The first  assertion is a consequence of Lemma \ref{top4uniformestimate}. 
Let $\mu\in  \cM( \varrho_{\hbar_{\ell}} , \ell \in \N)$ and $ \tilde \mu_t \in \cM(\tilde \varrho_{\hbar_{\ell}}(t) , \ell \in \N)$. 
On one hand, we have 
\[\begin{aligned}
&\lim_{\ell } {\rm Tr} \big[ \tilde \varrho_{\hbar_\ell}(t)\ \cW(\xi ) \big]= \lim_{\ell } {\rm Tr} \big[  \varrho_{\hbar_\ell} \ e^{-i \frac{t}{\hbar} d\Gamma(\omega)} \cW(\xi ) \ e^{i \frac{t}{\hbar} d\Gamma(\omega)} \big]
 = \lim_{\ell } {\rm Tr} \big[  \varrho_{\hbar_\ell} \  \cW(\Phi^{\it f}_t(\xi )) \ \big]
\\ &= \int_{X^0} e^{Q(\Phi^{\it f}_t(\xi ) \  ,u)} \ d\mu(u)
= \int_{X^0} e^{Q(\xi , \ \Phi^{\it f}_{-t}(u))} \ d\mu(u)
= \int_{X^0} e^{Q(\xi , u)} \ d  (\Phi^{\it f}_{-t})_\sharp \mu(u).
\end{aligned}\]
On the other hand, we have 
\[\begin{aligned}
\lim_{\ell } {\rm Tr} \big[ \tilde \varrho_{\hbar_\ell}(t)\ \cW(\xi ) \big]= \int_{X^0} e^{Q(\xi , u)} \ d  \tilde \mu_t(u).
\end{aligned}\]
We conclude then that 
\[\tilde \mu_t= (\Phi^{\it f}_{-t})_\sharp \mu .\]
\end{proof}
\noindent Below, we start the proof Theorems \ref{top4theorem1} and \ref{top4theorem2}.
\paragraph{\it Proof of Theorem \ref{top4theorem1}.}\label{top4proof1}
Let $u_0 \in X^\sigma$  and defines the density matrices  for all $\hbar \in (0,1)$ as follows
\[\varrho_\hbar:=\cC_\hbar(u_0).\]
Then since $u_0 \in X^\sigma$, we can assert by Lemma \ref{top4coherentstates} that the family of density matrices  $ (\varrho_\hbar)_{\hbar\in(0,1)}$ satisfies \eqref{top4S0} and \eqref{top4S1}.
Thus, with this choice of density matrices and using the arguments of 
Proposition \ref{top4uniquewignermeasure}, we can assert that for each sequence $(\hbar_{n})_{n\in \N}$ with  $\underset{n\rightarrow\infty}{\hbar_n\rightarrow 0}$, there exists a subsequence  $(\hbar_{n_\ell})_{\ell \in \N} $ with  $\underset{\ell \rightarrow\infty}{\hbar_{n_\ell}\rightarrow 0}$ and a family of Borel probability measure $\{\tilde \mu_t\}_{t\in \R}$ in $X^0$ such that 
\[ \cM(e^{i \frac{t}{\hbar_{n_\ell}} d\Gamma(\omega)} \ e^{-i \frac{t}{\hbar_{n_\ell}} \hat{H}} \  \cC_{\hbar_{n_\ell}}(u_0) \ e^{i \frac{t}{\hbar_{n_\ell}} \hat H} \ e^{-i \frac{t}{\hbar_{n_\ell}} d\Gamma(\omega)},\ell \in \N)=\{\tilde{\mu}_t\}.\]
Now, on one hand , we do have   from Proposition \ref{top4satsifyliou} that $\{\tilde \mu_t\}_{t\in \R}$ is weakly narrowly continuous solution to the Liouville equation \eqref{top4le}; from the other hand, 
 from Lemma \ref{top4integr} , we can assert that all the prerequists to apply Theorem \ref{top4globalflow} are in our hand.
 To recover the proof of Theorem \ref{top4theorem1}, we follow the steps below.  
\begin{itemize}
\item [$ \triangleleft$] We apply Theorem \ref{top4globalflow}  with the measure $\tilde{\mu}_t$ obtained above, we get the global well posedness of the initial value problem \eqref{top4ivp} $\tilde \mu_0$-almost all initial data in $X^\sigma$  as well as the existence of  a generalized Borel measurable global flow $\tilde{\Phi}_t$  as follows
 \[\begin{array}{rcl}
\tilde{\Phi}_t \, : \,\mathfrak{G} & \longrightarrow & X^\sigma\\
u_0 & \longmapsto & u(t),
\end{array}  \]
where $\mathfrak{G}$ is the ensemble of initial data obtained from Theorem \ref{top4globalflow}.
\vskip 3mm
\item[$ \triangleleft$] Let $u_0\in X^\sigma$. From  Lemma \ref{top4coherentstates}, we have 
$\tilde{\mu}_0(\mathfrak{G})=\delta_{u_0}(\mathfrak{G})=1$. This implies $u_0\in \mathfrak{G}$;
\vskip 3mm
\item[$ \triangleleft$] Use the equivalence between the solution to \eqref{top4ivp} and \eqref{top4particlefieldequation}, we can show the existence and uniqueness of the  solution of \eqref{top4particlefieldequation} with a generalized global  flow  
\[\Phi_t(u_0)=\Phi_t^{\it f}\circ\tilde{\Phi}_t(u_0),\]
where $\Phi_t^{\it f}$ is the free flow and $\tilde \Phi_t$ is the generalized flow of $\eqref{top4ivp}$;
\end{itemize}

\paragraph{\it Proof of Theorem \ref{top4theorem2}}
We have here to prove the validity of Bohr's correspondence principle. Assume $(\varrho_{\hbar})_\hbar$ is  a family of density matrices satisfying the Assumptions \eqref{top4S0} and \eqref{top4S1}. Then, using  Proposition \ref{top4uniquewignermeasure}, we can assert that for each sequence $(\hbar_{n})_{n \in \N}$ with  $\underset{n\rightarrow\infty}{\hbar_n\rightarrow 0}$, there exists a subsequence  $(\hbar_{n_\ell})_{\ell \in \N} $ with  $\underset{\ell \rightarrow\infty}{\hbar_{n_\ell}\rightarrow 0}$ and a family of Borel probability measures $ \{\tilde \mu_t\}_{t\in \R}$  in $X^0$ such that
\[
\cM(\tilde \varrho_{\hbar_{n_\ell}}(t),\ell \in \N)=\{\tilde \mu_t\}.
\]
By Lemma \ref{top4reee}, we have 
\[\cM(\varrho_{\hbar_{n_\ell}}(t),\ell \in \N)= \{(\Phi_t^{\it f})_\sharp \tilde{\mu}_t; \ \tilde{\mu}_t\in \cM(\tilde \varrho_{\hbar_{n_\ell}}(t),\ell \in \N)  \} .\]
This implies that 
\[\cM(\varrho_{\hbar_{n_\ell}}(t),\ell \in \N)=\{\mu_t\}=\{(\Phi_t^{\it f})_\sharp \tilde{\mu}_t \} \]
From (ii) in Porbabilistic representation, we can assert that for any bounded Borel functions $\psi:X^\sigma \rightarrow \R$
\[\int_{X^\sigma} \psi(u)  \ d\tilde{\mu}_t(u)=\int_{\cF_I} \psi(e_t(u_0,u(\cdot))) \ d\eta(u_0,u(\cdot)).\]
Since, we have the generalized global flow $\tilde{\Phi}_t$ to \eqref{top4ivp}, we get 
\[e_t(u_0,u(\cdot))=\tilde{\Phi}_t(e_0(u_0,u(\cdot)))=\tilde{\Phi}_t(u_0).\]
This gives
\[\begin{aligned}
\int_{X^\sigma} \psi(u) \ d\tilde{\mu}_t(u)= \int_{X^\sigma \times \cC(\overline{I},X^\sigma)} \psi \circ \tilde{\Phi}_t(e_0(u_0,u(\cdot))) \ d\eta(u_0,u(\cdot))
 =\int_{X^\sigma}  \psi \circ \tilde{\Phi}_t(u) \ d \tilde{\mu}_0(u).
\end{aligned}
\]
We conclude that $\tilde{\mu}_t=(\tilde{\Phi}_t)_\sharp \tilde{\mu}_0.$
This implies that:
\[\begin{aligned}
\mu_t=(\Phi_t^{\it f})_\sharp \tilde{\mu}_t
=( \Phi_t^{\it f}\circ \tilde \Phi_t)_\sharp\tilde{\mu}_0  
= (\Phi_t)_\sharp \tilde\mu_0
= (\Phi_t)_\sharp \mu_0
\end{aligned}\]
and where we have used $ \tilde{\mu}_0=\mu_0$ as a consequence of 
\[\tilde{\varrho}_\hbar(0)=\varrho_\hbar(0)=\varrho_\hbar.\]
\hfill $\square$
\appendix
\section{Prokhorov theorem}
Let $X$ be separable metric space. The proof of the following result is  proved in \cite[Theorem 5.1.3]{AmbrosioLuigi2005GFIM}.
\begin{theorem}[Prokhorov Theorem]\label{prokhorov}
If a set $\cK \subset \cP(X)$ is tight i.e.
\[\forall \eps >0, \ \exists K_\eps \ \text{compact in $X$ such that } \mu(X \setminus K_\eps) \leq \eps, \  \forall \mu \in \cK,\]
then $\cK $ is relatively compact in $\cP(X)$.
\end{theorem}
\section{Useful results}
The following results relate the Liouville equations and the Characteristic equations satisfied by a family of Wigner measures. For more details, we refer the reader to \cite[Proposition 4.2]{c}. Let ${\rm H}$ be a Hilbert space.
\begin{lemma}[Equivalence]\label{top4relation}
Let $v:\R \times {\rm H} \rightarrow {\rm H} $ be a continuous vector field such that it is bounded on bounded sets. Let $ I \ni t \rightarrow \mu_t$ a weakly narrowly continuous curve in $\cP({\rm H})$ such that we have the following integrability condition 
\[ \int_I \int_{{\rm H}} \Vert v(t,u) \Vert_{\rm H} \ d\mu_t(u) \ dt<+\infty.  \]
Then, the following statements are equivalent:
\begin{itemize}
\item [(i)] $\{\mu_t\}_{t\in I}$ is a solution of  Liouville equation \eqref{top4le};
\item [(ii)] $\{\mu_t\}_{t\in I}$ solves the characteristic equation \eqref{top4chara} for all $t\in I$and for all $y\in {\rm H}$.
\end{itemize}
\end{lemma}
The subsequent outcomes illustrate how to build the global solution to the \eqref{top4ivp} utilizing measure-theoretical approaches and certain probabilistic representations of the measure-valued solutions for the Liouville equation. Additional information on the topic can be found in \cite{ammari2023sure} or in the Appendices of \cite{Z}.

\begin{theorem}[Global flow of the initial value problem]\label{top4globalflow}
Let $v:\R\times X^\sigma\rightarrow X^\sigma$ be a continuous vector field bounded on bounded sets. Assume 
\begin{itemize}
\item $\exists t\in \R \rightarrow \tilde\mu_t \in \cP(X^\sigma) $ a weakly narrowly continuous solution to \eqref{top4le} satisfying the integrability condition \eqref{top4integrability} on $I$; 
\item There is at most one solution of the initial value problem \eqref{top4ivp} over any bounded open interval $I$ containing the origin .
\end{itemize}
Then for $\tilde \mu_0$-almost all initial conditions $u_0$ in $X^\sigma$, there exists a unique global strong solution to \eqref{top4ivp}. In addition, the set 
\[\begin{aligned}& \mathfrak{G}:=\{ u_0 \in X^\sigma: \exists u(\cdot)  \ \text{a global strong solution of \eqref{top4ivp}} \\& \qquad \qquad \qquad  \quad  \  \text{with the initial condition $u_0$}\},
\end{aligned}\]
is Borel subset of $X^\sigma $ with $\tilde \mu_0(\mathfrak{G})=1$ and for any time $t\in \R$ the map $ u_0 \in \mathfrak{G} \rightarrow \tilde{\Phi}_t(u_0)=u(t)$
is Borel measurable.
\end{theorem}
\begin{proposition}[Superposition principle]
There  exists $\eta\in \cP(X^\sigma\times \cC(\overline{I},X^\sigma))$ satisfying:
\begin{itemize}
\item [(i)] $\eta(\cF_I)=1$ where 
\[\cF_I:=\Big\{ (u_0,u(\cdot))\in X^\sigma\times \cC(\overline{I},X^\sigma): \ u(\cdot) \ \text{satisfies \eqref{top4ivp} on $I$ with  $u_0$} \Big\}\]
\item[(ii)] $\tilde{\mu}_t= (e_t)_\sharp\eta, \quad \forall t\in I$, where the map \[e_t:(u_0,u(\cdot))\in X^\sigma\times \cC(\overline{I},X^\sigma)\rightarrow u(t)\in X^\sigma\] is the evaluation map.
\end{itemize}
\end{proposition}
\section{Technical results about convergence}
\medskip
Finally, we prove two technical lemmas which are  useful for the study of the quantum-classical convergence in Subsection \ref{top4convv}. We denote by $\mathcal{F}$ the Fourier transform on $\R^d$.
\begin{lemma}\label{chap4.corr.convergence.ca}
Let   $(\varrho_{\hbar})_{\hbar \in (0,1)}$ be a family of density matrices on the Hilbert space $\cH$ satisfying \eqref{top4S0}-\eqref{top4S1} for $\sigma=\frac 1 2$. Assume that  for some sequence $(\hbar_\ell)_{\ell\in\N}\subset (0,1)$, $\hbar_\ell \to 0$, there exists a (unique) Borel measure $\mu\in \mathcal{P}(X^0)$ such that
$$
\mathcal{M}(\varrho_{\hbar_\ell}, \ell \in\N)=\{\mu\}.
$$
Then for  any  $\varphi\in L^2(\R^d)$, $\beta\in \mathcal{F}(L^1(\R^d))$ and $\xi\in X^0$, $j=1,\cdots,n$,
\begin{eqnarray*}
\lim_{\ell \to\infty} {\rm Tr}\Big[ \beta(\hat q_j) \,  \hat a_{\hbar_\ell}(\varphi)  \, \mathcal{W}(\xi)\,\varrho_{\hbar_\ell}\Big]&=& \int_{X^0} \beta(q_j) \langle \varphi, \alpha\rangle_{L^2(\R^d)} \,\, e^{Q(\xi,u)} \, d\mu(u)\,, \\
\lim_{\ell \to\infty} {\rm Tr}\Big[ \beta(\hat q_j) \, \hat a^*_{\hbar_\ell}(\varphi) \, \mathcal{W}(\xi)\,\varrho_{\hbar_\ell}\Big]&=& \int_{X^0} \beta(q_j) \langle \alpha, \varphi\rangle_{L^2(\R^d)} \,\, e^{Q(\xi,u)} \, d\mu(u)\,,
\end{eqnarray*}
with $u=(p,q,\alpha)\in X^0$ and $Q(\cdot,\cdot)$ is the phase given in \eqref{chap4.eq.Q.phase}.
\end{lemma}
\begin{proof}
The two limits are similar. By linear combinations one can use instead the fields operators
$$
\hat\phi_\hbar(\varphi)=\frac{\hat a_\hbar^*(\varphi)+\hat a_\hbar(\varphi)}{\sqrt{2}}  \qquad \text{ and } \qquad \hat\pi_\hbar(\varphi)=\frac{i \hat a_\hbar^*(\varphi)-i\hat a_\hbar(\varphi)}{\sqrt{2}}\,.
$$
So, it is enough to show
\begin{equation}\label{chap4.proof.lem1.tec.0}
\lim_{\ell \to\infty} {\rm Tr}\Big[ \beta(\hat q_j) \, \hat\phi_{\hbar_\ell }(\varphi) \, \mathcal{W}(\xi)\,\varrho_{\hbar_\ell}\Big]= \sqrt{2} \int_{X^0} \beta(q_j)  \Re e\langle \alpha, \varphi\rangle_{L^2(\R^d)} \,\, e^{Q(\xi,u)} \, d\mu(u)\,.
\end{equation}
Our goal is to prove \eqref{chap4.proof.lem1.tec.0}. Since $\beta=\mathcal{F}(g)$ for some $g\in L^1(\R^d)$, one can write
\begin{equation}\label{chap4.proof.lem1.tec.1}
{\rm Tr}\Big[ \beta(\hat q_j) \, \hat\phi_{\hbar_\ell}(\varphi) \, \mathcal{W}(\xi)\,\varrho_{\hbar_\ell}\Big]=
\int_{\R^d} g(y) \,  
{\rm Tr}\Big[ e^{-2\pi i y\cdot \hat q_j} \, \hat\phi_{\hbar_\ell}(\varphi) \, \mathcal{W}(\xi)\,\varrho_{\hbar_\ell}\Big] \, dy. 
\end{equation}
Furthermore, dominated convergence applies to the right hand side of  \eqref{chap4.proof.lem1.tec.1} thanks to the assumptions \eqref{top4S0}-\eqref{top4S1} and the estimates in Lemma \ref{top4lemma1}. Thus, the limit
\eqref{chap4.proof.lem1.tec.0} reduces to
\begin{equation}\label{chap4.proof.lem1.tec.2}
\lim_{\ell \to\infty} {\rm Tr}\Big[ e^{-2\pi i y\cdot \hat q_j} \, \hat\phi_{\hbar_\ell}(\varphi) \, \mathcal{W}(\xi)\,\varrho_{\hbar_\ell}\Big]= \sqrt{2} \int_{X^0} e^{-2\pi i y\cdot  q_j} \,  \Re e\langle \alpha, \varphi\rangle_{L^2(\R^d)} \,\, e^{Q(\xi,u)} \, d\mu(u)\,.
\end{equation}
Now, applying \cite[Lemma B.2]{Z}, we obtain \eqref{chap4.proof.lem1.tec.2} for all $y\in\R^d$ since 
$$
e^{-2\pi i y\cdot \hat q_j}  \mathcal{W}(\xi)= W_1(-2\pi y,0) \,\mathcal{W}(\xi)=\mathcal{W}(-2\pi y,0,0) \,\mathcal{W}(\xi)= e^{i\hbar_{\ell} \pi y\cdot p_0}\mathcal{W}(\tilde\xi)\,,
$$
with $\tilde\xi=(-2\pi y,0,0)+\xi$ and $\xi=(p_0,q_0,\alpha)$. Recall that the Weyl-Heisenberg operator $W_1(\cdot)$ is given in \eqref{top4weylheseinbergtranslation} while $\mathcal{W}(\cdot)$ is defined by \eqref{chap4.eq.weyl.fock}-\eqref{chap4.eq.wcal}. 
\end{proof}

\begin{lemma}\label{chap4.corr.convergence.eixk.ca}
Let   $(\varrho_{\hbar})_{\hbar \in (0,1)}$ be a family of density matrices on the Hilbert space $\cH$ satisfying \eqref{top4S0}-\eqref{top4S1} for $\sigma=\frac 1 2$. Assume that  for some sequence $(\hbar_\ell )_{\ell \in\N}\subset (0,1)$, $\hbar_\ell \to 0$, there exists a (unique) Borel measure $\mu\in \mathcal{P}(X^0)$ such that
$$
\mathcal{M}(\varrho_{\hbar_\ell}, \ell\in\N)=\{\mu\}.
$$
Then for  any  $\varphi\in L^2(\R_k^d), \xi\in X^0$ and $j=1,\cdots,n$,
\begin{eqnarray*}
\lim_{\ell \to\infty} {\rm Tr}\Big[ \hat a_{\hbar_\ell }(e^{-2\pi i \,k\cdot \hat q_j} \varphi) \, \mathcal{W}(\xi)\,\varrho_{\hbar_\ell}\Big]&=& \int_{X^0}  \langle e^{-2\pi i\,k \cdot q_j} \varphi, \alpha\rangle_{L^2(\R_k^d)} \, \, e^{Q(\xi,u)} \, d\mu(u)\,, \label{chap4.eq.a.cv.tec.1}\\
\lim_{\ell \to\infty} {\rm Tr}\Big[ \hat a^*_{\hbar_\ell}(e^{-2\pi i\, k \cdot \hat q_j}\varphi))\,  \mathcal{W}(\xi)\,\varrho_{\hbar_\ell }\Big]&=& \int_{X^0} \langle \alpha, e^{-2\pi i\,k\cdot q_j} \varphi\rangle_{L^2(\R_k^d)} \,\, e^{Q(\xi,u)} \, d\mu(u)\,, \label{chap4.eq.a.cv.tec.2}
\end{eqnarray*}
with $u=(p,q,\alpha)\in X^0$ and $Q(\cdot,\cdot)$ is the phase given in \eqref{chap4.eq.Q.phase}.
\end{lemma}
\begin{proof}
According to Definition \ref{top4definitionwigner}  of Wigner measures  and   \cite[Theorem 6.2 and Proposition 6.4]{ammari2008ahp}, we deduce that
$$
\mathcal{M}(\mathcal{W}(\xi)\varrho_{\hbar_\ell}, \ell \in\N)=\{e^{Q(\xi, \cdot)} \mu\}\,.
$$
Here, we have used the extension of the notion of Wigner measures to trace-class operators which are not necessary non-negative nor trace normalized (see \cite[Proposition 6.4]{ammari2008ahp}).
Let $\{e_m\}_{m\in\N}$ be O.N.B of the Hilbert space $L^2(\R_k^d)$.   The two limits are similar (almost conjugate) and it is enough to explain the argument for  the second one.   We denote $\varrho_{\hbar_\ell }(\xi):=\mathcal{W}(\xi)\varrho_{\hbar_\ell}$ and
$b(\hat q_j)=e^{-2\pi i k\cdot \hat q_j} \varphi$.  We have
\[\begin{aligned}
&\Big \vert {\rm Tr}\Big[ \hat a^*_{\hbar_{\ell}}(b(\hat q_j)) \ {\varrho}_{\hbar_{\ell}}(\xi)\Big]- \int_{X^0}   e^{Q(\xi,u)}  \ \langle   \alpha, b(q_j) \rangle_{L^2} \, d{\mu}(u)\Big \vert
\\&\leq \underbrace{ \Big \vert {\rm Tr}\Big[\Big( \hat a^*_{\hbar_{\ell}}( b(\hat q_j))- \sum_{m=1}^R
\hat a^*_{\hbar_{\ell}} ( e_m )  \mathcal{F}[\varphi \bar e_m](\hat q_j) \Big)
\  \ {\varrho}_{\hbar_{\ell}}(\xi)\Big]\Big \vert }_{(1)}
\\ & +\sum_{m=1}^R \underbrace{\Big\vert  {\rm Tr}\Big[\hat a^*_{\hbar_{\ell}} ( e_m )  \mathcal{F}[\varphi \bar e_m](\hat q_j)
\  \ {\varrho}_{\hbar_{\ell}}(\xi)\Big]
-\int_{X^0}   e^{Q(\xi,u)}   \langle\alpha, e_m \rangle_{L^2} \mathcal{F}[\varphi \bar e_m](q_j) \,
 \, d{\mu}(u)\Big \vert}_{(2)} ,
\\ & +\underbrace{\Big\vert \int_{X^0}   e^{Q(\xi,u)}  \big(\sum_{m=1}^R \langle\alpha, e_m \rangle_{L^2} \mathcal{F}[\varphi \bar e_m](q_j) \,
- \langle\alpha, b(q_j) \rangle_{L^2} \Big) \, d{\mu}(u)\Big \vert}_{(3)} .
\end{aligned}\]
Using estimates as in Lemma \ref{top4lemma1} and assumptions \eqref{top4S0}-\eqref{top4S1}, one proves
\[\begin{aligned}
(1)^2 & \lesssim   \Big\Vert \langle x\rangle^{-1} \Big (e^{-2\pi i k\cdot x } \varphi-\sum_{m=1}^R   e_m  \langle e_m, e^{-2\pi i k\cdot x }  \varphi  \rangle_{L^2(\R^d_k)} \Big) \Big\Vert^2_{L^\infty(\R^d_x, L^2(\R^d_k))}
\\ & \lesssim \sup_{x\in\R^d} \sum_{m=R+1}^{\infty}   \langle x\rangle^{-1} \Big \vert \langle e_m, e^{-2\pi i k\cdot x }  \varphi  \rangle_{L^2(\R^d_k)} \Big \vert ^2.
\end{aligned}\]
So, thanks to a further localization argument in the variable $x$ combined to Dini's theorem, one concludes that $(1)$ converges to zero uniformly in $\hbar\in(0,1)$ as $R\to \infty$. Similarly, using the pointwise convergence for any $x\in\R^d$,
\[\begin{aligned}
\sum_{m=1}^R \langle\alpha, e_m \rangle_{L^2} \mathcal{F}[\varphi \bar e_m](x) \,
- \langle\alpha, e^{-2\pi i k\cdot x} \varphi \rangle_{L^2}  \underset{R\to\infty}{\rightarrow} 0,
\end{aligned}\]
and dominated convergence, one concludes that $(3)$ converges to zero as $R\to\infty$. Now, applying Lemma  \ref{chap4.corr.convergence.ca} with $ \beta:=\mathcal{F}[\varphi \bar e_m]\in \mathcal{F}(L^1(\R^d))$, we obtain that $(2)$ converges also to  zero for any fixed $R\in\N$ as $\hbar_\ell \to 0$. Hence, using an $\varepsilon/3$-argument we prove the claimed statement.
\end{proof}


\bibliographystyle{plain}
\bibliography{./biblio}

\end{document}